\documentclass[a4paper,twocolumn,notitlepage,nofootinbib,longbibliography,superscriptaddress,floatfix]{revtex4-2}

\usepackage[caption=false]{subfig}
\usepackage[normalem]{ulem}
\usepackage[colorlinks=true,citecolor=teal,linkcolor=orange,urlcolor=orange]{hyperref}
\usepackage{
    adjustbox,
    algorithm,
    algpseudocode,
    amscd,
    amsfonts,
    amsmath,
    amssymb,
    amsthm,
    bm,
    booktabs,
    csquotes,
    enumerate,
    enumitem,
    graphicx,
    IEEEtrantools,
    indentfirst,
    latexsym,
    mathrsfs,
    mathtools,
    multirow,
    nicefrac,
    placeins,
    relsize,
    soul,
    tabls,
    times,
    xcolor,
    comment
}

\newtheorem{theorem}{Theorem}
\newtheorem{definition}{Definition}
\newtheorem{lemma}[theorem]{Lemma}
\newtheorem{proposition}[theorem]{Proposition}
\newtheorem{corollary}[theorem]{Corollary}

\newtheorem{question}{Question}

\newtheorem{innercustomgeneric}{\customgenericname}
\providecommand{\customgenericname}{}
\newcommand{\newcustomtheorem}[2]{%
  \newenvironment{#1}[1]
  {%
   \renewcommand\customgenericname{#2}%
   \renewcommand\theinnercustomgeneric{##1}%
   \innercustomgeneric
  }
  {\endinnercustomgeneric}
}

\newcustomtheorem{customthm}{Theorem}
\newcustomtheorem{customlemma}{Lemma}
\newcustomtheorem{customproposition}{Proposition}

\newenvironment{remark}[1][Remark]{\begin{trivlist}
\item[\hskip \labelsep {\bfseries #1}]}{\end{trivlist}}

\newcommand*{\conj}{\ensuremath{^{*}}}
\newcommand*{\dagg}{\ensuremath{^{\dagger}}}

\DeclareMathOperator{\Tr}{tr}

\DeclareMathOperator{\sign}{sign}
\DeclareMathOperator{\Herm}{Herm}

\DeclareMathOperator{\poly}{poly}
\DeclareMathOperator{\BQP}{\mathsf{BQP}}
\DeclareMathOperator{\BPP}{\mathsf{BPP}}

\DeclareMathOperator{\SWAP}{SWAP}
\DeclareMathOperator{\FT}{FT}
\renewcommand{\exp}{\ensuremath{\mathrm{exp}}}

\newcommand{\longto}{\longrightarrow}

\providecommand{\he}{\ensuremath{\hat{e}}}

\providecommand{\to}{\ensuremath{\Tilde{o}}}

\providecommand{\tr}{\ensuremath{\Tilde{r}}}

\providecommand{\trho}{\ensuremath{\Tilde{\rho}}}

\providecommand{\tcalO}{\ensuremath{\Tilde{\mathcal{O}}}}

\providecommand{\calC}{\ensuremath{\mathcal{C}}}

\providecommand{\calH}{\ensuremath{\mathcal{H}}}

\providecommand{\calM}{\ensuremath{\mathcal{M}}}

\providecommand{\calO}{\ensuremath{\mathcal{O}}}

\providecommand{\calX}{\ensuremath{\mathcal{X}}}

\providecommand{\bbC}{\ensuremath{\mathbb{C}}}

\providecommand{\bbE}{\ensuremath{\mathbb{E}}}

\providecommand{\bbI}{\ensuremath{\mathbb{I}}}

\providecommand{\bbN}{\ensuremath{\mathbb{N}}}

\providecommand{\bbP}{\ensuremath{\mathbb{P}}}

\providecommand{\bbR}{\ensuremath{\mathbb{R}}}

\graphicspath{ {images/} }

\newcommand{\fu}{Dahlem Center for Complex Quantum Systems, Freie Universit\"{a}t Berlin, 14195 Berlin, Germany}
\newcommand{\hzb}{Helmholtz-Zentrum Berlin f{\"u}r Materialien und Energie, 14109 Berlin, Germany}
\newcommand{\hhi}{Fraunhofer Heinrich Hertz Institute, 10587 Berlin, Germany}
\newcommand{\aqa}{$\langle aQa^L\rangle $ Applied Quantum Algorithms, Universiteit Leiden}
\newcommand{\liacs}{LIACS, Universiteit Leiden, Niels Bohrweg 1, 2333 CA Leiden, Netherlands}

\newcommand{\stkout}[1]{\ifmmode\text{\st{\ensuremath{#1}}}\else\st{#1}\fi}
\newif\ifverbose

\verbosefalse

\begin{document}

\title{On the expressivity of embedding quantum kernels}
\date{\today}

\author{Elies Gil-Fuster}
\email{emgilfuster@gmail.com}
\affiliation{\fu}
\affiliation{\hhi}

\author{Jens Eisert}
\affiliation{\fu}
\affiliation{\hhi}
\affiliation{\hzb}

\author{Vedran Dunjko}
\affiliation{\aqa}
\affiliation{\liacs}

\begin{abstract}
    One of the most natural connections between quantum and classical machine learning has been established in the context of kernel methods.
    Kernel methods rely on kernels, which are inner products of feature vectors living in large feature spaces.
    Quantum kernels are typically evaluated by explicitly constructing quantum feature states and then taking their inner product, here called
    embedding quantum kernels.
    Since classical kernels are usually evaluated without using the feature vectors explicitly, we wonder how expressive embedding quantum kernels are.
    In this work, we raise the fundamental question: can all quantum kernels be expressed as the inner product of quantum feature states?
    Our first result is positive: Invoking computational universality, we find that for any kernel function there always exists a corresponding quantum feature map and an embedding quantum kernel. The more operational reading of the question is concerned with efficient constructions, however.
    In a second part, we formalize the question of universality of efficient embedding quantum kernels.
    For shift-invariant kernels, we use the technique of random Fourier features to show that they are universal within the broad class of all kernels which allow a variant of efficient Fourier sampling.
    We then extend this result to a new class of so-called composition kernels, which we show also contains projected quantum kernels introduced in recent works.
    After proving the universality of embedding quantum kernels for both shift-invariant and composition kernels, we identify the directions towards new, more exotic, and unexplored quantum kernel families, for which it still remains open whether they correspond to efficient embedding quantum kernels.
\end{abstract}

\maketitle


\section{Introduction}\label{s:intro}

    Quantum devices carry the promise of surpassing classical computers in certain computational tasks~\cite{shor1994algorithms, nielsen2000quantum, montanaro2016quantum, arute2019quantum, wu2021strong, hangleiter2022computational}.
    With machine learning playing a crucial role in predictive tasks based on training data, the question arises naturally to investigate to what extent quantum computers may assist in tackling \emph{machine 
    learning} (ML) tasks.
    Indeed, such tasks are among the potential applications foreseen for near-term and intermediate-term quantum devices~\cite{biamonte2017quantum,dunjko2018machine,carleo2019machine,schuld2021machine, cerezo2021variational, bharti2022noisy}.
    
    In the evolving field of \emph{quantum machine learning} (QML), researchers have explored the integration of quantum devices to enhance learning algorithms~\cite{schuld2017implementing, havlicek2019supervised, schuld2019quantum, benedetti2019generative, perezsalinas2020reuploading, lloyd2020quantum, hubregtsen2021training}.
    The most-studied approach to QML relies on learning models based on \emph{parametrized quantum circuits} (PQCs)~\cite{benedetti2019parametrized}, sometimes referred to as quantum neural networks.
    When considering learning tasks with classical input data, PQCs must embed data into quantum states.
    This way, PQCs are built from encoding and trainable parts, and real-valued outputs are extracted from measuring certain observables.
    Since the inception of the field, a strong parallelism has been drawn between PQC-based QML models and \emph{kernel methods}~\cite{schuld2019quantum, havlicek2019supervised, schuld2021kernels}.
    
    Kernel methods, like neural networks, have been used in ML for solving complex learning tasks.
    Yet, unlike neural networks, kernel methods reach the solution by solving a linear optimization task on a larger feature space, onto which input data is mapped.
    Consequently, the kernel approach is very well suited by our ability to map classical data onto the Hilbert space of quantum states.
    These maps are called \emph{quantum feature maps}, and they lead to \emph{quantum kernel methods}.
    Although kernel methods are more costly to implement than neural networks, they are guaranteed to produce optimal solutions.
    Much the same way, with quantum kernel methods, we are guaranteed to find better solutions than with other PQC-based models
    (where \enquote{better} means solutions which perform better on the training set;
    see Ref.~\cite{jerbi2023beyond} for a discussion on when this guarantee is not enough to ensure a learning advantage).

    For quantum kernel methods, plenty of knowledge is inherited from classical ML -- including kernel selection tools~\cite{hubregtsen2021training, altares2021automatic}, optimal solution guarantees~\cite{schuld2021kernels, jerbi2023beyond}, generalization bounds~\cite{gyurik2021structural}, and approximation protocols~\cite{landman2022classically, shin2023analyzing, Sweke2023Potential}.
    Nevertheless, there is one large difference between quantum and classical kernel methods, and namely one that affects the cornerstone of these techniques: the \emph{kernel function} (or just \emph{kernel}).
    Formally, all kernels correspond to the inner product of a pair of feature maps.
    Yet, first constructing the feature vector and second evaluating the inner product is often inefficient.
    Fortunately, many cases are known in which the inner product can be evaluated efficiently, by means other than constructing the feature map explicitly.
    For instance the Gaussian kernel is a prominent instance of this case.
    This is sometimes called \enquote{the kernel trick}, and as a result it is often the case that practitioners do not even specify the feature vectors when using kernel methods.
    In contrast to this, it is fair to say that quantum kernels hardly ever use this trick, with some exceptions~\cite{huang2021power, suzuki2022quantum}.

    \begin{figure*}[t]
        \centering
        \includegraphics{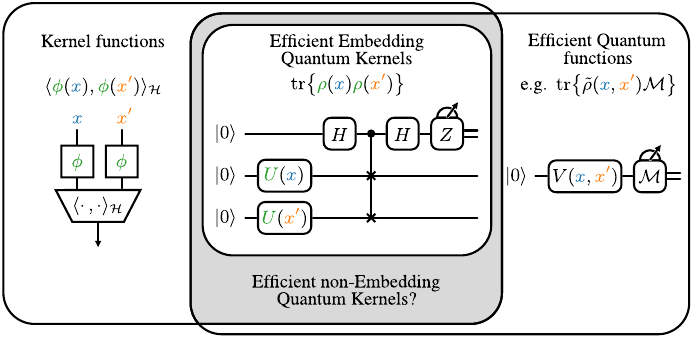}
        \caption{Illustration of the main question of this paper.
        Embedding Quantum Kernels (EQKs) have the form of an explicit inner product on the Hilbert space of quantum density matrices, which is evaluated using a quantum circuit.
        The box \enquote{Kernel functions} indicates that EQKs correspond to an inner product of feature vectors on a Hilbert space.
        The box \enquote{Efficient Quantum functions} restricts EQKs to functions that can be evaluated using a quantum computer in polynomial time, for instance these would include preparing a state-dependent state $\trho(x,x')$ and then measuring the expectation value of an observable $\calM$ on the data-dependent state.
        The box \enquote{Efficient Embedding Quantum Kernels} then clearly lives in the intersection of the two other boxes.
        The question we address here is then whether EQKs do cover the whole intersection.
        Said otherwise, can every efficient quantum kernel function be expressed as efficient EQKs?
        Or, on the contrary, do there exist efficient quantum kernels which are not expressible as efficient EQKs?
        }
        \label{fig:nEQK}
    \end{figure*}
    
    Almost all quantum kernels conceived in the literature are constructed explicitly from a quantum feature map, or \emph{quantum embedding}~\cite{mengoni2019kernel, schuld2021kernels},
    as discussed below.
    Specifically, one considers as quantum kernel $\kappa$ the inner product $\kappa(x,x')\coloneqq\langle\psi(x),\psi(x')\rangle$, where $x\mapsto \psi(x)$ is a representation of a quantum state, either a state vector or a density operator, and $\langle\cdot,\cdot\rangle$ is the appropriate inner product.
    In particular, quantum embeddings map classical data onto the Hilbert space of quantum states, or said otherwise, on the Hilbert space of quantum computations.
    We call \emph{embedding quantum kernels} (EQKs) the kernels which come from quantum embeddings.

    This difference between quantum and classical kernels raises some interesting questions, as for example:
    \emph{
        Are EQKs the whole story for quantum kernel methods?
        Can all quantum kernels be expressed as EQKs?
    }  

    In this manuscript we analyze what families of kernels are already covered by EQKs, see Fig.~\ref{fig:nEQK}.
    Our contributions are the following:
    \begin{enumerate}
        \item We show that all kernels can be seen as EQKs, thus proving their universality, when no considerations on efficiency are made.
        \item We formalize the question of expressivity of \emph{efficient} EQKs.
        We immediately provide a partial answer restricted to \emph{shift-invariant kernels}.
        We show that efficient EQKs are universal within the class of shift-invariant kernels, and we provide sufficient conditions for an EQK approximation to be produced efficiently in time.
        \item We introduce a new class of kernels, called composition kernels, containing also non-shift-invariant kernels.
        We prove that efficient EQKs are universal in the class of efficient composition kernels, from where we can show that the \emph{projected quantum kernel} from Ref.~\cite{huang2021power} can in fact also be realized as an EQK efficiently.
    \end{enumerate}
    In all, we unveil the universality of EQKs in two important function domains.

    The rest of this work is organized as follows.
    The mathematical background and relevant definitions appear in Section~\ref{s:preliminaries}.
    Related and prior work is elucidated in Section~\ref{s:related}.
    Next, we prove the universality of embedding quantum kernels and formally state Question~\ref{Q:existence_neqk} in Section~\ref{s:universality}.
    Our results on shift-invariant kernels are in Section~\ref{s:shift-invariant}, and the extension to composition kernels and the projected quantum kernel in Section~\ref{s:composition}.
    Finally, Section~\ref{s:outlook} contains a collection of questions left open.
    A summary of the manuscript and closing remarks constitute Section~\ref{s:conclusion}.


\section{Preliminaries}\label{s:preliminaries}

    In this section we fix notation and introduce the necessary bits of 
    mathematics on quantum kernel methods.

\subsection{Notation}

    For a vector $v=(v_i)_i\in\bbR^m$,  we call the $1$-norm
    (or $\ell_1$-norm)
    $\lVert v\rVert_1 = \sum_{i=1}^m \lvert v_i\rvert$ and the $2$-norm $\lVert v\rVert_2 = \sqrt{\sum_{i=1}^m v_i^2}$.
    We denote as $\ell_1^m$ the set of $m$-dimensional unit vectors with respect to the $1$-norm, and similarly $\ell_2^m$ the set of vectors that are normalized to have a unit $2$-norm.

    For a Hilbert space $\calH$, we use $\langle\cdot,\cdot\rangle_\calH$ to denote the inner product on that space.
    In the case of Euclidean spaces, we drop the subscript.
    Let $A,B\in\bbC^{m\times m}$ be square complex matrices.
    Then, we call the \emph{Hilbert-Schmidt inner product} of $A$ and $B$
    \begin{align}
        \langle A, B\rangle_{\text{HS}} &= \Tr\left\{A\dagg B\right\} = \sum_{i,j=1}^m A_{i,j}\conj B_{i,j}.
    \end{align}
    For Hermitian matrices, we have $A=A\dagg$, and so the HS inner product becomes just
    \begin{align}
        \langle A, B\rangle_{\text{HS}} &= \Tr\left\{AB\right\} = \sum_{i,j=1}^m A_{j,i} B_{i,j}.
    \end{align}
    The \emph{Frobenius norm} $\lVert\cdot\rVert_F$ of a matrix $A$ is the square root of the sum of the magnitude square of all its entries, which is also equal to the root of the Hilbert-Schmidt inner product of the matrix with itself, as
    \begin{align}
        \lVert A\rVert_F^2 &= 
        {\langle A, A\rangle_{\text{HS}}} = 
        {\sum_{i,j} \lvert A_{i,j}\rvert^2}.
    \end{align}
    In what follows, we call $\calX\subseteq\bbR^d$ a $d$-dimensional compact subset of the reals.
    We reserve $n$ to denote qubit numbers.
    As we explain below, we use $k$ to refer to arbitrary kernel functions, while $\kappa$ is used exclusively for embedding quantum kernels.

    When talking about efficient and inefficient approximations, we consider sequences of functions $\{k_s\}_{s\in\bbN}$.
    We refer to $s$ as the \emph{scale parameter}.
    Scale parameters can correspond to different qualities of the sequence, as for example the dimension of the input data, or the number of qubits involved in evaluating a function.
    Efficiency then means at most polynomial scaling in $s$, and inefficiency means at least exponential scaling in $s$, hence the name scale parameter.
 
\subsection{Kernel methods}\label{ss:kernels}

    Kernel methods solve ML tasks as linear optimization problems on large feature spaces, sometimes implicitly.
    The connection between the input data and the feature space comes from the use of a kernel function.
    In this work we do not busy ourselves with how the solution is found, but rather we focus on our ability to evaluate kernel functions, which are defined as follows.
    
    \begin{definition}[\bf Kernel function]\label{def:kernel}
        A kernel function $k\colon\calX\times\calX\to\bbR$ is a map from pairs of inputs on the reals fulfilling two properties:
        \begin{enumerate}
            \item Symmetry under exchange: $k(x,x')=k(x',x)$ for every $x,x'\in\calX$.
            \item Positive semi-definite (PSD): for any integer $m\in\bbN$, for any sets $\{x_1, \ldots, x_m\}\subseteq\calX$ and $\{a_1, \ldots, a_m\}\subseteq\bbR$ of size $m$, it holds
                \begin{align}
                    \sum_{i,j=1}^m a_i a_j k(x_i,x_j)\geq0.
                \end{align}
                This is equivalent to saying that the \emph{Gram matrix} $K\coloneqq[k(x_i,x_j)]_{i,j=1}^m$ is positive semi-definite for any $m$ and any $\{x_i\}_{i=1}^m$.
        \end{enumerate}
    \end{definition}

    Other standard definitions exclude the PSD property.
    In this work we do not study indefinite, non-PSD kernels, although the topic is certainly of interest.
    The common optimization algorithms used in kernel methods (SVM, KRR) require kernels to be PSD.
    Even though we said we do not deal with the optimization part in this manuscript, we study only the kernels that would be used in a (Q)ML context.

    Symmetry and PSD are properties usually linked to inner products, which partly justifies our definition of a Gram matrix as the one built from evaluating the kernel function on pairs of inputs.
    Indeed, by Mercer's Theorem (detailed in Appendix~A) for any kernel function $k$, there exists a Hilbert space $\calH$ and a \emph{feature map} $\phi\colon\calX\to\calH$ such that, for every pair of inputs $x,x'\in\calX$, it holds that evaluating the kernel is equivalent to computing the $\calH$-inner product of pairs of \emph{feature vectors} $k(x,x') = \langle\phi(x),\phi(x')\rangle_\calH$.
    We say every kernel $k$ has an associated feature map $\phi$, which turns each datum $x$ into a feature vector $\phi(x)$, living in a \emph{feature space} $\calH$.
    
    One notable remark is that different kernel functions have different feature maps and feature spaces associated to them.
    Each learning task comes with a specific data distribution.
    The ultimate goal is to, given the data distribution, find a map onto a feature space where the problem becomes solvable by a linear model.
    The \emph{model selection} challenge for kernel methods is to find a kernel function whose associated feature map and feature space fulfill this linear separation condition.
    Intuitively, in a classification task, we would want data from the same class to be mapped to the same corner of Hilbert space, and data from different classes to be mapped far away from one another.
    Then our project can be framed within the \emph{quantum kernel selection} problem, we want to identify previously unexplored classes of quantum kernels.

\subsection{Embedding quantum kernels}\label{ss:EQKs}

    In the following, we introduce the relevant concepts in quantum kernel methods.
    Our presentation largely follows the lines of Ref.~\cite{schuld2021kernels}.
    We use the name \enquote{embedding quantum kernels} following Ref.~\cite{hubregtsen2021training}; many other works use simply \enquote{quantum kernels} when referring to the same concept.
    Hopefully our motives are clear by now.
    
    Learning models based on PQCs are the workhorses in much of today's approaches to QML.
    With some notable exceptions, these PQC-based models comprise a two-step process:
    \begin{enumerate}
        \item Prepare a data-dependent quantum state. For every $x\in\calX$, produce $\lvert \phi(x)\rangle = U(x)\lvert0\rangle$ with some data-embedding unitary $U(x)$.
        \item Measure the expected value of a variationally-tunable observable on the data-dependent state. Given a parametrized observable $\calM(\vartheta)$, evaluate
        \begin{align}
            \langle\calM(\vartheta)\rangle_{\phi(x)} &= \Tr\left\{\lvert \phi(x)\rangle\!\langle \phi(x)\rvert \calM(\vartheta)\right\}.
        \end{align}
    \end{enumerate}
    For instance, for binary classification, one would then consider labeling functions like $h(x) = \sign\left(\langle\calM(\vartheta)\rangle_{\phi(x)} + b\right)$ and optimize the variational parameters $\vartheta$ and $b$.

    Examples of how $U(x)$ could look like in practice include \emph{amplitude encoding}~\cite{schuld2021kernels}, different \emph{data re-uploading} encoding strategies~\cite{schuld2021fourier, perezsalinas2020reuploading,caro2021encoding}, or the \emph{IQP ansatz}~\cite{havlicek2019supervised}.
    In turn, variationally-tunable observables are usually realized as a fixed \enquote{easy} observable (like a single Pauli operator), preceded by a few layers of brickwork-like $2$-local trainable gates.
    In our current presentation, we do not restrict a particular form for $U(x)$ or $\calM(\vartheta)$, but rather allow for any general form possible, as long as it can be implemented on a quantum computer.
    
    This approach can also be understood through the lens of kernel methods.
    In this case, the feature space is explicitly chosen to be the Hilbert space of Hermitian matrices (of which density matrices are a subset, so also quantum states).
    Given a data-dependent state preparation $x\mapsto\lvert \phi(x)\rangle$, one could choose to call it a \emph{quantum feature map} $\rho$ and promote it to quantum density operators, or quantum states, as
    \begin{align}
        \rho\colon\calX \to & 
        \Herm , \\
        x \mapsto & \rho(x) = \lvert \phi(x)\rangle\!\langle \phi(x)\rvert.
    \end{align}
    Another fitting name would be \emph{quantum feature state}.
    In a slightly more general view, we call \emph{data embedding} any map from classical data onto quantum density matrices of fixed dimension $\rho\colon\calX\to\Herm(2^n)$ for $n$-qubit systems.
    With this we abandon the need for a unitary gate applied to the $\lvert0\rangle$ state vector.

    Together with the Hilbert-Schmidt inner product, quantum feature maps give rise to an important family of kernel functions:
    
    \begin{definition}[Embedding quantum kernel (EQK)]\label{def:EQK}
        Given a data embedding $x\mapsto\rho(x)$ used to encode classical data $x$ in a quantum density operator $\rho(x)$, we call \emph{embedding quantum kernel (EQK)} $\kappa_\rho$ the Hilbert-Schmidt inner product of pairs of quantum feature vectors:
        \begin{align}
            \kappa_\rho \colon \calX\times\calX \to &\bbR ,\\
            (x,x') \mapsto & \kappa_\rho(x,x')  \coloneqq \Tr\left\{\rho(x)\rho(x')\right\}.
        \end{align}
        Fig.~\ref{fig:EQK} illustrates this construction.
    \end{definition}

    \begin{figure}
        \centering
        \includegraphics{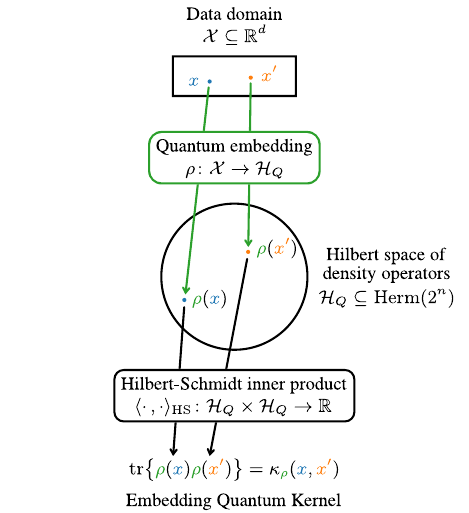}
        \caption{Schematic of the different ingredients that form embedding quantum kernels.
            The data input is mapped onto the \enquote{quantum feature space} of \emph{quantum density operators} via a \emph{quantum embedding}.
            There the \emph{Embedding Quantum Kernel} is defined as the \emph{Hilbert-Schmidt inner product} of pairs of quantum features.
        }
        \label{fig:EQK}
    \end{figure}

    Since $\kappa_\rho$ is defined explicitly as an inner product for any $\rho$, it follows that it is a PSD, symmetric function.
    For ease of notation, we will write just $\kappa$ whenever $\rho$ is unimportant or clear from context.
    There are, from the outset, a few good reasons to consider EQKs as QML models:
    \begin{enumerate}
        \item It is 
        possible to construct EQK families that are not classically estimable unless $\BQP=\BPP$, thus opening the door to quantum advantages \cite{Liu2021discretelog}.
        \item A core necessity for successful kernel methods is to map complex data to high-dimensional spaces where feature vectors become linearly separable.
            The Hilbert space of quantum states thus becomes a prime candidate due to its exponential dimension, together with our ability to estimate inner products efficiently\footnote{However, high dimensions are related to overfitting, so as usual there is a trade-off.}.
        \item Since EQKs are explicitly defined from the data embeddings, we are free to design embeddings with specific desired properties.
    \end{enumerate}
    And, a priori, other than the shot noise\footnote{It is known that the finite approximation to the expectation values can break down PSD.}, EQKs do not add drawbacks to the list of general issues for kernel methods.
    So, EQKs are a well-founded family of quantum kernel functions.
    Nevertheless, their reliance on a specific data embedding could be a limiting factor \emph{a priori}.
    
    In the case of classical computations we are not accustomed to think in terms of the Hilbert space of the computation (even though it does exist).
    There are numerous examples of kernels which have been designed without focusing on the feature map.
    The most prominent example of a kernel function used in ML is the \emph{radial basis function (RBF) Gaussian kernel} (or just Gaussian kernel, for short),
    given by
    \begin{align}
        k_\sigma(x,x') &= \exp\left(-\frac{\lVert x-x'\rVert^2}{2\sigma^2}\right).
    \end{align}
    On the one hand, we know the Gaussian kernel is a PSD function by virtue of Bochner's theorem (stated in Section~\ref{s:shift-invariant} as Theorem~\ref{thm:Bochner}), which explains it via the \emph{Fourier transform} (FT).
    On the other hand, though, one can prove that the Gaussian kernel corresponds to an inner product in the Hilbert space of all monomials of the components of $x$, from where it follows that the Gaussian kernel can learn any continuous function on a compact domain Ref.~\cite{micchelli2006universal}, provided enough data is given.
    The feature map corresponding to the Gaussian kernel is infinite-dimensional.
    This fact alone makes the Gaussian kernel not immediately identifiable as a reasonable EQK, which motivates the question:
    \emph{can all quantum kernels be expressed as embedding quantum kernels?}
    
    In this section, we have  introduced kernel functions, embedding quantum kernels, and the question whether EQKs are all  there is to quantum kernels.
    In the following sections, we first see that EQKs are universal, then we formalize a question about the expressivity of efficient EQKs, and finally we answer the question by proving universality of EQKs in two restricted kernel families.


\section{Related work}\label{s:related}

    An introduction to quantum kernel methods can be found in the note~\cite{schuld2021kernels} and in the review~\cite{mengoni2019kernel}.
    We refrain from including here a full compendium of all quantum kernel works to date, as good references can already be found in Ref.~\cite{bharti2022noisy}.
    Instead, we make an informed selection of papers that bear relation with the object of our study: which quantum kernels are \emph{embedding quantum kernels} (EQKs).

    Kernel methods use different optimization algorithms depending on each task, the most prominent examples being \emph{support vector machine} and \emph{kernel ridge regression}.
    The early years of QML owed much of their activity to the HHL algorithm~\cite{harrow2009quantum}.
    The quantum speed-up for linear algebra tasks was leveraged to propose a \emph{quantum support vector machine} (qSVM) algorithm~\cite{rebentrost2014quantum}.
    The qSVM algorithm is listed among the first steps of QML historically.
    Since qSVM is a quantum application for kernel methods, it is no wonder that the term quantum kernel methods was introduced for concepts around it.
    This, however, is not what we mean by quantum kernel methods in this work.
    Instead, we occupy ourselves with kernel methods where the kernel function itself requires a quantum computer to be evaluated, independently of the nature of the optimization algorithm.
    The optimization step comes only after the kernel function has been evaluated on the training data, so the object of study of qSVM is not the same as in this manuscript.
    We study the expressivity of a known kind of kernel functions, and not how to speed up the optimization of otherwise classical algorithms.

    Among the first references mentioning the evaluation of kernel functions with quantum computers is Ref.~\cite{schuld2017implementing}, where the differences between quantum kernels and qSVM have been made explicit. References~\cite{schuld2019quantum,havlicek2019supervised} have showcased the link between quantum kernels and PQCs and have demonstrated their implementation experimentally.
    In them, the authors have mentioned the parallelisms between quantum feature maps and the kernel trick.
    In particular, Ref.~\cite{schuld2019quantum} has coined the distinction between an \emph{implicit quantum model} (quantum kernel method), and an \emph{explicit quantum model} (PQC model with gradient-based optimization).
    Implicit models are in this way an analogous name for EQKs, where emphasis is made on the distinction from other PQC-based models.

    Quantum kernels have enjoyed increasing attention since they have been used to prove an advantage of QML over classical ML in Ref.~\cite{Liu2021discretelog}.
    The authors have morphed the discrete logarithm problem into a learning task and then used a quantum kernel to solve it efficiently, which cannot be done classically according to well-established cryptographic assumptions.
    In that, the approach taken is similar to that of
    Ref.~\cite{sweke2021quantum} (showing a quantum advantage in
    distribution learning).
    As such, this has been among the first demonstrations of quantum advantage in ML, albeit in an artificially constructed learning task.
    Importantly, the quantum kernel used in this work has explicitly been constructed from a quantum embedding, 
    so it is an EQK in the sense of this work.

    One important difference between EQKs and other PQC-based approaches is that, for EQKs, the only design choice is the data embedding itself.
    Ref.~\cite{lloyd2020quantum} has wondered how to construct optimal quantum feature maps using measurement theory.
    The work has proposed constructing embeddings specific to learning tasks, which resonates with the idea that some feature maps are better than others for practically relevant problems.
    Ref.~\cite{hubregtsen2021training}, where the term EQKs has been  introduced for the first time, has presented the possibility of optimizing the feature map variationally, drawing bridges to ideas from \emph{data re-uploading}~\cite{perezsalinas2020reuploading,schuld2021fourier}.
    Other than the trainable kernels of Ref.~\cite{hubregtsen2021training}, the data re-uploading framework did not fit  with the established quantum kernel picture up until this point in time.
    
    In Ref.~\cite{jerbi2023beyond}, the differences between explicit models, implicit models, and re-uploading models have been analyzed.
    On the one hand, the authors have found that the optimality of kernel methods might not be optimal enough, since they show a learning task in which kernels perform much worse than explicit models when evaluated on data outside the training set.
    On the other hand, a rewriting algorithm has been devised to convert re-uploading circuits into equivalent encoding-first circuits, so re-uploading models are explicit models.
    By construction, each explicit model has a quantum kernel associated to it.
    So, for the first time, using the rewriting via explicit models, Ref.~\cite{jerbi2023beyond} made data re-uploading models fit in the quantum kernel framework.
    This reference takes the explicit versus implicit distinction from~\cite{schuld2019quantum}, which means that it also only considers EQKs when it comes to kernels.

    Some examples of different kinds of quantum kernels can be found in Refs.~\cite{huang2021power, suzuki2022quantum}, which upon first glance do not resemble the previous proposals.
    In Ref.~\cite{huang2021power} a type of kernel functions based on the classical shadows protocol has been proposed, under the name of \emph{projected quantum kernel}.
    These functions have been analyzed in the context of a quantum-classical learning separation, now considering real-world data, as opposed to the results of Ref.~\cite{Liu2021discretelog}.
    Even though the projected quantum kernel uses an explicit feature map which requires a quantum computer, the feature vectors are only polynomial in size, and they are stored in classical memory.
    Next, the kernel function is the Gaussian kernel evaluated on pairs of feature vectors, and not their Euclidean or Hilbert-Schmidt inner product.
    These two differences set the projected quantum kernel apart from the rest.
    In turn, the authors of Ref.~\cite{suzuki2022quantum} set out to address a looming issue for EQKs called \emph{exponential kernel concentration}~\cite{thanasilp2022exponential}, or \emph{vanishing similarity}, which is tantamount to the \emph{barren plateau problem}~\cite{mcclean2018barren} that could arise also for kernels.
    They show that the new construction, called \enquote{anti-symmetric logarithmic derivative quantum Fisher kernel}, 
    does not suffer from the vanishing similarity problem.
    Here, the classical input data is mapped onto an exponentially large feature space.
    Given a PQC, classical inputs are mapped to a long array, with as many entries as trainable parameters in the PQC.
    Each entry in this array is the product of the unitary matrix implemented by the PQC and its derivatives with respect to each variational parameter.
    Interestingly, this kernel can be seen as the Euclidean inner product of a flattened vector of unitaries, with a metric induced by the initial quantum state.
    In this way, classical data is not mapped onto the Hilbert space of quantum states, but the inner product used is the same as in regular EQKs.
    This realization paves the way to expressing these kernels as EQKs.

    Recent efforts in de-quantizing PQC-based models via \emph{classical surrogates}~\cite{schreiber2022classical} have touched upon quantum kernels.
    In Refs.~\cite{landman2022classically, shin2023analyzing, Sweke2023Potential}, the authors propose using a classical kernel-approximation protocol based on \emph{random Fourier features} (RFF) to furnish classical learning models capable of approximating the performance of PQC-based architectures.
    The techniques used in these works are not only very promising for de-quantization of PQC-based learning models, but also they are relevant to our discussion.
    Below, we also use the RFF approach, albeit in quite a different way, as it is not our goal to find classical approximations of quantum functions, but rather quantum approximations of other quantum functions.
    The goal of the study in Refs.~\cite{landman2022classically, shin2023analyzing, Sweke2023Potential} is to, given a PQC-based model, construct a classical kernel model with guarantees that they are similarly powerful.
    In this way, the input to the algorithms is a PQC (either encoding-first, or data re-uploading), and the output is a classical kernel.
    Conversely, in our algorithms, the input is a kernel function, and the output is an EQK-based approximation of the same function.

    In Ref.~\cite{yamasaki2020learning}, we find an earlier study of RFFs in a QML scenario.
    The authors focus on a more advanced version of RFFs, called \emph{optimized Fourier Features}, which involves sampling from a data-dependent distribution.
    In the classical literature, it was found that sampling from this distribution could be \enquote{hard in practice}~\cite{bach2017equivalence}, so the authors of Ref.~\cite{yamasaki2020learning} set out to propose an efficient quantum algorithm for sampling.
    This way, a quantum algorithm is proposed to speed up a training algorithm for a classical learning architecture, similar to the case of qSVM~\cite{rebentrost2014quantum}.

    Similarly, Ref.~\cite{nakaji2022deterministic} considered combining notions from EQKs, the projected quantum kernel, and their respective approximations with the RFF algorithm.
    The authors explored the options of combining distance-based and inner-product-based kernels in an attempt to address the vanishing similarity problem.
    In turn, this paper tackles the suboptimal scaling of kernel methods by making use of RFFs.
    Some of the techniques introduced in this reference are very much in line with the composition kernels we introduce in Section~\ref{s:composition}.
    The main similarity between our work and Ref.~\cite{nakaji2022deterministic} is the identification of feature maps for different quantum kernels, including the use of RFFs and the study of the projected quantum kernel.
    The main difference is the perspective in which these objects are studied.
    In Ref.~\cite{nakaji2022deterministic} the authors look for kernel constructions that have advantageous properties when it comes to solving learning tasks.
    Conversely, in this work we want to establish what are the ultimate limits in expressivity of using kernel functions based on quantum feature maps.


\section{The universality of quantum feature maps}\label{s:universality}  
    
    In the previous section we saw how we can define quantum kernel functions explicitly from a given data embedding.
    This section, together with the following two sections, contains the main results of this manuscript.
    The statements shift around different notions of efficiency and different classes of kernel functions.
    In Fig.~\ref{fig:venn_diagram} we provide a small sketch of how each of our results relates to one another from a zoomed-out perspective.
    The leitmotiv is that we can always restrict further what restrictions we are satisfied with when concocting EQK-based approximations of kernels.
    We find all kernels can be approximated as EQKs if our only restriction is to use finitely many quantum resources.
    But then, we specialize the search for kernels which can be approximated efficiently as EQKs, with a distinction between space-efficiency (number of qubits required) and total run-time efficiency.
    It should also be said that when we talk about time efficiency we always consider \enquote{quantum time}, so we assume we are always allowed access to a quantum computer.
    This way, the analysis from now on departs from the usual quantum-classic separation mindset.
    
    From the outset, two basic results combined, namely Mercer's feature space construction (elaborated in Appendix~A), together with the universality of quantum circuits, already certify the possibility of all kernel functions to be realized as EQKs.
    If we demanded mathematical equality, Mercer's construction could require infinite-dimensional Hilbert spaces, and thus quantum computers with \emph{infinitely many qubits}.
    Instead, for practical purposes, from now on we do not talk about \enquote{evaluating} functions in one or another way, but rather about \enquote{approximating} them to some precision.
    With this, Theorem~\ref{thm:approx_universality} confirms that we can always approximate any kernel function as an EQK to arbitrary precision \emph{with finitely many resources}.
    
    For this universality statement, we allow for extra multiplicative and additive factors.
    Instead of talking about $\varepsilon$-approximation as $\lvert k(x,x')-\Tr\{\rho_n(x)\rho_n(x')\}\rvert<\varepsilon$, we consider $ \left\lvert k(x,x') - 2^n \Tr\{\rho_n(x)\rho_n(x')\} + 1\right\rvert < \varepsilon$.
    These extra factors come from Lemma~\ref{l:innerprodqs}, introduced right below, and they do not represent an obstacle against universality.
    
    \begin{theorem}[Approximate universality of finite-dimensional quantum feature maps]\label{thm:approx_universality}
        Let $k\colon\calX\times\calX\to\bbR$ be a kernel function.
        Then, for any $\varepsilon>0$ there exists $n\in\bbN$ and a data embedding $\rho_n$ onto the Hilbert space of quantum states of $n$ qubits such that
        \begin{align}\label{eq:universality}
            \left\lvert k(x,x') - 2^n \Tr\{\rho_n(x)\rho_n(x')\} + 1\right\rvert < \varepsilon
        \end{align}
        for almost all $x,x'\in\calX$.
    \end{theorem}

    The statement that Eq.~\eqref{eq:universality} holds for \emph{almost all $x,x'\in\calX$} comes from measure theory, and it is a synonym of \enquote{except in sets of measure $0$}, or equivalently \enquote{with probability $1$}.
    That means that although there might exist individual adversarial instances of $x,x'\in\calX$ for which the inequality does not hold, these \enquote{bad} instances are sparse enough that the event of drawing them from the relevant probability distribution has associated probability $0$.
    
    Theorem~\ref{thm:approx_universality} says every kernel function can be approximated as an EQK up to a multiplicative and an additive factor using finitely-many qubits.
    Before we give the theorem proof, we introduce the useful Algorithm~\ref{alg:QEPIP} to map classical vectors to quantum states that we can then use to evaluate Euclidean inner products as EQKs.
    Then, Lemma~\ref{l:mappingqs} contains the correctness statement and runtime complexity of Algorithm~\ref{alg:QEPIP}, and Lemma~\ref{l:innerprodqs} shows the relation between the Euclidean inner product of the encoded real vectors and the Hilbert-Schmidt inner product of the encoding quantum states.
    
    \begin{algorithm}[H]
        \caption{Classical to quantum embedding, $\mathtt{C2QE}$}
        \label{alg:QEPIP}
        \begin{algorithmic}[1]
            \Require A $1$-norm unit vector $r\in\ell_1^d$.
            \Ensure A quantum state $\rho_r\propto\bbI+\sum_{i=1}^d r_iP_i$. \Comment See Lemma~\ref{l:mappingqs}.
            \State Set $n=\lceil\log_4(d+1)\rceil$.
            \State Pad $r$ with $0$s until it has length $4^n-1$.
            \State\label{algline:sample_i} Draw $i\in\{1,\ldots,4^n-1\}$ with probability $\lvert r_i\rvert$.
            \State Prepare $\rho_i = \frac{1}{2^n}\left(\bbI+\sign(r_i)P_i\right)$.
            \State\Return $\rho_i$.
        \end{algorithmic}
    \end{algorithm}

    Notice that the output of Algorithm~\ref{alg:QEPIP}, $\frac{1}{2^n}\left(\bbI \pm P\right)$, is a single (pure) eigenstate of a Pauli operator $P$ with eigenvalue $\pm1$.
    Nevertheless, as Line~\ref{algline:sample_i} involves drawing an individual index $i\in\{1,\ldots,4^n-1\}$, we see Algorithm~\ref{alg:QEPIP} as a random algorithm, which prepares a mixed state as a classical mixture of pure states.
    
    \begin{lemma}[Correctness and runtime of Algorithm~\ref{alg:QEPIP}]\label{l:mappingqs}
        Let $r\in\ell_1^d\subseteq\bbR^d$ be a unit vector with respect to the $1$-norm, $\lVert r\rVert_1=1$.
        Take $n=\lceil\log_4(d+1)\rceil$ and pad $r$ with $0$s until it has length $4^n-1$.
        Let $(P_i)_{i=1}^{4^n-1}$ be the set of all Pauli matrices on $n$ qubits without the identity.
        Then Algorithm~\ref{alg:QEPIP} prepares the following state as a classical mixture
        \begin{align}
            \rho_{(\cdot)}\colon\ell_1^d\to & \operatorname{Herm}(2^n) ,\\
            r \mapsto &\rho_r = \frac{\bbI + \sum_{i=1}^{4^n-1} r_i P_i}{2^n}.
        \end{align}
        The total runtime complexity $t$ of Algorithm~\ref{alg:QEPIP} fulfills $t\in\calO(\poly(d))$.
    \end{lemma}

    \begin{lemma}[Euclidean inner products]
    \label{l:innerprodqs}
        Let $r,r'\in\bbR^d$ be unit vectors with respect to the $1$-norm $\lVert r^{(\prime)}\Vert_1=1$.
        Then, for $\rho_r,\rho_{r'}$ as produced in Algorithm~\ref{alg:QEPIP}, the following identity holds
         \begin{align}
            \langle r,r'\rangle &= 2^n \Tr\{\rho_{r}\rho_{r'}\}-1.
        \end{align}
    \end{lemma}
    \begin{proof}[Proof of Lemmas~\ref{l:mappingqs} and~\ref{l:innerprodqs}]
        The proofs are presented in Appendix~B.
    \end{proof}

    \begin{remark}{(Prefactors)\textbf{.}}
        The $2^n$ multiplicative factor is of no concern, since $n\in\calO(\log(d))$ and we are interested in methods that are allowed to scale polynomially in $d$.
        In general $\rho_r$ will be a mixed quantum state which can be efficiently prepared.
        Also, the map is injective, but not surjective.
    \end{remark}

    We are now in the position to prove Theorem~\ref{thm:approx_universality}.
    We note that the first step in the proof is to furbish a \enquote{classical} feature map, and in a second step we turn the resulting feature vector into a quantum state.
    This way, Theorem~\ref{thm:approx_universality} makes no pretense at quantum advantage, but rather establishes the ultimate expressivity of embedding quantum kernels.
    We also point out that not imposing space efficiency breaks the division between quantum and classical computing.
    Every quantum circuit running in finite time can be fully simulated by a classical circuit also running in finite time.
    This way, the universality statement coming from Theorem~\ref{thm:approx_universality} is not unique to quantum kernels, since formally at this point there is no distinction between classical and quantum functions.
    
    \begin{proof}[Proof of Theorem~\ref{thm:approx_universality}]

        We prove this statement directly using a corollary of Mercer's theorem and the universality of quantum computing.
        First, we invoke Corollary~A.3 in Appendix~A, which ensures the existence of a finite-dimensional feature map $\Phi_m:\calX\to\bbR^m$ for which it holds that
        \begin{align}
            \left| k(x,x') - \langle\Phi_m(x),\Phi_m(x')\rangle\right| &< \varepsilon.
        \end{align}        
        Without loss of generality, we assume $\lVert\Phi_m(x)\rVert_1=1$ for all $x\in\calX$.
        Now, we can prepare the quantum state $\rho_{\Phi_m(x)}$, which will use $\lceil\log_4(m+1)\rceil$ many qubits.
        By preparing two such states, one for $\Phi_m(x)$ and one for $\Phi_m(x')$, we can compute their inner product as the Hilbert-Schmidt inner product of the quantum states as in Lemma~\ref{l:innerprodqs}, to get
        \begin{align}
            \langle \Phi_m(x),\Phi_m(x')\rangle &= 2^n \Tr\{\rho_{\Phi_m(x)}\rho_{\Phi_m(x')}\}-1.
        \end{align}
        For reference, notice $\Tr\{\rho_{\Phi_m(x)}\rho_{\Phi_m(x')}\}$ could be computed using the $\SWAP$ test, 
        to additive precision in the number of shots.
        With this, we can ensure good approximation to polynomial additive precision efficiently
        \begin{align}
            \lvert k(x,x') - 2^n \Tr\{\rho_{\Phi_m(x)}\rho_{\Phi_m(x')}\} +1\rvert < \varepsilon,
        \end{align}
        for almost every $x,x'\in\calX$.
        This completes the proof.
    \end{proof}
    
    Notice this is an conceptually motivated existence result only.
    Noteworthy is also that no claims are made about practicality neither in Theorem~\ref{thm:approx_universality} nor in Algorithm~\ref{alg:QEPIP}.
    These results only aim at establishing the ultimate existence of EQKs for any kernel, and they are not claims that computing these kernels as EQKs should lead to any form of advantage over the inner product of the classical feature vectors.
    
    The only statement is that there exists an EQK using a finite number of qubits, but it does not get into how quickly the number of qubits will grow for increasingly computational complicated kernel function $k$, or increased required precision $\varepsilon>0$.
    The number of qubits $n$ will depend on some properties of the kernel $k$ and the approximation error $\varepsilon$, and if for example we had exponential scaling of the required number of qubits on some of these quantities, then Theorem~\ref{thm:approx_universality} would bring no practical application.
    Similarly, one should also consider the time it would take to find such an EQK approximation, independently of the memory and run-time requirements of preparing the feature vectors and computing their inner product.

    Let us take the anti-symmetric logarithmic derivative quantum Fisher kernel of Ref.~\cite{suzuki2022quantum} as an example.
    Upon first inspection, evaluating that kernel does not seem to rely in the usual quantum feature map and Hilbert-Schmidt inner product combination.
    Yet, a feature map can be identified by rewriting some of the variables involved as a vector of exponential length.
    That means the scaling in the number of qubits required to encode the feature vector is at worst polynomial, according to Lemma~\ref{l:innerprodqs}.
    The normalization requirement from Lemma~\ref{l:innerprodqs} could still prevent the feature map to be encoded using the construction we propose, but for the sake of illustration let us assume normalization is not a problem.
    Then, we have found a way of realizing the same kernel as an EQK.
    In this case, even though the scaling of the qubit number is no more than polynomial, the scaling in total run-time of the EQK approximation would still be at worst exponential.
    
    The message remains that, although all kernel functions can be realized as EQKs, there could still exist kernel functions which cannot be realized as EQKs \emph{efficiently}.
    Pointing back to Fig.~\ref{fig:nEQK}, this explains why we added the word \enquote{efficient} to the sets of quantum functions and EQKs.
    In order to talk about efficiency we need to replace individual functions $k$ by function sequences $\{k_s\}_{s\in\bbN}$, where $s$ is the scale parameter.
    When we refer to an efficient $\varepsilon$-approximation, we refer to an algorithm making use of up to $\poly(s,1/\varepsilon)$ resources, for each $k_s$ in the sequence.
    We now formally present the question we aim at answering.

    \begin{question}[Expressivity of efficient EQKs]\label{Q:existence_neqk}
        Let $\{k_s\}_{s\in\bbN}$ be a sequence of kernel functions, let $\varepsilon>0$ be a precision parameter, and consider the properties:
        \begin{enumerate}
            \item \emph{Quantum efficiency:} There is an algorithm that takes a specification of $k_s$ as input and produces an $\varepsilon$-approximation of $k_s$ with a quantum computer efficiently in $s$ and $1/\varepsilon$.
            \item \emph{Embedding-quantum efficiency:} There is an algorithm that takes a specification of $k_s$ as input and produces an $\varepsilon$-approximation of $k_s$ as an EQK efficiently in $s$ and $1/\varepsilon$.
            \item \emph{Classical inefficiency:} Any algorithm that takes a specification of $k_s$ as input and produces an $\varepsilon$-approximation of $k_s$ with a classical computer must be inefficient in either $s$ or $1/\varepsilon$.
        \end{enumerate}
        Then, assuming $\{k_s\}_s$ fulfills classical inefficiency, does quantum efficiency imply embedding-quantum efficiency?
    \end{question}

    The question above contains a few moving pieces which still need to be made fully precise, as for instance: the meaning of the scaling parameter $s$, the sequence of domains $\calX_s$ from where each $k_s$ takes its input, any restrictions on the functions $k_s$, in what form must the functions $k_s$ be specified, the choice of notion of $\varepsilon$-approximation, and the choice between space and time efficiency.
    These are left open on purpose to admit for diverse approaches to studying the question.

    We could have required a stronger sense of inefficiency, namely that there cannot exist an efficient and uniform construction for $\varepsilon$-approximating $k_s$ with classical computers.
    Instead, we judge it enough to require that, even if such an efficient approximation could exist, it would be impossible to find efficiently, conditional on $\BQP\nsubseteq\mathsf{P/poly}$.
    We added classical inefficiency because otherwise, in principle, one could have all: quantum efficiency, embedding-quantum efficiency, and classical efficiency, in which case the result would not be interesting for QML.
    An interesting question is then to imagine how classical kernel functions could also be embedding classical kernels, in the sense of data being mapped onto the Hilbert space of classical computation.

    In later sections, we fix all of these in our answers to Question~\ref{Q:existence_neqk}.
    For instance, $s$ for us is the dimension of the domains $\calX_s$ from where data is taken, the kernel functions can be specified either as black-boxes or as the description of circuits, we will take infinity-norm approximation almost everywhere, and we alternate between qubit number efficiency and run-time efficiency.

    As an aside, Question~\ref{Q:existence_neqk} also invites research in the existence of quantum kernels beyond EQKs, so efficient quantum kernels which do not admit efficient EQK-based approximations.
    At the moment of writing, the authors are not aware of any concrete example of a quantum efficient kernel function which is provably not embedding-quantum efficient.
    Moreover, we do not know any quantum efficient kernel for which we do not have an explicit way of constructing the efficient quantum embedding.
    The outlook in Section~\ref{s:outlook}, and Appendices~A and~E contain several directions in which we expect candidate kernels beyond efficient EQKs to appear.
    It should be said that searching for quantum kernels beyond EQKs is slightly contradictory to the foundational philosophy of QML in the beginnings, which actively sought to express everything in terms of inner products in the Hilbert space of density operators.
    The foundational works~\cite{schuld2017implementing,schuld2019quantum,havlicek2019supervised} have motivated the use of embedding quantum kernels precisely because the inner product could be taken directly efficiently.
    Nevertheless, we point out the possibility of alternative constructions, keeping focus on methods which could still harbour quantum advantages.

    \begin{figure}[t]
        \centering
        \includegraphics{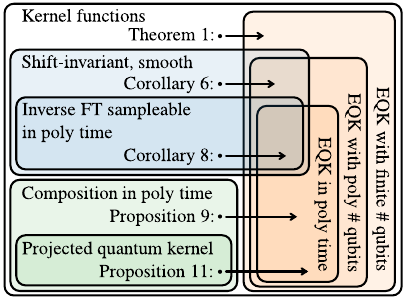}
        \caption{Venn diagram with set relations of different classes of kernels and quantum kernels.
            Each of the arrows represents a reduction found in this manuscript, they should be read as \enquote{for any element of the first set, there exists an element of the second set which is a good approximation.}
            In the case of Theorem~\ref{thm:approx_universality}, elements are individual functions.
            In every other case, elements are sequences of kernel functions, for which the notions of efficiency make sense.
            In summary, we find that efficient 
            \emph{embedding quantum kernels} (EQK) can approximate two important classes of kernels: shift-invariant and composition kernels.
        }
        \label{fig:venn_diagram}
    \end{figure}


\section{The universality of efficient shift-invariant embedding quantum kernels}\label{s:shift-invariant}

    \begin{figure*}[t]
        \centering
        \includegraphics{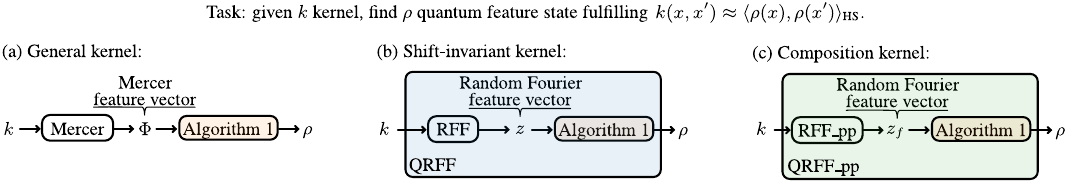}
        \caption{
            Conceptual sketch of three constructions to approximate kernel functions as Embedding Quantum Kernels (EQKs).
            The three parts correspond to different kernel families: (a) General kernels refers to any PSD kernel function, as introduced in Definition~\ref{def:kernel}; (b) Shift-invariant kernels are introduced in Definition~\ref{def:shift-invariant}; and (c) Composition kernels are a new family that we introduce in Section~\ref{s:composition}.
            The boxes refer to routines taken either from the existing literature (Mercer, from Corollary~A.3 in Appendix~A; and RFF, from Theorem~\ref{thm:RFF_existence} originally in Ref.~\cite{rahimi2007random}) or introduced in this manuscript (Algorithm~\ref{alg:QEPIP}; QRFF, as Algorithm~\ref{alg:QRFF}; RFF\_pp, as Algorithm~\ref{alg:RFF_PP}; and QRFF\_pp, as Algorithm~\ref{alg:QRFF_PP}).
            The details for each of the three families are elucidated in the corresponding sections: the universality of EQKs for general kernels is explained in Section~\ref{s:universality}, with Theorem~\ref{thm:approx_universality}; the universality of efficient EQKs for shift-invariant kernels appears in Section~\ref{s:shift-invariant}, studied formally in Corollaries~\ref{cor:secondderivative} and~\ref{cor:time_efficient}; finally, composition kernels are introduced in Section~\ref{s:composition}, with Proposition~\ref{prop:composition} stating their efficient approximation as EQKs, and Proposition~\ref{prop:projected} confirming that composition kernels contain the so-called projected quantum kernels presented in Refs.~\cite{huang2021power,huang2022provably}.
        }
        \label{fig:fig1}
    \end{figure*}

    In this section we present our second result: All shift-invariant kernels admit space-efficient EQK approximation provided they are smooth enough.
    We also give sufficient conditions for a constructive time-efficient EQK approximation.
    We arrive at these results in two steps:
    We first prove an upper bound on the Hilbert space dimension required for an approximation as an explicit inner product, classically.
    We next construct and EQK based on this classical approximation.
    
    Shift-invariant kernels have enjoyed significant attention in the ML literature.
    On the one hand, the Gaussian RBF kernel (arguably the most well-known shift-invariant kernel) has been found useful in a range of data-driven tasks.
    On the other hand, as we see in this section, shift-invariant kernels are more amenable to analytical study than other classes of kernels.
    The property of shift-invariance, combined with exchange symmetry and PSD, allows for deep mathematical characterization of functions.
    
    Let us first introduce the class of functions of interest:
    
    \begin{definition}[Shift-invariant kernel function]\label{def:shift-invariant}
        A kernel function $k\colon\calX\times\calX\to\bbR$ is called \emph{shift-invariant} if, for every $x,x',\xi\in\calX$, it holds that
        \begin{align}
            k(x+\xi,x'+\xi) &= k(x,x').
        \end{align}
    \end{definition}
    
    As is standard, we then write shift-invariant kernels as a function of a single argument $k(x-x')$ (which amounts to taking $\xi=-x'$.)
    We define $\Delta\coloneqq x-x'$ and then talk about $k(\Delta)$.

    One motivation for using shift-invariant kernels, as a more restricted function family, is that they have a few useful properties that ease their analysis.
    Also, it is difficult to decide whether an arbitrary function is PSD, so characterizing general kernel functions is difficult.
    Conversely, for shift-invariant functions, Bochner's theorem gives a condition that is equivalent with being PSD.

    \begin{theorem}[Bochner~\cite{Rudin1990fourier}]\label{thm:Bochner}
        Let $k$ be a continuous, even, shift-invariant function on $\bbR^d$.
        Then $k$ is PSD if and only if $k$ is the \emph{Fourier transform} (FT)
        of a non-negative measure $p$,
        \begin{align}
            k(\Delta) &= \int_{\bbR^d} p(\omega) e^{i\langle\omega,\Delta\rangle}\mathrm{d}\omega.
        \end{align}
        Furthermore, if $k(0)=1$, then $p$ is a probability distribution $\int_{\bbR^d} p(\omega)\mathrm{d}\omega=1$.
    \end{theorem}

    Bochner's theorem is also a central ingredient in a powerful kernel-approximation algorithm known as \emph{random Fourier features} 
    (RFF)~\cite{rahimi2007random} presented here as Algorithm~\ref{alg:RFF}.
    \begin{algorithm}[H]
        \caption{random Fourier features~\cite{rahimi2007random}, $\mathtt{RFF}$}
        \label{alg:RFF}
        \begin{algorithmic}[1]
            \Require A PSD, continuous, shift-invariant kernel ${k(x-x')}$.
            \Ensure A feature map $z$ so that $\langle z(x),z(x')\rangle\approx k(x-x')$.
            \State \label{algline:iFT} $p(\omega)\gets\frac{1}{2\pi}\int_{\bbR^d} e^{-i\langle\omega,\Delta\rangle}k(\Delta)\mathrm{d}\Delta$ \Comment Inverse FT of $k$.
            \State \label{algline:sample} $\{\omega_1,\ldots,\omega_{D/2}\}\sim p^{D/2}$ \Comment Draw $D/2$ i.i.d. samples from $p$.
            \State \label{algline:featmap} $z(\cdot) \gets \sqrt{\frac{2}{D}} \begin{pmatrix}
                \cos\langle\omega_1,\cdot\rangle \\ \sin\langle\omega_1,\cdot\rangle\\ \vdots \\ \cos\langle\omega_{D/2},\cdot\rangle \\ \sin\langle\omega_{D/2}, \cdot\rangle
            \end{pmatrix}$.
            \State\Return $z$
        \end{algorithmic}
    \end{algorithm}

    The input to Algorithm~\ref{alg:RFF} is any shift-invariant kernel function $k$, and the output is a feature map $z$ such that the same kernel can be $\varepsilon$-approximated as an explicit inner product.
    In Step~\ref{algline:iFT}, the inverse Fourier transform $p$ of the kernel $k$ is produced.
    In Step~\ref{algline:sample}, $D/2$ many samples $\omega_i$ are drawn i.i.d.~according to $p$.
    In Step~\ref{algline:featmap}, the feature map $z$ is constructed using the samples $\omega_i$ and the sine and cosine functions.
    That the inner product of $z$ gives a good approximation to $k$ is further elucidated in Appendix~A.
    
    With RFF, given a kernel function, Algorithm~\ref{alg:RFF} produces a randomized feature map such that the kernel corresponding to the inner product of pairs of such maps is an unbiased estimator of the initial kernel.
    The vital question is how large the dimension $D$ has to be in order to ensure $\varepsilon$ approximation error, which is the object of study of Theorem~\ref{thm:RFF_existence}.

    \begin{theorem}[Random Fourier features, Claim 1 in Ref.~\cite{rahimi2007random}]\label{thm:RFF_existence}
        Let $\calX\subseteq\bbR^d$ be a compact data domain.
        Let $k$ be a continuous shift-invariant kernel acting on $\calX$, fulfilling $k(0)=1$.
        Then, for the probabilistic feature map $z(\cdot)\colon\calX\to\bbR^D$ produced by Algorithm~\ref{alg:RFF},
        \begin{align}
            \bbP&\left[\sup_{x,x'\in\calX} \lvert \langle z(x),z(x')\rangle - k(x,x')\rvert \geq\varepsilon\right] \leq \\
            &\leq 2^8 \left(\frac{\sigma_p\operatorname{diam}(\calX)}{\varepsilon}\right)^2\exp\left(-\frac{D\varepsilon^2}{8(d+2)}\right)
        \end{align}
        holds for almost every $x,x'\in\calX$.
        Here $\operatorname{diam}(\calX)=\sup_{x,x'\in\calX}\{\lVert x-x'\rVert\}$ is the diameter of $\calX$, and $\sigma_p$ is the variance of the inverse FT of $k$ interpreted as a probability distribution
        \begin{align}
            p(\omega) &= \FT^{-1}[k](\omega),\\
            \sigma_p^2 &= \bbE_{p}\left[\lVert\omega\rVert^2\right].
        \end{align}
        In particular, it follows that for any constant success probability, there exists an $\varepsilon$-approximation of $k$ as an Euclidean inner product where the feature space dimension $D$ satisfies
        \begin{align}
            D &\in 
            \calO\left(\frac{d}{\varepsilon^2}\log\frac{\sigma_p\operatorname{diam}(\calX)}{\varepsilon}\right).
        \end{align}
    \end{theorem}
    
    The RFF construction can fail to produce a good approximation with a certain probability, but the failure probability can be pushed down arbitrarily close to $0$ efficiently.
    
    This theorem can be understood as a probabilistic existence result for efficient embedding-quantum approximations of kernel functions, which we present next as our second main contribution.
    We consider the input dimension $d$ as our scaling parameter, which plays the role of $s$.
    We further take $\varepsilon$-approximation to be the supremum of the pointwise difference almost everywhere, which we inherit from Theorem~\ref{thm:RFF_existence}.
    In this result we do not need to fix how the input kernel sequence $\{k_s\}_{s\in\bbN}$ is specified, but we assume it is specified in a way that allows us to approximate it using a quantum computer.
    When it comes to what definition of efficiency we need to use, we consider first space efficiency, since we talk about the required number of dimensions to approximate the kernel as an explicit inner product.
    Later on we pin down the time complexity as well.

    In the following, for clarity of presentation, we decouple the construction of EQKs via RFFs as a two-step process: first produce the (classical) feature map $z$, and second realize the inner product as an EQK as in Lemma~\ref{l:innerprodqs}.
    As of now, Algorithm~\ref{alg:RFF} produces a classical approximation via Random Fourier Features, not an EQK yet.
    Next we add a smoothness assumption to Theorem~\ref{thm:RFF_existence} to produce Corollary~\ref{cor:secondderivative}.
    Further below we introduce Algorithm~\ref{alg:QRFF}, which takes the feature map $z$ from Algorithm~\ref{alg:RFF} and produces an EQK from it using Algorithm~\ref{alg:QEPIP}.

    \begin{corollary}[Smooth shift-invariant kernels]\label{cor:secondderivative}
        Let $\calX_d\subseteq[-R,R]^d$ be a compact domain.
        For $\varepsilon>0$, let $k_d$ be a continuous shift-invariant kernel function.
        Assume the kernel fulfills $k_d(0)=1$ and it has bounded second derivatives at the origin $\lvert\partial^2_i k_d(0)\rvert\leq B$.
        Then Algorithm~\ref{alg:RFF} produces an $\varepsilon$-approximation of $k_d$ as an explicit inner product.
        In particular, the scaling of the required dimension $D$ of the (probabilistic) feature map is
        \begin{align}
            D &\in\calO\left(\frac{d}{\varepsilon^2}\log\frac{Rd\sqrt{B}}{\varepsilon}\right).
        \end{align}
    \end{corollary}
    \begin{proof}
        From Bochner's theorem (Theorem~\ref{thm:Bochner}), 
        we know that $k_d$ is the FT of a probability distribution $p_d$.
        Next, a standard Fourier identity allows us to relate the variance $\sigma^2_d$ of $p_d$ and the trace of the Hessian $H(k_d)$ at the origin $\Delta=0$ (see, e.g., Refs.~\cite{rahimi2007random, scholkopf2002learning}):
        \begin{align}
            \sigma^2_d &= -\Tr\left\{\left.H(k_d)\right|_{\Delta=0}\right\} = -\sum_{i=1}^d\left.\frac{\partial^2k_d}{\partial\Delta_i^2}\right|_{\Delta=0}.
        \end{align}
        Using the assumption that $\lvert\partial_i^2 k(0)\rvert\leq B$ for all $i\in[d]$, this results in 
        $\sigma_d^2\leq dB$.
        In parallel, we can upper bound the diameter of $\calX_d$ by the diameter of $[-R,R]^d$ to obtain $\operatorname{diam}(\calX_d)\leq2R\sqrt{d}$.
        By plugging the bounds on $\sigma^2_d$ and $\operatorname{diam}(\calX_d)$ into the bound of Theorem~\ref{thm:RFF_existence} we obtain the claimed result
        \begin{align}
            D &\in\calO\left(\frac{d}{\varepsilon^2}\log\frac{\sigma_{d}\operatorname{diam}(\calX_d)}{\varepsilon}\right) \in\calO\left(\frac{d}{\varepsilon^2}\log\frac{Rd\sqrt{B}}{\varepsilon}
            \right).
        \end{align}
    \end{proof}

    In Appendix~C, we provide Corollary~C.1, which fixes the number of bits of precision required to achieve a close approximation to the second derivative using finite difference methods.

    Indeed, Corollary~\ref{cor:secondderivative} ensures that the Hilbert-space dimension required to $\varepsilon$-approximate any shift-invariant kernel fulfilling mild conditions of smoothness scales at most polynomially in the relevant scale parameters, if $R$ and $B$ are considered to be constant.
    This does not provide an answer to Question~\ref{Q:existence_neqk} yet, as the result only talks about existence of a space-efficient approximation, and not about the complexity of finding such an approximation only from the specification of each kernel in the sequence $\{k_d\}_d$.
    Noteworthy is that so far we have not made any assumptions on the complexity of finding the inverse FT of $k_d$, nor the complexity of sampling from it.
    
    Before going forward, one could enquire whether the upper bound from Corollary~\ref{cor:secondderivative} could be forced to scale exponentially in $d$.
    In this direction, one would e.g. consider cases where $R$ and $B$ are not fixed, but rather also depend on $d$.
    But, since both $R$ and $B$ appear inside a logarithm, in order to achieve a scaling exponential in $d$ overall, we would need to let either $R$ or $B$ to grow doubly-exponentially.
    While that remains a possibility, we point out that in such a case we would easily lose the ability to $\varepsilon$-approximate each of the kernels $k_d$ quantum-efficiently to begin with, so we judge these scenarios as less relevant.
    And even then, it should be noted that Theorem~\ref{thm:RFF_existence} offers only an upper bound to the required dimension $D$.
    In order to discuss whether $D$ can be forced to scale exponentially, we would also need a lower bound.
    The same reasoning applies for $\varepsilon$.

    Notice in Corollary~\ref{cor:secondderivative} we talk about the required feature dimension, not the number of qubits.
    Indeed, we can encode the feature vectors (which are nicely normalized) onto quantum states, as presented in Algorithm~\ref{alg:QRFF}, which we call \emph{quantum random Fourier features} (QRFF).
    What QRFF does is first obtain the probablistic map $z$ from RFF, and then encode it into quantum states and take their inner product as in Lemma~\ref{l:innerprodqs}.
    By construction, the feature maps produced by RFF are unit vectors with respect to the $2$-norm, $\lVert z(x)\rVert_2=1$ for any $x\in\calX$.
    For Algorithm~\ref{alg:QEPIP} we require unit vectors with respect to the $1$-norm, so we need to renormalize the vectors and this introduces another multiplicative factor, which now depends on the input vectors:
    
    \begin{lemma}[Inner product normalization]\label{l:2normQRFF}
        Let $r,r'\in\ell_2^d$ be $2$-norm unit vectors $\lVert r\rVert_2 = \lVert r'\rVert_2=1$.
        Then, the  identity 
         \begin{align}
            \langle r,r'\rangle &= \lVert r\rVert_1\lVert r'\rVert_1\left(2^n \Tr\left\{\rho_{\tr}\rho_{\tr'}\right\} -1\right)
        \end{align}
        holds, where $\tr^{(\prime)}=r^{(\prime)}/\lVert r^{(\prime)}\rVert_1\in\ell_1^d$ corresponds to 
        renormalizing with respect to the $1$-norm.
        Here $\rho_{\tr}$ refers to encoding $\tr$ onto a quantum state using Algorithm~\ref{alg:QEPIP}.
    \end{lemma}

    \begin{proof}
    The proof is given in Appendix~B.
    \end{proof}
    \begin{remark}{(Bounded pre-factors).}
        The re-normalization is of no concern, since the fact $r,r'$ are $2$-norm unit vector implies that their $1$-norm is bounded $\lVert r\rVert_1\in[1,\sqrt{d}]$.
    \end{remark}
    This explains the extra factor $g(x)g(x')$ appearing in Algorithm~\ref{alg:QRFF}, where we have
    \begin{align}\label{eq:g(x)}
        g(x) &= \lVert z(x)\rVert_1\\
        &= \sqrt{\frac{2}{D}}\sum_{i=1}^{D/2} \lvert\cos\langle\omega_i,x\rangle\rvert + \lvert\sin\langle\omega_i,x\rangle\rvert,\nonumber
    \end{align}
    and $g(x)\in[1,\sqrt{D}]$, so $g(x)g(x')\in[1,D]$ for all $x,x'\in\calX$.
    
    \begin{algorithm}[H]
        \caption{Quantum random Fourier features, $\mathtt{QRFF}$}
        \label{alg:QRFF}
        \begin{algorithmic}[1]
            \Require A PSD, continuous, shift-invariant kernel ${k(x-x')}$.
            \Ensure A quantum feature map $\rho$ so that $k(x-x')\approx g(x)g(x') (2^n\Tr\{\rho(x),\rho(x')\}-1)$. \Comment $g(x)$ defined in Eq.~\eqref{eq:g(x)}.
            \State \label{algline:callRFF} $z \gets \mathtt{RFF}(k)$ \Comment Apply Algorithm~\ref{alg:RFF}.
            \State \label{algline:qfeatmap} $\rho(\cdot) \gets \mathtt{C2QE}\left(z(\cdot)\right)$ \Comment Apply Algorithm~\ref{alg:QEPIP}.
            \State\Return$\rho$
        \end{algorithmic}
    \end{algorithm}
    
    The number of qubits required $n$ scales logarithmically in the dimension of the feature vector $n\in\calO(\log(D))$.
    So the number of qubits $n$ necessary for $\varepsilon$-approximating shift-invariant kernel functions as EQKs has scaling $n\in\tcalO\left(\log d/\varepsilon^2\right)$, where the tilde hides doubly-logarithmic contributions.

    With these, we can almost conclude that Corollary~\ref{cor:secondderivative} results in a positive answer to Question~\ref{Q:existence_neqk} for shift-invariant kernels provided they are smooth, with smoothness quantified as the magnitude of the second derivative at the origin.
    We are still missing the complexity of producing samples from the inverse FT of each $k_d$, as in Step~\ref{algline:sample} of Algorithm~\ref{alg:RFF}.
    The efficiency criterion taken here would be the required number of qubits.
    Nevertheless, it is true that for any smooth shift-invariant kernel, there exists an $\varepsilon$-approximation as an EQK using at most logarithmically many qubits following Algorithm~\ref{alg:QRFF}.
    
    The complexity of Algorithm~\ref{alg:RFF} could become arbitrarily large based on the difficulty of sampling from the distribution $p_d$ corresponding to the inverse FT of $k_d$.
    In order to ensure embedding-quantum efficiency, we need to add the requirement of efficient sampling from the inverse FT of the kernel.
    With this, we proceed to state our main result: efficient EQKs are universal within the class of smooth shift-invariant kernels whose inverse FT is efficiently sampleable.
    
    \begin{corollary}[Polynomial run-time]
    \label{cor:time_efficient}
        Under the assumptions of Corollary~\ref{cor:secondderivative}, let $p_d$ be the inverse FT of the kernel $k_d$.
        If obtaining samples from $p_d$ can be done in time $t\in\calO(\poly(d))$, then both Algorithms~\ref{alg:RFF} and~\ref{alg:QRFF} have total run-time complexity polynomial in $d$.
    \end{corollary}
    \begin{proof}
        For Algorithm~\ref{alg:RFF}, Step~\ref{algline:iFT} is only a mathematical definition, Step~\ref{algline:sample} runs in time linear in $D$ and polynomial in $d$ by assumption, finally Step~\ref{algline:featmap} also runs in time linear in $D$ and polynomial in $d$.
        Corollary~\ref{cor:secondderivative} further states that $D$ is at most essentially linear in $d$.
        In total, Algorithm~\ref{alg:RFF} has run-time complexity at most polynomial in $d$.

        In turn, the first step of Algorithm~\ref{alg:QRFF} is calling Algorithm~\ref{alg:RFF}, which we just saw takes time polynomial in $d$.
        Step~\ref{algline:qfeatmap} uses Algorithm~\ref{alg:QEPIP}, which also runs in time polynomial in $d$.
        With this it is also clear that the total run-time complexity of Algorithm~\ref{alg:QRFF} is at most polynomial in $d$.
    \end{proof}

    With the added condition of producing samples in polynomial time, we obtain a positive answer to Question~\ref{Q:existence_neqk}.
    Namely, we show the existence of time-efficient EQK approximations for any kernel fulfilling the smoothness and efficient sampling conditions from Corollaries~\ref{cor:secondderivative} and~\ref{cor:time_efficient}.
    Notice our result does not even require quantum efficiency of $k_d$, so in fact we have proved something stronger.

    One point to address is whether assuming quantum efficiency for evaluating $k_d$ directly implies the ability to sample efficiently from the inverse FT of $k_d$.
    In the cases where this is true, adding quantum-efficiency to the assumptions of Corollary~\ref{cor:secondderivative} suffices to prove that efficient EQKs are universal within the class of efficient shift-invariant kernels.
    At the face of it, it is unclear whether for all reasonable ways of specifying $k_d$ the capacity to efficiently evaluate it using a quantum computer implies an efficient algorithm of sampling from the distribution obtained by the inverse FT.
    If the task were to sample from $k_d$ itself (in case $k_d$ were a probability distribution), then we know that being able to evaluate $k_d$ does not imply the capacity to efficiently sample from $k_d$\footnote{In the black box model, lower bounds on Grover's search algorithm imply an exponential cost of sampling.
    Also, NP-hardness of e.g. sampling from low-temperature Gibbs states of the Ising model~\cite{troyer2005computational} provides examples relevant for other input models.}.
    Then, it is unclear why sampling from the inverse FT of $k_d$ would be easy in every case, especially, e.g., in the black mox model.
    Resolving this question is left as an open problem, and we note it is not an entirely new one~\cite{kalai2013complexity, schwarz2013simulating}.

    Finally, it should be noted that, although the feature map produced by Algorithm~\ref{alg:RFF} can be stored classically (assuming $D\in\calO(d/\varepsilon^2)$), this is not a de-quantization algorithm.
    Algorithm~\ref{alg:RFF} requires sampling from the inverse FT of the input kernels $k_d$.
    Reason dictates that, if the kernel is quantum efficient and classically inefficient to $\varepsilon$-approximate, in general sampling from its inverse FT should also be at least classically inefficient.
    This is what we meant earlier when we said we consider quantum time: even though the intermediate variable $z(x)$ can be stored classically, producing it requires usage of a quantum computer.
    
    In this section, we show that smooth shift-invariant kernels admit space-efficient EQK approximations.
    We also give sufficient conditions for the same result to hold for embedding-quantum efficiency in total runtime.
    In the next section, we see how the same ideas extend to another class of kernel functions, beyond shift-invariant ones.

\section{Composition kernel and projected quantum kernel}\label{s:composition}

    In this section we introduce a new class of quantum kernels which can still be turned into EQKs using another variant of Algorithm~\ref{alg:RFF}.
    We show that the new class also admits time-efficient approximations as EQK, and that the projected quantum kernel from Ref.~\cite{huang2021power} belongs to this class.
    During the publication phase of this manuscript we were made aware of the similarities between the composition kernels we introduce here and the \enquote{distance kernel with RDM} proposed and studied in Ref.~\cite{nakaji2022deterministic}.
    
    Again, we separate the EQK-based approximation into two steps: first we propose a variant of RFF producing a classical feature map, and second we construct an EQK that evaluates the inner product of pairs of features.

    First we introduce the new class.
    Consider the usual Gaussian kernel with parameter $\sigma>0$ defined as
    \begin{align}
        k(x,x') &= e^{-\frac{\lVert x-x'\rVert^2}{2\sigma^2}},
    \end{align}
    only now we allow for some pre-processing of the inputs ${x\mapsto f(x)}$, resulting in
    \begin{align}\label{eq:composition}
        k_f(x,x') &= e^{-\frac{\lVert f(x)-f(x')\rVert^2}{2\sigma^2}}.
    \end{align}
    The introduction of $f$ breaks shift invariance in general, hence the need to specify both arguments independently $x,x'\in\calX$.
    Since $k_f$ is a PSD kernel for any function $f$, we refer to such constructions as \emph{composition kernel}.
    Next we propose a generalization of Algorithm~\ref{alg:RFF} that also works for composition kernels, as Algorithm~\ref{alg:RFF_PP}, which we call simply \emph{random Fourier features with pre-processing} (RFF\_pp).
    
    \begin{algorithm}[H]
        \caption{RFF with pre-processing, $\mathtt{RFF\_pp}$}
        \label{alg:RFF_PP}
        \begin{algorithmic}[1]
            \Require A composition kernel $k_f(x,x')$, see Eq.~\eqref{eq:composition}.
            \Ensure A feature map $z_f$ so that $k_f(x,x')\approx\langle z_f(x),z_f(x')\rangle$ .
            \State $z \gets \mathtt{RFF}(k\colon f(\calX)\to\bbR)$ \Comment Apply Algorithm~\ref{alg:RFF}.\label{line:domain}
            \State $z_f(\cdot) \gets z(f(\cdot))$ \Comment Apply the pre-processing function.
            \State\Return $z_f$
        \end{algorithmic}
    \end{algorithm}

    Since the core feature of Algorithm~\ref{alg:RFF_PP} is to call Algorithm~\ref{alg:RFF}, we inherit the $\varepsilon$-approximation guarantee.
    Notice in line~\ref{line:domain} of Algorithm~\ref{alg:RFF_PP} we invoke Algorithm~\ref{alg:RFF}, 
    but taking as input domain the range of the pre-processing function, $f(\calX)$, instead of the original domain $\calX$.
    If we assume $f$ to be continuous, it follows that $f(\calX)$ is also compact, which we require for the application of Theorem~\ref{thm:RFF_existence}.

    \begin{proposition}[Performance guarantee of
    Algorithm~\ref{alg:RFF_PP}]\label{prop:composition}
        Let ${f\colon\calX\to[-B,B]^{g_1(d)}}$ be a pre-processing function, and let $k_f$ be the Gaussian kernel composed with $f$, as introduced in Eq.~\eqref{eq:composition}.
        Let finally the parameter of the Gaussian kernel be $\sigma=g_2(d)$.
        If $g_1(d)\in\calO(\poly(d))$ and $g_2(d)\in\Omega(\poly(d)^{-1})$, then Algorithm~\ref{alg:RFF_PP} produces an $\varepsilon$-approximation of $k_f(x,x')$ as an explicit inner product.
        In particular, the required dimension $D$ of the randomized feature map is at most polynomial in the input dimension $d$ and the inverse error $1/\varepsilon$.
        \begin{align}
            D &\in\calO\left(\frac{\poly(d)}{\varepsilon^2}\log\left(\frac{dB}{\varepsilon}\right)\right) \in \tcalO\left(\frac{\poly(d)}{\varepsilon^2}\right).
        \end{align}
    \end{proposition}
    
    \begin{proof}
    The proof is provided in Appendix~D.
    \end{proof}

    Notice Proposition~\ref{prop:composition} shares a deep similarity with Corollary~\ref{cor:secondderivative}.
    Both are direct applications of Theorem~\ref{thm:RFF_existence} to kernels fulfilling different properties.
    They have nevertheless one stark difference.
    The assumptions in Proposition~\ref{prop:composition} are sufficient to guarantee we can sample efficiently from the inverse FT of the kernels in the sequence $k_d$.
    This is because the probability distributions involved, $p_d$, are nothing but Gaussian distribution themselves.
    As shown in Algorithm~\ref{alg:RFF_PP}, the pre-processing function $f$ shows up only after the sampling step.
    So, unlike Corollary~\ref{cor:secondderivative}, Proposition~\ref{prop:composition} already represents a direct positive answer to Question~\ref{Q:existence_neqk}, relying on of Algorithm~\ref{alg:QRFF_PP}, which follows.

    In Algorithm~\ref{alg:QRFF_PP}, \emph{quantum random Fourier features with pre-processing} (QRFF\_pp), we take the output of Algorithm~\ref{alg:RFF_PP} and convert it into a quantum feature map with Algorithm~\ref{alg:QEPIP}, leading to the corresponding EQK approximation.
    Again now the multiplicative term $g(f(x))g(f(x'))$ appears from the re-normalization in Lemma~\ref{l:2normQRFF}.
    
    \begin{algorithm}[H]
        \caption{QRFF with pre-processing, $\mathtt{QRFF\_pp}$}
        \label{alg:QRFF_PP}
        \begin{algorithmic}[1]
            \Require A composition kernel $k_f(x,x')$, see Eq.~\eqref{eq:composition}.
            \Ensure A quantum feature map $\rho_f$ so that $k_f(x,x')\approx g(f(x))g(f(x'))(2^n\Tr\{\rho_f(x)\rho_f(x')\}-1)$. \Comment $g(x)$ defined in Eq.~\eqref{eq:g(x)}.
            \State $z_f \gets \mathtt{RFF\_pp}(k_f)$ \Comment Apply Algorithm~\ref{alg:RFF_PP}.
            \State $\rho_f(\cdot) \gets \mathtt{C2QE}\left(z_f(\cdot)\right)$ \Comment Apply Algorithm~\ref{alg:QEPIP}.
            \State\Return$\rho_f$
        \end{algorithmic}
    \end{algorithm}

    The scaling in the number of qubits is once more logarithmic in the scaling of the number of required dimensions given in Proposition~\ref{prop:composition}.
    For the number of qubits $n$, the scaling is $n\in \tcalO\left(\log(d/\varepsilon^2)\right)$.
    Moreover, the run-time complexity of Algorithm~\ref{alg:QRFF_PP} is polynomial in $d$ and in the run-time complexity of evaluating $f$, since now the sampling step corresponds to sampling from a product Gaussian distribution.

    For completeness sake, it would be good also to rule out the possibility of classically efficient approximation.
    Although it might be counter-intuitive, in principle there could exist pre-processing functions $f$ that are hard to evaluate but which result in a composition kernel $k_f$ that is not hard to evaluate.
    In the following we just confirm that there exist pre-processing functions which are hard to evaluate classically which result in composition kernels which are also hard to evaluate classically.
    
    \begin{proposition}[No efficient classical approximation]\label{prop:f_hardness}
        There exists a function $f\colon\calX\to[0,1]^d$ which can be $\varepsilon$-approximated quantum efficiently in $d$ for which the composition kernel $k_f$ cannot be $\varepsilon$-approximated classically efficiently in $d$, with
        \begin{align}
            k_f(x,x') &\coloneqq \exp\left(-\lVert f(x)-f(x')\rVert^2\right).
        \end{align}
        We take $\sigma^2=1/2$ for simplicity.
    \end{proposition}
    
    \begin{proof}
        Here, the proof is presented in Appendix~D.
    \end{proof}
    
    By selecting pre-processing functions that are quantum efficient but classically inefficient, we reach a class of quantum kernels that are not shift-invariant, but for which the RFF\_pp construction still applies, from Algorithm~\ref{alg:RFF_PP}.
    Next, we show that this class of kernels contains the recently introduced \emph{projected quantum kernel}~\cite{huang2021power}.
    
    Going back to the quantum feature map $x\mapsto\rho(x)$, the authors of Ref.~\cite{huang2021power} proposed a mapping from the exponentially-sized $\rho(x)$, to an array of reduced density matrices $(\rho_k(x))_{k=1}^N$, for an $N$-qubit quantum state\footnote{
        We distinguish $N$, the number of qubits used to evaluate the projected quantum kernel, from $n$, the number of qubits required to approximate it as an EQK.
    }.
    According to our definitions, we would not call this a quantum feature map, since the data is not mapped onto the Hilbert space of quantum states.
    Rather, it is natural to think of this as a quantum pre-processing function.
    A nice property of this alternative mapping is that the feature vectors can be efficiently stored in classical memory, even though obtaining each of the $\rho_k(x)$ matrices can only be done efficiently using a quantum computer in general.
    Once these are stored, the projected quantum kernel $k^\text{PQ}$ is the composition kernel as we introduced earlier, just setting $f$ to be the function that computes the entries of all the reduced density matrices
    \begin{align}
        k^\text{PQ}(x,x') &= \exp\left(-\gamma\sum_{k=1}^N\lVert\rho_k(x)-\rho_k(x')\rVert_F^2\right),
    \end{align}
    where $\gamma>0$ can be safely taken as $\gamma=1/(2\sigma^2)$, $\rho_k(x)$ is the reduced density matrix of the $k^\text{th}$ qubit of $\rho$, and recall $\lVert\cdot\rVert_F$ is the Frobenius norm.
    The number of qubits $N$ is left as a degree of freedom, but for $d$-dimensional input data, reason would say the number of qubits used would be $N\propto d$, or at most $N\in\calO(\poly(d))$.

    The projected quantum kernel enjoys valuable features.
    On the one hand, it was used in an effort to prove quantum-classical learning separations in the context of data from quantum experiments~\cite{huang2021power}.
    On the other hand, the projected quantum kernel is less vulnerable to the exponential kernel concentration problem~\cite{thanasilp2022exponential}.
    The projected quantum kernel is also deeply related to the shadow tomography formalism and guarantees can be placed on its performance for quantum phase recognition~\cite{huang2022provably} among others.

    \begin{proposition}[Projected quantum kernel as an efficient EQK]\label{prop:projected}
        The projected quantum kernel $k^\text{PQ}$ fulfills the assumptions of Proposition~\ref{prop:composition}, so it can be efficiently approximated as an EQK with the number of dimensions required $D$ fulfilling
    \begin{align}
        D\in \tcalO\left(\frac{\poly(d)}{\varepsilon^2}\right).
    \end{align}
    \end{proposition}

    \begin{proof}
        The proof of this statement is presented in Appendix~D.
    \end{proof}

    \begin{remark}{(General projected quantum kernels)\textbf{.}}
        In the proof, we have assumed that the projected quantum kernel is taken with respect to the subsystems being each individual qubit.
        A general definition of the projected quantum kernel would allow for other subsystems, but the requirement that the feature map must be efficient to store classically prevents that the number of subsystems scales more than polynomially, and that the local dimension of each subsystem scales more than logarithmically, in which case we still have $g_1(d)\in\calO(\poly(d))$.
    \end{remark}

    The projected quantum kernel poses a clear example of a quantum kernel that is not constructed from a feature map onto the Hilbert space of density operators.
    Its closeness to the classical shadow formalism has earned it attention from the moment it was proposed.
    Thus, proving that also the projected quantum kernel admits an efficient approximation as an EQK results in another important kernel family that is also covered by EQKs.

    In order to identify scenarios in which the scaling $D$ could become super-polynomial in the relevant parameters $d,1/\varepsilon$, one would require for example $g_1(d)\propto\exp(d)$, or $g_2(d)\propto1/\exp(\exp(d))$, using the definitions from Proposition~\ref{prop:composition}.
    Indeed, in those scenarios we would not be able to guarantee the space efficiency of Algorithm~\ref{alg:RFF_PP} while keeping good $\varepsilon$-approximation.
    Nevertheless the conditions $g_1(d)\propto\exp(d)$ or $g_2(d)\propto1/\exp(\exp(d))$ alone would not be sufficient to prove that Algorithm~\ref{alg:RFF_PP} must fail.
    Said otherwise, this is again because from Theorem~\ref{thm:RFF_existence} we only have an upper bound on the required space complexity, and not a lower bound.
    At the same time, though, both cases would prevent us from being able to quantum-efficiently $\varepsilon$-approximate the projected quantum kernel in the first place, so these scenarios must be ruled out as they break the hypotheses.
    And even if they were allowed, recall that we can store classical vectors in quantum states using only logarithmically many qubits in the length of the classical vector.
    That means that single exponential $g_1(d)$ and inverse doubly-exponential $g_2(d)$ would not be enough to require exponential embedding-quantum complexity in the number of qubits.
    In order to have the number of qubits to be, e.g., 
    $n\in\calO(\exp(d))$, we would require either $g_1(d)\in\calO(\exp(\exp(d)))$, or $g_2(d)\in\Omega(1/\exp(\exp(\exp(d))))$, which would prevent our ability to evaluate the projected kernel even further.

    With these considerations, we have shown that projected quantum kernels, despite their using a different feature map and inner product, can be efficiently realized as EQKs.
    As a recapitulation, Fig.~\ref{fig:venn_diagram} summarizes all our contributions as a collection of set inclusions in a Venn diagram.
    With this, to the best of our knowledge, we conclude that all quantum kernels in the literature are either EQKs directly, or they can be efficiently realized as EQKs.


\section{Outlook}\label{s:outlook}

    This manuscript so far offers restricted positive answers to Question~\ref{Q:existence_neqk}.
    In this section, we list some promising directions to search for other restricted answers to the same question.
    Answering Question~\ref{Q:existence_neqk} negatively would involve proving an non-existence result, for which one needs a different set of tools than the ones we have used so far.
    In particular, one would need to wield lower bounds for non-existence, which in this corner of the literature appear to be trickier to find.

\subsection{Time efficient EQKs}

    It should be possible to come up with a concrete construction that ensures also time efficiency of Algorithm~\ref{alg:RFF} while applying Corollary~\ref{cor:time_efficient}.
    For that to happen, it would be enough to give not only the kernel function, but also an efficient algorithm to sample from its inverse FT.
    We recognize this question as the clearest next step in this new research direction.

\subsection{Non-stationary kernels and indefinite kernels}

    The theory of integral kernel operators and their eigenvalues has been already extensively developed in the context of functional analysis~\cite{karlin1964total, hansen2010discrete, ha1986eigenvalues}, also with applications to random processes and correlation functions \cite{yaglom1987correlation}.
    In particular, for $1$-dimensional kernel functions, it is known that reasonable smoothness assumptions lead to kernel spectra which are concentrated on few eigenfunctions.
    This fact alone hints strongly toward the existence of low-dimensional EQK-based approximations for a larger class of kernel functions than the ones we explored here.
    While we restricted ourselves to shift-invariant and composition kernels, the theory of eigenvalues of integral kernel operators seems to suggest that similar results could be obtained for any smooth kernel function.
    That being said, the generalization from $1$ to $d$-dimensional data might well come with factors which grow exponentially with $d$, somewhat akin to the well-known \emph{curse of dimensionality}.

    Alternatively, instead of arbitrary PSD kernel functions, one might also consider other restricted classes of known kernels.
    For example rotation-invariant kernels, of which polynomial kernels are a central example, offer an interesting starting point.
    Polynomial kernels are not shift-invariant, but they are often derived with an explicit feature map in mind.
    A potentially interesting research line would be to generalize polynomial kernels in a way that is less straightforward to turn into an EQK.

    Our results relied strongly on the \emph{random Fourier features} (RFF) algorithm of Ref.~\cite{rahimi2007random}.
    The RFF approximation in turn rests on a sampling protocol that owes its success to Bochner's theorem for shift-invariant PSD functions (Theorem~\ref{thm:Bochner}).
    Yaglom's theorem (Theorem~A.4 in Appendix~A) is referred to as the generalization of Bochner's theorem to functions that are not shift-invariant.
    It would be interesting to see whether there are other, non shift-invariant kernels for which Yaglom's theorem could be used to furbish an explicit feature map with approximation guarantees akin to those of Theorem~\ref{thm:RFF_existence}. 

    The extension from PSD to non-PSD functions could be similar to the extension from shift-invariant to smooth functions we just described.
    Intuitively, finding an EQK approximation for a kernel is similar to finding its singular value decomposition in an infinite-dimensional function space.
    In this space, PSD functions result in the \enquote{left}- and \enquote{right}-hand matrices in the singular value decomposition to be equal (up to conjugation).
    Conversely, for non-PSD functions, the main difference would be the difference between left and right feature maps.

    In this direction, we identify four promising research lines: first, generalize our study to the general class of $1$-dimensional smooth functions; second, study restricted classes of higher-dimensional kernel functions for which the integral kernel operator retains desirable qualities for approximation; third, find restricted cases where a kernel approximation based on Yaglom's theorem is possible; and fourth, extend similar results for indefinite kernel functions.
    We comment on these directions in Appendix~A.

\subsection{The variational QML lens}
    
    When studying EQK-based approximations to given kernel functions, we have not restricted ourselves to PQC-based functions.
    Even though the literature on quantum kernel methods up to this point has not dealt exclusively with variational circuits, a large body of work has indeed focused on functions estimated as the expectation value of a fixed observable with respect to some parametrized quantum state.
    In this sense, combining the results from Ref.~\cite{schuld2021fourier} and the rules for generating PSD-functions from Ref.~\cite{scholkopf2002learning}, one can develop a framework for how to approach quantum kernel functions holistically.
    We give some first steps explicitly in Appendix~E.
    We expect many discoveries to result from an earnest study of what are the main recipes to construct PQC-based quantum kernels other than by the established EQK principles.

\section{Conclusion}\label{s:conclusion}

    In this manuscript we raise and partially answer a fundamental question regarding the expressivity of mainstream quantum kernels.
    We identify that most quantum kernel approaches are of a restricted type, which we call \emph{embedding quantum kernels} (EQKs), and we ask whether this family covers all quantum kernel functions.
    If we leave notions of efficiency aside, we show that all kernel functions can be realized as EQKs, proving their universality.
    Universality is an important ground fact to establish, as it softly supports the usage of EQKs in practice.

    Learning whether EQKs are indeed everything we need for QML tasks or whether we should look beyond them is 
    important in the context of model selection.
    In ML for kernel methods, model selection boils down to choosing a kernel.
    For ML to be successful, it is long known that models should be properly aligned with the data; for different learning tasks, different models should be used.
    This notion is captured by the concept of \emph{inductive bias}~\cite{kuebler2021inductive, peters2022generalization}, which is different for every kernel.
    Given a learning task with data, the first step should always be to gather information about potential structures in the data in order to select a model with a beneficial inductive bias.
    In this scenario, it becomes crucial to then have access to a model with good data-alignment.
    In Ref.~\cite{kuebler2021inductive} some simple EQKs were found to possess inductive biases which are uninteresting for practically relevant data sets.
    These findings fueled the need to ask whether EQKs can also have interesting inductive biases, which is softly confirmed based on their universality.
    When searching for new quantum advantages in ML, the more interesting classes of kernels we have access to, the better.

    We propose Question~\ref{Q:existence_neqk} as a new line of research in \emph{quantum machine learning} (QML).
    Characterizing a relation of order between general quantum kernel functions and the structured EQK functions when considering computational efficiency could have an important impact in the quest for quantum advantages.
    Indeed, if it were found that all practically relevant kernels can be realized as EQKs, then researchers in quantum kernel methods would not need to look for novel, different models to solve learning tasks with.
    Nevertheless, for now there is still room for practically interesting kernels beyond efficient EQKs to exist.
    
    Operationally, the results presented in this manuscript aim at confirming that the class of kernels currently used in practice already provably covers large, promising families of kernel functions.
    The question we focus on is what is ultimately possible with efficient EQKs, and not so much what is the best way of evaluating kernel functions, nor how to construct new, essentially different quantum kernels.
    For instance, given a kernel function and an efficient algorithm to evaluate it, it is not our goal to provide a new evaluation algorithm which is more efficient in some sense.
    Instead, we are interested in knowing whether there exists another similarly efficient algorithm but which fulfills the added structural restriction of being based on an efficient quantum embedding.
    Whether the embedding-based evaluation algorithm has a slightly higher or lower computational cost is beyond our scope, we are interested in the embedding-based algorithm to be efficient by itself.

    After raising Question~\ref{Q:existence_neqk}, we give an answer for the restricted case of shift-invariant kernels.
    In Corollary~\ref{cor:secondderivative} we show that, under reasonable assumptions, all shift-invariant kernels admit a memory-efficient approximation as EQKs.
    Shift-invariant kernels are widely used also in classical ML, so showing that EQKs are still universal for shift-invariant kernels with efficiency considerations is an important milestone.

    While shift-invariant kernels enjoy a privileged position in the classical literature, they have not been as instrumental in QML so far.
    Indeed, many well-known quantum feature maps lead to EQK functions which are not shift-invariant.
    Also, the milestone work~\cite{huang2021power} has introduced the so-called projected quantum kernel, which has been used to prove learning separations between classical and quantum ML models.
    The projected quantum kernel is not shift-invariant.
    In Proposition~\ref{prop:composition}, we adapt our result for shift-invariant kernels so that they carry over to another class of kernel functions, which we call composition kernels, in which the projected quantum kernel is contained.
    This way, we present the phenomenon that even kernel functions that did not arise from an explicit quantum feature map can admit an efficient approximation as EQKs.

    In all, we have seen that two relevant classes of kernels, namely shift-invariant and composition kernels, always admit an efficient approximation as EQKs.
    Our results are clear manifestations of the expressive power of EQKs.
    Nevertheless, many threads remain open to fully characterize the landscape of all efficient quantum kernel functions.
    We invite researchers to join us in this research line by listing promising approaches as outlook in Section~\ref{s:outlook}.

\subsection*{Acknowledgements}

    The authors would like to thank Ryan Sweke, Sofiene Jerbi, and Johannes Jakob Meyer for insightful comments in an earlier version of this draft, and also thank them and Riccardo Molteni for helpful discussions.

    This work has been in part supported by the BMWK 
    (EniQmA),
    BMBF (RealistiQ, MUNIQC-Atoms), 
    the Munich Quantum Valley (K-8), the Quantum 
    Flagship (PasQuans2, Millenion), 
    QuantERA (HQCC),
    the Cluster of Excellence MATH+, the DFG (CRC 183), 
    and the Einstein Foundation 
    (Einstein Research Unit on Quantum
    Devices).
    VD acknowledges the support from the Quantum Delta NL program.
    This work was in part supported by the Dutch Research Council (NWO/OCW), as part of the Quantum Software Consortium program (project number 024.003.037), and is also supported by the European Union under Grant Agreement 101080142, the project EQUALITY.

\bibliography{sources.bib}

\begin{thebibliography}{57}%
\makeatletter
\providecommand \@ifxundefined [1]{%
 \@ifx{#1\undefined}
}%
\providecommand \@ifnum [1]{%
 \ifnum #1\expandafter \@firstoftwo
 \else \expandafter \@secondoftwo
 \fi
}%
\providecommand \@ifx [1]{%
 \ifx #1\expandafter \@firstoftwo
 \else \expandafter \@secondoftwo
 \fi
}%
\providecommand \natexlab [1]{#1}%
\providecommand \enquote  [1]{``#1''}%
\providecommand \bibnamefont  [1]{#1}%
\providecommand \bibfnamefont [1]{#1}%
\providecommand \citenamefont [1]{#1}%
\providecommand \href@noop [0]{\@secondoftwo}%
\providecommand \href [0]{\begingroup \@sanitize@url \@href}%
\providecommand \@href[1]{\@@startlink{#1}\@@href}%
\providecommand \@@href[1]{\endgroup#1\@@endlink}%
\providecommand \@sanitize@url [0]{\catcode `\\12\catcode `\$12\catcode
  `\&12\catcode `\#12\catcode `\^12\catcode `\_12\catcode `\%12\relax}%
\providecommand \@@startlink[1]{}%
\providecommand \@@endlink[0]{}%
\providecommand \url  [0]{\begingroup\@sanitize@url \@url }%
\providecommand \@url [1]{\endgroup\@href {#1}{\urlprefix }}%
\providecommand \urlprefix  [0]{URL }%
\providecommand \Eprint [0]{\href }%
\providecommand \doibase [0]{https://doi.org/}%
\providecommand \selectlanguage [0]{\@gobble}%
\providecommand \bibinfo  [0]{\@secondoftwo}%
\providecommand \bibfield  [0]{\@secondoftwo}%
\providecommand \translation [1]{[#1]}%
\providecommand \BibitemOpen [0]{}%
\providecommand \bibitemStop [0]{}%
\providecommand \bibitemNoStop [0]{.\EOS\space}%
\providecommand \EOS [0]{\spacefactor3000\relax}%
\providecommand \BibitemShut  [1]{\csname bibitem#1\endcsname}%
\let\auto@bib@innerbib\@empty
\bibitem [{\citenamefont {Shor}(1994)}]{shor1994algorithms}%
  \BibitemOpen
  \bibfield  {author} {\bibinfo {author} {\bibfnamefont {P.~W.}\ \bibnamefont
  {Shor}},\ }\bibfield  {title} {\bibinfo {title} {Algorithms for quantum
  computation: discrete logarithms and factoring},\ }in\ \href
  {https://doi.org/10.1109/SFCS.1994.365700} {\emph {\bibinfo {booktitle}
  {Proceedings 35th Ann. Symp. Found. Compu. Sc.}}}\ (\bibinfo {organization}
  {IEEE},\ \bibinfo {year} {1994})\ pp.\ \bibinfo {pages}
  {124--134}\BibitemShut {NoStop}%
\bibitem [{\citenamefont {Nielsen}\ and\ \citenamefont
  {Chuang}(2000)}]{nielsen2000quantum}%
  \BibitemOpen
  \bibfield  {author} {\bibinfo {author} {\bibfnamefont {M.~A.}\ \bibnamefont
  {Nielsen}}\ and\ \bibinfo {author} {\bibfnamefont {I.~L.}\ \bibnamefont
  {Chuang}},\ }\href
  {https://www.cambridge.org/highereducation/books/quantum-computation-and-quantum-information/01E10196D0A682A6AEFFEA52D53BE9AE#overview}
  {\emph {\bibinfo {title} {Quantum information and quantum computation}}}\
  (\bibinfo  {publisher} {Cambridge university press Cambridge},\ \bibinfo
  {year} {2000})\BibitemShut {NoStop}%
\bibitem [{\citenamefont {Montanaro}(2016)}]{montanaro2016quantum}%
  \BibitemOpen
  \bibfield  {author} {\bibinfo {author} {\bibfnamefont {A.}~\bibnamefont
  {Montanaro}},\ }\bibfield  {title} {\bibinfo {title} {Quantum algorithms: an
  overview},\ }\href {https://doi.org/10.1038/npjqi.2015.23} {\bibfield
  {journal} {\bibinfo  {journal} {npj Quant. Inf.}\ }\textbf {\bibinfo {volume}
  {2}},\ \bibinfo {pages} {15023} (\bibinfo {year} {2016})}\BibitemShut
  {NoStop}%
\bibitem [{\citenamefont {Arute}\ \emph {et~al.}(2019)\citenamefont {Arute}
  \emph {et~al.}}]{arute2019quantum}%
  \BibitemOpen
  \bibfield  {author} {\bibinfo {author} {\bibfnamefont {F.}~\bibnamefont
  {Arute}} \emph {et~al.},\ }\bibfield  {title} {\bibinfo {title} {Quantum
  supremacy using a programmable superconducting processor},\ }\href
  {https://doi.org/10.1038/s41586-019-1666-5} {\bibfield  {journal} {\bibinfo
  {journal} {Nature}\ }\textbf {\bibinfo {volume} {574}},\ \bibinfo {pages}
  {505} (\bibinfo {year} {2019})}\BibitemShut {NoStop}%
\bibitem [{\citenamefont {Wu}\ \emph {et~al.}(2021)\citenamefont {Wu} \emph
  {et~al.}}]{wu2021strong}%
  \BibitemOpen
  \bibfield  {author} {\bibinfo {author} {\bibfnamefont {Y.}~\bibnamefont {Wu}}
  \emph {et~al.},\ }\bibfield  {title} {\bibinfo {title} {Strong quantum
  computational advantage using a superconducting quantum processor},\ }\href
  {https://link.aps.org/doi/10.1103/PhysRevLett.127.180501} {\bibfield
  {journal} {\bibinfo  {journal} {Phys. Rev. Lett.}\ }\textbf {\bibinfo
  {volume} {127}},\ \bibinfo {pages} {180501} (\bibinfo {year}
  {2021})}\BibitemShut {NoStop}%
\bibitem [{\citenamefont {Hangleiter}\ and\ \citenamefont
  {Eisert}(2023)}]{hangleiter2022computational}%
  \BibitemOpen
  \bibfield  {author} {\bibinfo {author} {\bibfnamefont {D.}~\bibnamefont
  {Hangleiter}}\ and\ \bibinfo {author} {\bibfnamefont {J.}~\bibnamefont
  {Eisert}},\ }\bibfield  {title} {\bibinfo {title} {Computational advantage of
  quantum random sampling},\ }\href
  {https://doi.org/10.1103/RevModPhys.95.035001} {\bibfield  {journal}
  {\bibinfo  {journal} {Rev. Mod. Phys.}\ }\textbf {\bibinfo {volume} {95}},\
  \bibinfo {pages} {035001} (\bibinfo {year} {2023})}\BibitemShut {NoStop}%
\bibitem [{\citenamefont {Biamonte}\ \emph {et~al.}(2017)\citenamefont
  {Biamonte}, \citenamefont {Wittek}, \citenamefont {Pancotti}, \citenamefont
  {Rebentrost}, \citenamefont {Wiebe},\ and\ \citenamefont
  {Lloyd}}]{biamonte2017quantum}%
  \BibitemOpen
  \bibfield  {author} {\bibinfo {author} {\bibfnamefont {J.}~\bibnamefont
  {Biamonte}}, \bibinfo {author} {\bibfnamefont {P.}~\bibnamefont {Wittek}},
  \bibinfo {author} {\bibfnamefont {N.}~\bibnamefont {Pancotti}}, \bibinfo
  {author} {\bibfnamefont {P.}~\bibnamefont {Rebentrost}}, \bibinfo {author}
  {\bibfnamefont {N.}~\bibnamefont {Wiebe}},\ and\ \bibinfo {author}
  {\bibfnamefont {S.}~\bibnamefont {Lloyd}},\ }\bibfield  {title} {\bibinfo
  {title} {Quantum machine learning},\ }\href
  {https://doi.org/10.1038/nature23474} {\bibfield  {journal} {\bibinfo
  {journal} {Nature}\ }\textbf {\bibinfo {volume} {549}},\ \bibinfo {pages}
  {195} (\bibinfo {year} {2017})}\BibitemShut {NoStop}%
\bibitem [{\citenamefont {Dunjko}\ and\ \citenamefont
  {Briegel}(2018)}]{dunjko2018machine}%
  \BibitemOpen
  \bibfield  {author} {\bibinfo {author} {\bibfnamefont {V.}~\bibnamefont
  {Dunjko}}\ and\ \bibinfo {author} {\bibfnamefont {H.~J.}\ \bibnamefont
  {Briegel}},\ }\bibfield  {title} {\bibinfo {title} {Machine learning {\&}
  artificial intelligence in the quantum domain: a review of recent progress},\
  }\href {https://doi.org/10.1088/1361-6633/aab406} {\bibfield  {journal}
  {\bibinfo  {journal} {Rep. Prog. Phys.}\ }\textbf {\bibinfo {volume} {81}},\
  \bibinfo {pages} {074001} (\bibinfo {year} {2018})}\BibitemShut {NoStop}%
\bibitem [{\citenamefont {Carleo}\ \emph {et~al.}(2019)\citenamefont {Carleo},
  \citenamefont {Cirac}, \citenamefont {Cranmer}, \citenamefont {Daudet},
  \citenamefont {Schuld}, \citenamefont {Tishby}, \citenamefont
  {Vogt-Maranto},\ and\ \citenamefont {Zdeborov\'a}}]{carleo2019machine}%
  \BibitemOpen
  \bibfield  {author} {\bibinfo {author} {\bibfnamefont {G.}~\bibnamefont
  {Carleo}}, \bibinfo {author} {\bibfnamefont {I.}~\bibnamefont {Cirac}},
  \bibinfo {author} {\bibfnamefont {K.}~\bibnamefont {Cranmer}}, \bibinfo
  {author} {\bibfnamefont {L.}~\bibnamefont {Daudet}}, \bibinfo {author}
  {\bibfnamefont {M.}~\bibnamefont {Schuld}}, \bibinfo {author} {\bibfnamefont
  {N.}~\bibnamefont {Tishby}}, \bibinfo {author} {\bibfnamefont
  {L.}~\bibnamefont {Vogt-Maranto}},\ and\ \bibinfo {author} {\bibfnamefont
  {L.}~\bibnamefont {Zdeborov\'a}},\ }\bibfield  {title} {\bibinfo {title}
  {Machine learning and the physical sciences},\ }\href
  {https://link.aps.org/doi/10.1103/RevModPhys.91.045002} {\bibfield  {journal}
  {\bibinfo  {journal} {Rev. Mod. Phys.}\ }\textbf {\bibinfo {volume} {91}},\
  \bibinfo {pages} {045002} (\bibinfo {year} {2019})}\BibitemShut {NoStop}%
\bibitem [{\citenamefont {Schuld}\ and\ \citenamefont
  {Petruccione}(2021)}]{schuld2021machine}%
  \BibitemOpen
  \bibfield  {author} {\bibinfo {author} {\bibfnamefont {M.}~\bibnamefont
  {Schuld}}\ and\ \bibinfo {author} {\bibfnamefont {F.}~\bibnamefont
  {Petruccione}},\ }\href {https://doi.org/10.1007/978-3-030-83098-4} {\emph
  {\bibinfo {title} {Machine Learning with Quantum Computers}}}\ (\bibinfo
  {publisher} {Springer International Publishing},\ \bibinfo {year}
  {2021})\BibitemShut {NoStop}%
\bibitem [{\citenamefont {Cerezo}\ \emph {et~al.}(2021)\citenamefont {Cerezo}
  \emph {et~al.}}]{cerezo2021variational}%
  \BibitemOpen
  \bibfield  {author} {\bibinfo {author} {\bibfnamefont {M.}~\bibnamefont
  {Cerezo}} \emph {et~al.},\ }\bibfield  {title} {\bibinfo {title} {Variational
  quantum algorithms},\ }\href {https://doi.org/10.1038/s42254-021-00348-9}
  {\bibfield  {journal} {\bibinfo  {journal} {Nature Rev. Phys.}\ }\textbf
  {\bibinfo {volume} {3}},\ \bibinfo {pages} {625} (\bibinfo {year}
  {2021})}\BibitemShut {NoStop}%
\bibitem [{\citenamefont {Bharti}\ \emph {et~al.}(2022)\citenamefont {Bharti}
  \emph {et~al.}}]{bharti2022noisy}%
  \BibitemOpen
  \bibfield  {author} {\bibinfo {author} {\bibfnamefont {K.}~\bibnamefont
  {Bharti}} \emph {et~al.},\ }\bibfield  {title} {\bibinfo {title} {Noisy
  intermediate-scale quantum algorithms},\ }\href
  {https://link.aps.org/doi/10.1103/RevModPhys.94.015004} {\bibfield  {journal}
  {\bibinfo  {journal} {Rev. Mod. Phys.}\ }\textbf {\bibinfo {volume} {94}},\
  \bibinfo {pages} {015004} (\bibinfo {year} {2022})}\BibitemShut {NoStop}%
\bibitem [{\citenamefont {Schuld}\ \emph {et~al.}(2017)\citenamefont {Schuld},
  \citenamefont {Fingerhuth},\ and\ \citenamefont
  {Petruccione}}]{schuld2017implementing}%
  \BibitemOpen
  \bibfield  {author} {\bibinfo {author} {\bibfnamefont {M.}~\bibnamefont
  {Schuld}}, \bibinfo {author} {\bibfnamefont {M.}~\bibnamefont {Fingerhuth}},\
  and\ \bibinfo {author} {\bibfnamefont {F.}~\bibnamefont {Petruccione}},\
  }\bibfield  {title} {\bibinfo {title} {Implementing a distance-based
  classifier with a quantum interference circuit},\ }\href
  {https://doi.org/10.1209/0295-5075/119/60002} {\bibfield  {journal} {\bibinfo
   {journal} {Europhys. Lett.}\ }\textbf {\bibinfo {volume} {119}},\ \bibinfo
  {pages} {60002} (\bibinfo {year} {2017})}\BibitemShut {NoStop}%
\bibitem [{\citenamefont {Havl{\'\i}{\v c}ek}\ \emph
  {et~al.}(2019)\citenamefont {Havl{\'\i}{\v c}ek} \emph
  {et~al.}}]{havlicek2019supervised}%
  \BibitemOpen
  \bibfield  {author} {\bibinfo {author} {\bibfnamefont {V.}~\bibnamefont
  {Havl{\'\i}{\v c}ek}} \emph {et~al.},\ }\bibfield  {title} {\bibinfo {title}
  {Supervised learning with quantum-enhanced feature spaces},\ }\href
  {https://doi.org/10.1038/s41586-019-0980-2} {\bibfield  {journal} {\bibinfo
  {journal} {Nature}\ }\textbf {\bibinfo {volume} {567}},\ \bibinfo {pages}
  {209} (\bibinfo {year} {2019})}\BibitemShut {NoStop}%
\bibitem [{\citenamefont {Schuld}\ and\ \citenamefont
  {Killoran}(2019)}]{schuld2019quantum}%
  \BibitemOpen
  \bibfield  {author} {\bibinfo {author} {\bibfnamefont {M.}~\bibnamefont
  {Schuld}}\ and\ \bibinfo {author} {\bibfnamefont {N.}~\bibnamefont
  {Killoran}},\ }\bibfield  {title} {\bibinfo {title} {{Quantum machine
  learning in feature Hilbert spaces}},\ }\href
  {https://link.aps.org/doi/10.1103/PhysRevLett.122.040504} {\bibfield
  {journal} {\bibinfo  {journal} {Phys. Rev. Lett.}\ }\textbf {\bibinfo
  {volume} {122 4}},\ \bibinfo {pages} {040504} (\bibinfo {year}
  {2019})}\BibitemShut {NoStop}%
\bibitem [{\citenamefont {Benedetti}\ \emph
  {et~al.}(2019{\natexlab{a}})\citenamefont {Benedetti}, \citenamefont
  {Garcia-Pintos}, \citenamefont {Perdomo}, \citenamefont {Leyton-Ortega},
  \citenamefont {Nam},\ and\ \citenamefont
  {Perdomo-Ortiz}}]{benedetti2019generative}%
  \BibitemOpen
  \bibfield  {author} {\bibinfo {author} {\bibfnamefont {M.}~\bibnamefont
  {Benedetti}}, \bibinfo {author} {\bibfnamefont {D.}~\bibnamefont
  {Garcia-Pintos}}, \bibinfo {author} {\bibfnamefont {O.}~\bibnamefont
  {Perdomo}}, \bibinfo {author} {\bibfnamefont {V.}~\bibnamefont
  {Leyton-Ortega}}, \bibinfo {author} {\bibfnamefont {Y.}~\bibnamefont {Nam}},\
  and\ \bibinfo {author} {\bibfnamefont {A.}~\bibnamefont {Perdomo-Ortiz}},\
  }\bibfield  {title} {\bibinfo {title} {A generative modeling approach for
  benchmarking and training shallow quantum circuits},\ }\href
  {https://doi.org/10.1038/s41534-019-0157-8} {\bibfield  {journal} {\bibinfo
  {journal} {npj Quant. Inf.}\ }\textbf {\bibinfo {volume} {5}},\ \bibinfo
  {pages} {45} (\bibinfo {year} {2019}{\natexlab{a}})}\BibitemShut {NoStop}%
\bibitem [{\citenamefont {P{\'{e}}rez-Salinas}\ \emph
  {et~al.}(2020)\citenamefont {P{\'{e}}rez-Salinas} \emph
  {et~al.}}]{perezsalinas2020reuploading}%
  \BibitemOpen
  \bibfield  {author} {\bibinfo {author} {\bibfnamefont {A.}~\bibnamefont
  {P{\'{e}}rez-Salinas}} \emph {et~al.},\ }\bibfield  {title} {\bibinfo {title}
  {Data re-uploading for a universal quantum classifier},\ }\href
  {https://doi.org/10.22331/q-2020-02-06-226} {\bibfield  {journal} {\bibinfo
  {journal} {Quantum}\ }\textbf {\bibinfo {volume} {4}},\ \bibinfo {pages}
  {226} (\bibinfo {year} {2020})}\BibitemShut {NoStop}%
\bibitem [{\citenamefont {Lloyd}\ \emph {et~al.}(2020)\citenamefont {Lloyd},
  \citenamefont {Schuld}, \citenamefont {Ijaz}, \citenamefont {Izaac},\ and\
  \citenamefont {Killoran}}]{lloyd2020quantum}%
  \BibitemOpen
  \bibfield  {author} {\bibinfo {author} {\bibfnamefont {S.}~\bibnamefont
  {Lloyd}}, \bibinfo {author} {\bibfnamefont {M.}~\bibnamefont {Schuld}},
  \bibinfo {author} {\bibfnamefont {A.}~\bibnamefont {Ijaz}}, \bibinfo {author}
  {\bibfnamefont {J.}~\bibnamefont {Izaac}},\ and\ \bibinfo {author}
  {\bibfnamefont {N.}~\bibnamefont {Killoran}},\ }\bibfield  {title} {\bibinfo
  {title} {Quantum embeddings for machine learning},\ }\href
  {https://arxiv.org/abs/2001.03622} {\bibfield  {journal} {\bibinfo  {journal}
  {arXiv:2001.03622}\ } (\bibinfo {year} {2020})}\BibitemShut {NoStop}%
\bibitem [{\citenamefont {Hubregtsen}\ \emph {et~al.}(2022)\citenamefont
  {Hubregtsen}, \citenamefont {Wierichs}, \citenamefont {Gil-Fuster},
  \citenamefont {Derks}, \citenamefont {Faehrmann},\ and\ \citenamefont
  {Meyer}}]{hubregtsen2021training}%
  \BibitemOpen
  \bibfield  {author} {\bibinfo {author} {\bibfnamefont {T.}~\bibnamefont
  {Hubregtsen}}, \bibinfo {author} {\bibfnamefont {D.}~\bibnamefont
  {Wierichs}}, \bibinfo {author} {\bibfnamefont {E.}~\bibnamefont
  {Gil-Fuster}}, \bibinfo {author} {\bibfnamefont {P.-J. H.~S.}\ \bibnamefont
  {Derks}}, \bibinfo {author} {\bibfnamefont {P.~K.}\ \bibnamefont
  {Faehrmann}},\ and\ \bibinfo {author} {\bibfnamefont {J.~J.}\ \bibnamefont
  {Meyer}},\ }\bibfield  {title} {\bibinfo {title} {Training quantum embedding
  kernels on near-term quantum computers},\ }\href
  {https://doi.org/10.1103/physreva.106.042431} {\bibfield  {journal} {\bibinfo
   {journal} {Phys. Rev. A}\ }\textbf {\bibinfo {volume} {106}},\ \bibinfo
  {pages} {04243} (\bibinfo {year} {2022})}\BibitemShut {NoStop}%
\bibitem [{\citenamefont {Benedetti}\ \emph
  {et~al.}(2019{\natexlab{b}})\citenamefont {Benedetti}, \citenamefont {Lloyd},
  \citenamefont {Sack},\ and\ \citenamefont
  {Fiorentini}}]{benedetti2019parametrized}%
  \BibitemOpen
  \bibfield  {author} {\bibinfo {author} {\bibfnamefont {M.}~\bibnamefont
  {Benedetti}}, \bibinfo {author} {\bibfnamefont {E.}~\bibnamefont {Lloyd}},
  \bibinfo {author} {\bibfnamefont {S.}~\bibnamefont {Sack}},\ and\ \bibinfo
  {author} {\bibfnamefont {M.}~\bibnamefont {Fiorentini}},\ }\bibfield  {title}
  {\bibinfo {title} {Parameterized quantum circuits as machine learning
  models},\ }\href {https://doi.org/10.1088/2058-9565/ab4eb5} {\bibfield
  {journal} {\bibinfo  {journal} {Quantum Sc. Tech.}\ }\textbf {\bibinfo
  {volume} {4}},\ \bibinfo {pages} {043001} (\bibinfo {year}
  {2019}{\natexlab{b}})}\BibitemShut {NoStop}%
\bibitem [{\citenamefont {Schuld}(2021)}]{schuld2021kernels}%
  \BibitemOpen
  \bibfield  {author} {\bibinfo {author} {\bibfnamefont {M.}~\bibnamefont
  {Schuld}},\ }\bibfield  {title} {\bibinfo {title} {Quantum machine learning
  models are kernel methods},\ }\href {https://arxiv.org/abs/2101.11020}
  {\bibfield  {journal} {\bibinfo  {journal} {arXiv:2101.11020}\ } (\bibinfo
  {year} {2021})}\BibitemShut {NoStop}%
\bibitem [{\citenamefont {Jerbi}\ \emph {et~al.}(2023)\citenamefont {Jerbi},
  \citenamefont {Fiderer}, \citenamefont {Nautrup}, \citenamefont {K\"{u}bler},
  \citenamefont {Briegel},\ and\ \citenamefont {Dunjko}}]{jerbi2023beyond}%
  \BibitemOpen
  \bibfield  {author} {\bibinfo {author} {\bibfnamefont {S.}~\bibnamefont
  {Jerbi}}, \bibinfo {author} {\bibfnamefont {L.~J.}\ \bibnamefont {Fiderer}},
  \bibinfo {author} {\bibfnamefont {H.~P.}\ \bibnamefont {Nautrup}}, \bibinfo
  {author} {\bibfnamefont {J.~M.}\ \bibnamefont {K\"{u}bler}}, \bibinfo
  {author} {\bibfnamefont {H.~J.}\ \bibnamefont {Briegel}},\ and\ \bibinfo
  {author} {\bibfnamefont {V.}~\bibnamefont {Dunjko}},\ }\bibfield  {title}
  {\bibinfo {title} {Quantum machine learning beyond kernel methods},\ }\href
  {https://doi.org/10.1038/s41467-023-36159-y} {\bibfield  {journal} {\bibinfo
  {journal} {Nature Commun.}\ }\textbf {\bibinfo {volume} {14}} (\bibinfo
  {year} {2023})}\BibitemShut {NoStop}%
\bibitem [{\citenamefont {Altares-L{\'{o}}pez}\ \emph
  {et~al.}(2021)\citenamefont {Altares-L{\'{o}}pez}, \citenamefont {Ribeiro},\
  and\ \citenamefont {Garc{\'{\i}}a-Ripoll}}]{altares2021automatic}%
  \BibitemOpen
  \bibfield  {author} {\bibinfo {author} {\bibfnamefont {S.}~\bibnamefont
  {Altares-L{\'{o}}pez}}, \bibinfo {author} {\bibfnamefont {A.}~\bibnamefont
  {Ribeiro}},\ and\ \bibinfo {author} {\bibfnamefont {J.~J.}\ \bibnamefont
  {Garc{\'{\i}}a-Ripoll}},\ }\bibfield  {title} {\bibinfo {title} {Automatic
  design of quantum feature maps},\ }\href
  {https://doi.org/10.1088/2058-9565/ac1ab1} {\bibfield  {journal} {\bibinfo
  {journal} {Quantum Science and Technology}\ }\textbf {\bibinfo {volume}
  {6}},\ \bibinfo {pages} {045015} (\bibinfo {year} {2021})}\BibitemShut
  {NoStop}%
\bibitem [{\citenamefont {Gyurik}\ and\ \citenamefont
  {Dunjko}(2023)}]{gyurik2021structural}%
  \BibitemOpen
  \bibfield  {author} {\bibinfo {author} {\bibfnamefont {C.}~\bibnamefont
  {Gyurik}}\ and\ \bibinfo {author} {\bibfnamefont {V.}~\bibnamefont
  {Dunjko}},\ }\bibfield  {title} {\bibinfo {title} {Structural risk
  minimization for quantum linear classifiers},\ }\href
  {https://doi.org/10.22331/q-2023-01-13-893} {\bibfield  {journal} {\bibinfo
  {journal} {Quantum}\ }\textbf {\bibinfo {volume} {7}},\ \bibinfo {pages}
  {893} (\bibinfo {year} {2023})}\BibitemShut {NoStop}%
\bibitem [{\citenamefont {Landman}\ \emph {et~al.}(2022)\citenamefont
  {Landman}, \citenamefont {Thabet}, \citenamefont {Dalyac}, \citenamefont
  {Mhiri},\ and\ \citenamefont {Kashefi}}]{landman2022classically}%
  \BibitemOpen
  \bibfield  {author} {\bibinfo {author} {\bibfnamefont {J.}~\bibnamefont
  {Landman}}, \bibinfo {author} {\bibfnamefont {S.}~\bibnamefont {Thabet}},
  \bibinfo {author} {\bibfnamefont {C.}~\bibnamefont {Dalyac}}, \bibinfo
  {author} {\bibfnamefont {H.}~\bibnamefont {Mhiri}},\ and\ \bibinfo {author}
  {\bibfnamefont {E.}~\bibnamefont {Kashefi}},\ }\bibfield  {title} {\bibinfo
  {title} {{Classically approximating variational quantum machine learning with
  random Fourier features}},\ }\href {https://arxiv.org/abs/2210.13200}
  {\bibfield  {journal} {\bibinfo  {journal} {arXiv:2210.13200}\ } (\bibinfo
  {year} {2022})}\BibitemShut {NoStop}%
\bibitem [{\citenamefont {Shin}\ \emph {et~al.}(2023)\citenamefont {Shin},
  \citenamefont {Teo},\ and\ \citenamefont {Jeong}}]{shin2023analyzing}%
  \BibitemOpen
  \bibfield  {author} {\bibinfo {author} {\bibfnamefont {S.}~\bibnamefont
  {Shin}}, \bibinfo {author} {\bibfnamefont {Y.~S.}\ \bibnamefont {Teo}},\ and\
  \bibinfo {author} {\bibfnamefont {H.}~\bibnamefont {Jeong}},\ }\bibfield
  {title} {\bibinfo {title} {Analyzing quantum machine learning using tensor
  network},\ }\href {https://arxiv.org/abs/2307.06937} {\bibfield  {journal}
  {\bibinfo  {journal} {arXiv:2307.06937}\ } (\bibinfo {year}
  {2023})}\BibitemShut {NoStop}%
\bibitem [{\citenamefont {Sweke}\ \emph {et~al.}(2023)\citenamefont {Sweke},
  \citenamefont {Recio}, \citenamefont {Jerbi}, \citenamefont {Gil-Fuster},
  \citenamefont {Fuller}, \citenamefont {Eisert},\ and\ \citenamefont
  {Meyer}}]{Sweke2023Potential}%
  \BibitemOpen
  \bibfield  {author} {\bibinfo {author} {\bibfnamefont {R.}~\bibnamefont
  {Sweke}}, \bibinfo {author} {\bibfnamefont {E.}~\bibnamefont {Recio}},
  \bibinfo {author} {\bibfnamefont {S.}~\bibnamefont {Jerbi}}, \bibinfo
  {author} {\bibfnamefont {E.}~\bibnamefont {Gil-Fuster}}, \bibinfo {author}
  {\bibfnamefont {B.}~\bibnamefont {Fuller}}, \bibinfo {author} {\bibfnamefont
  {J.}~\bibnamefont {Eisert}},\ and\ \bibinfo {author} {\bibfnamefont {J.~J.}\
  \bibnamefont {Meyer}},\ }\bibfield  {title} {\bibinfo {title} {Potential and
  limitations of random fourier features for dequantizing quantum machine
  learning},\ }\href {https://arxiv.org/abs/2309.11647} {\bibfield  {journal}
  {\bibinfo  {journal} {arXiv:2309.11647}\ } (\bibinfo {year}
  {2023})}\BibitemShut {NoStop}%
\bibitem [{\citenamefont {Huang}\ \emph {et~al.}(2021)\citenamefont {Huang},
  \citenamefont {Broughton}, \citenamefont {Mohseni}, \citenamefont {Babbush},
  \citenamefont {Boixo}, \citenamefont {Neven},\ and\ \citenamefont
  {McClean}}]{huang2021power}%
  \BibitemOpen
  \bibfield  {author} {\bibinfo {author} {\bibfnamefont {H.-Y.}\ \bibnamefont
  {Huang}}, \bibinfo {author} {\bibfnamefont {M.}~\bibnamefont {Broughton}},
  \bibinfo {author} {\bibfnamefont {M.}~\bibnamefont {Mohseni}}, \bibinfo
  {author} {\bibfnamefont {R.}~\bibnamefont {Babbush}}, \bibinfo {author}
  {\bibfnamefont {S.}~\bibnamefont {Boixo}}, \bibinfo {author} {\bibfnamefont
  {H.}~\bibnamefont {Neven}},\ and\ \bibinfo {author} {\bibfnamefont {J.~R.}\
  \bibnamefont {McClean}},\ }\bibfield  {title} {\bibinfo {title} {Power of
  data in quantum machine learning},\ }\href
  {https://doi.org/10.1038%2Fs41467-021-22539-9} {\bibfield  {journal}
  {\bibinfo  {journal} {Nature Commun.}\ }\textbf {\bibinfo {volume} {12}},\
  \bibinfo {pages} {2631} (\bibinfo {year} {2021})}\BibitemShut {NoStop}%
\bibitem [{\citenamefont {Suzuki}\ \emph {et~al.}(2022)\citenamefont {Suzuki},
  \citenamefont {Kawaguchi},\ and\ \citenamefont
  {Yamamoto}}]{suzuki2022quantum}%
  \BibitemOpen
  \bibfield  {author} {\bibinfo {author} {\bibfnamefont {Y.}~\bibnamefont
  {Suzuki}}, \bibinfo {author} {\bibfnamefont {H.}~\bibnamefont {Kawaguchi}},\
  and\ \bibinfo {author} {\bibfnamefont {N.}~\bibnamefont {Yamamoto}},\
  }\bibfield  {title} {\bibinfo {title} {{Quantum Fisher kernel for mitigating
  the vanishing similarity issue}},\ }\href {https://arxiv.org/abs/2210.16581}
  {\bibfield  {journal} {\bibinfo  {journal} {arXiv:2210.16581}\ } (\bibinfo
  {year} {2022})}\BibitemShut {NoStop}%
\bibitem [{\citenamefont {Mengoni}\ and\ \citenamefont
  {Pierro}(2019)}]{mengoni2019kernel}%
  \BibitemOpen
  \bibfield  {author} {\bibinfo {author} {\bibfnamefont {R.}~\bibnamefont
  {Mengoni}}\ and\ \bibinfo {author} {\bibfnamefont {A.~D.}\ \bibnamefont
  {Pierro}},\ }\bibfield  {title} {\bibinfo {title} {Kernel methods in quantum
  machine learning},\ }\href {https://doi.org/10.1007/s42484-019-00007-4}
  {\bibfield  {journal} {\bibinfo  {journal} {Quant. Mach. Int.}\ }\textbf
  {\bibinfo {volume} {1}},\ \bibinfo {pages} {65} (\bibinfo {year}
  {2019})}\BibitemShut {NoStop}%
\bibitem [{\citenamefont {Schuld}\ \emph {et~al.}(2021)\citenamefont {Schuld},
  \citenamefont {Sweke},\ and\ \citenamefont {Meyer}}]{schuld2021fourier}%
  \BibitemOpen
  \bibfield  {author} {\bibinfo {author} {\bibfnamefont {M.}~\bibnamefont
  {Schuld}}, \bibinfo {author} {\bibfnamefont {R.}~\bibnamefont {Sweke}},\ and\
  \bibinfo {author} {\bibfnamefont {J.~J.}\ \bibnamefont {Meyer}},\ }\bibfield
  {title} {\bibinfo {title} {Effect of data encoding on the expressive power of
  variational quantum-machine-learning models},\ }\href
  {https://doi.org/10.1103/physreva.103.032430} {\bibfield  {journal} {\bibinfo
   {journal} {Phys. Rev. A}\ }\textbf {\bibinfo {volume} {103}},\ \bibinfo
  {pages} {032430} (\bibinfo {year} {2021})}\BibitemShut {NoStop}%
\bibitem [{\citenamefont {Caro}\ \emph {et~al.}(2021)\citenamefont {Caro},
  \citenamefont {Gil-Fuster}, \citenamefont {Meyer}, \citenamefont {Eisert},\
  and\ \citenamefont {Sweke}}]{caro2021encoding}%
  \BibitemOpen
  \bibfield  {author} {\bibinfo {author} {\bibfnamefont {M.~C.}\ \bibnamefont
  {Caro}}, \bibinfo {author} {\bibfnamefont {E.}~\bibnamefont {Gil-Fuster}},
  \bibinfo {author} {\bibfnamefont {J.~J.}\ \bibnamefont {Meyer}}, \bibinfo
  {author} {\bibfnamefont {J.}~\bibnamefont {Eisert}},\ and\ \bibinfo {author}
  {\bibfnamefont {R.}~\bibnamefont {Sweke}},\ }\bibfield  {title} {\bibinfo
  {title} {Encoding-dependent generalization bounds for parametrized quantum
  circuits},\ }\href {https://doi.org/10.22331/q-2021-11-17-582} {\bibfield
  {journal} {\bibinfo  {journal} {Quantum}\ }\textbf {\bibinfo {volume} {5}},\
  \bibinfo {pages} {582} (\bibinfo {year} {2021})}\BibitemShut {NoStop}%
\bibitem [{\citenamefont {Liu}\ \emph {et~al.}(2021)\citenamefont {Liu},
  \citenamefont {Arunachalam},\ and\ \citenamefont
  {Temme}}]{Liu2021discretelog}%
  \BibitemOpen
  \bibfield  {author} {\bibinfo {author} {\bibfnamefont {Y.}~\bibnamefont
  {Liu}}, \bibinfo {author} {\bibfnamefont {S.}~\bibnamefont {Arunachalam}},\
  and\ \bibinfo {author} {\bibfnamefont {K.}~\bibnamefont {Temme}},\ }\bibfield
   {title} {\bibinfo {title} {A rigorous and robust quantum speed-up in
  supervised machine learning},\ }\href
  {https://doi.org/10.1038\%2Fs41567-021-01287-z} {\bibfield  {journal}
  {\bibinfo  {journal} {Nature Phys.}\ }\textbf {\bibinfo {volume} {17}},\
  \bibinfo {pages} {1013} (\bibinfo {year} {2021})}\BibitemShut {NoStop}%
\bibitem [{\citenamefont {Micchelli}\ \emph {et~al.}(2006)\citenamefont
  {Micchelli}, \citenamefont {Xu},\ and\ \citenamefont
  {Zhang}}]{micchelli2006universal}%
  \BibitemOpen
  \bibfield  {author} {\bibinfo {author} {\bibfnamefont {C.~A.}\ \bibnamefont
  {Micchelli}}, \bibinfo {author} {\bibfnamefont {Y.}~\bibnamefont {Xu}},\ and\
  \bibinfo {author} {\bibfnamefont {H.}~\bibnamefont {Zhang}},\ }\bibfield
  {title} {\bibinfo {title} {Universal kernels},\ }\href
  {https://www.jmlr.org/papers/v7/micchelli06a.html} {\bibfield  {journal}
  {\bibinfo  {journal} {J. Mach. Learn. Res.}\ }\textbf {\bibinfo {volume}
  {7}},\ \bibinfo {pages} {2651} (\bibinfo {year} {2006})}\BibitemShut
  {NoStop}%
\bibitem [{\citenamefont {Harrow}\ \emph {et~al.}(2009)\citenamefont {Harrow},
  \citenamefont {Hassidim},\ and\ \citenamefont {Lloyd}}]{harrow2009quantum}%
  \BibitemOpen
  \bibfield  {author} {\bibinfo {author} {\bibfnamefont {A.~W.}\ \bibnamefont
  {Harrow}}, \bibinfo {author} {\bibfnamefont {A.}~\bibnamefont {Hassidim}},\
  and\ \bibinfo {author} {\bibfnamefont {S.}~\bibnamefont {Lloyd}},\ }\bibfield
   {title} {\bibinfo {title} {Quantum algorithm for linear systems of
  equations},\ }\href {https://doi.org/10.1103/physrevlett.103.150502}
  {\bibfield  {journal} {\bibinfo  {journal} {Phys. Rev. Lett.}\ }\textbf
  {\bibinfo {volume} {103}},\ \bibinfo {pages} {150502} (\bibinfo {year}
  {2009})}\BibitemShut {NoStop}%
\bibitem [{\citenamefont {Rebentrost}\ \emph {et~al.}(2014)\citenamefont
  {Rebentrost}, \citenamefont {Mohseni},\ and\ \citenamefont
  {Lloyd}}]{rebentrost2014quantum}%
  \BibitemOpen
  \bibfield  {author} {\bibinfo {author} {\bibfnamefont {P.}~\bibnamefont
  {Rebentrost}}, \bibinfo {author} {\bibfnamefont {M.}~\bibnamefont
  {Mohseni}},\ and\ \bibinfo {author} {\bibfnamefont {S.}~\bibnamefont
  {Lloyd}},\ }\bibfield  {title} {\bibinfo {title} {Quantum support vector
  machine for big data classification},\ }\href
  {https://doi.org/10.1103/physrevlett.113.130503} {\bibfield  {journal}
  {\bibinfo  {journal} {Phys. Rev. Lett.}\ }\textbf {\bibinfo {volume} {113}},\
  \bibinfo {pages} {130503} (\bibinfo {year} {2014})}\BibitemShut {NoStop}%
\bibitem [{\citenamefont {Sweke}\ \emph {et~al.}(2021)\citenamefont {Sweke},
  \citenamefont {Seifert}, \citenamefont {Hangleiter},\ and\ \citenamefont
  {Eisert}}]{sweke2021quantum}%
  \BibitemOpen
  \bibfield  {author} {\bibinfo {author} {\bibfnamefont {R.}~\bibnamefont
  {Sweke}}, \bibinfo {author} {\bibfnamefont {J.-P.}\ \bibnamefont {Seifert}},
  \bibinfo {author} {\bibfnamefont {D.}~\bibnamefont {Hangleiter}},\ and\
  \bibinfo {author} {\bibfnamefont {J.}~\bibnamefont {Eisert}},\ }\bibfield
  {title} {\bibinfo {title} {On the quantum versus classical learnability of
  discrete distributions},\ }\href {https://doi.org/10.22331/q-2021-03-23-417}
  {\bibfield  {journal} {\bibinfo  {journal} {Quantum}\ }\textbf {\bibinfo
  {volume} {5}},\ \bibinfo {pages} {417} (\bibinfo {year} {2021})}\BibitemShut
  {NoStop}%
\bibitem [{\citenamefont {Thanasilp}\ \emph {et~al.}(2022)\citenamefont
  {Thanasilp}, \citenamefont {Wang}, \citenamefont {Cerezo},\ and\
  \citenamefont {Holmes}}]{thanasilp2022exponential}%
  \BibitemOpen
  \bibfield  {author} {\bibinfo {author} {\bibfnamefont {S.}~\bibnamefont
  {Thanasilp}}, \bibinfo {author} {\bibfnamefont {S.}~\bibnamefont {Wang}},
  \bibinfo {author} {\bibfnamefont {M.}~\bibnamefont {Cerezo}},\ and\ \bibinfo
  {author} {\bibfnamefont {Z.}~\bibnamefont {Holmes}},\ }\bibfield  {title}
  {\bibinfo {title} {Exponential concentration and untrainability in quantum
  kernel methods},\ }\href {https://arxiv.org/abs/2208.11060} {\bibfield
  {journal} {\bibinfo  {journal} {arXiv:2208.11060}\ } (\bibinfo {year}
  {2022})}\BibitemShut {NoStop}%
\bibitem [{\citenamefont {McClean}\ \emph {et~al.}(2018)\citenamefont
  {McClean}, \citenamefont {Boixo}, \citenamefont {Smelyanskiy}, \citenamefont
  {Babbush},\ and\ \citenamefont {Neven}}]{mcclean2018barren}%
  \BibitemOpen
  \bibfield  {author} {\bibinfo {author} {\bibfnamefont {J.~R.}\ \bibnamefont
  {McClean}}, \bibinfo {author} {\bibfnamefont {S.}~\bibnamefont {Boixo}},
  \bibinfo {author} {\bibfnamefont {V.~N.}\ \bibnamefont {Smelyanskiy}},
  \bibinfo {author} {\bibfnamefont {R.}~\bibnamefont {Babbush}},\ and\ \bibinfo
  {author} {\bibfnamefont {H.}~\bibnamefont {Neven}},\ }\bibfield  {title}
  {\bibinfo {title} {Barren plateaus in quantum neural network training
  landscapes},\ }\href {https://doi.org/10.1038/s41467-018-07090-4} {\bibfield
  {journal} {\bibinfo  {journal} {Nature Commun.}\ }\textbf {\bibinfo {volume}
  {9}},\ \bibinfo {pages} {4812} (\bibinfo {year} {2018})}\BibitemShut
  {NoStop}%
\bibitem [{\citenamefont {Schreiber}\ \emph {et~al.}(2023)\citenamefont
  {Schreiber}, \citenamefont {Eisert},\ and\ \citenamefont
  {Meyer}}]{schreiber2022classical}%
  \BibitemOpen
  \bibfield  {author} {\bibinfo {author} {\bibfnamefont {F.~J.}\ \bibnamefont
  {Schreiber}}, \bibinfo {author} {\bibfnamefont {J.}~\bibnamefont {Eisert}},\
  and\ \bibinfo {author} {\bibfnamefont {J.~J.}\ \bibnamefont {Meyer}},\
  }\bibfield  {title} {\bibinfo {title} {Classical surrogates for quantum
  learning models},\ }\href {https://doi.org/10.1103/PhysRevLett.131.100803}
  {\bibfield  {journal} {\bibinfo  {journal} {Phys. Rev. Lett.}\ }\textbf
  {\bibinfo {volume} {131}},\ \bibinfo {pages} {100803} (\bibinfo {year}
  {2023})}\BibitemShut {NoStop}%
\bibitem [{\citenamefont {Yamasaki}\ \emph {et~al.}(2020)\citenamefont
  {Yamasaki}, \citenamefont {Subramanian}, \citenamefont {Sonoda},\ and\
  \citenamefont {Koashi}}]{yamasaki2020learning}%
  \BibitemOpen
  \bibfield  {author} {\bibinfo {author} {\bibfnamefont {H.}~\bibnamefont
  {Yamasaki}}, \bibinfo {author} {\bibfnamefont {S.}~\bibnamefont
  {Subramanian}}, \bibinfo {author} {\bibfnamefont {S.}~\bibnamefont
  {Sonoda}},\ and\ \bibinfo {author} {\bibfnamefont {M.}~\bibnamefont
  {Koashi}},\ }\bibfield  {title} {\bibinfo {title} {Learning with optimized
  random features: Exponential speedup by quantum machine learning without
  sparsity and low-rank assumptions},\ }in\ \href
  {https://dl.acm.org/doi/10.5555/3495724.3496871} {\emph {\bibinfo {booktitle}
  {Proceedings of the 34th International Conference on Neural Information
  Processing Systems}}},\ \bibinfo {series and number} {NIPS'20}\ (\bibinfo
  {year} {2020})\BibitemShut {NoStop}%
\bibitem [{\citenamefont {Bach}(2017)}]{bach2017equivalence}%
  \BibitemOpen
  \bibfield  {author} {\bibinfo {author} {\bibfnamefont {F.}~\bibnamefont
  {Bach}},\ }\bibfield  {title} {\bibinfo {title} {On the equivalence between
  kernel quadrature rules and random feature expansions},\ }\href
  {http://jmlr.org/papers/v18/15-178.html} {\bibfield  {journal} {\bibinfo
  {journal} {J. Mach. Learn. Res.}\ }\textbf {\bibinfo {volume} {18}},\
  \bibinfo {pages} {1} (\bibinfo {year} {2017})}\BibitemShut {NoStop}%
\bibitem [{\citenamefont {Nakaji}\ \emph {et~al.}(2022)\citenamefont {Nakaji},
  \citenamefont {Tezuka},\ and\ \citenamefont
  {Yamamoto}}]{nakaji2022deterministic}%
  \BibitemOpen
  \bibfield  {author} {\bibinfo {author} {\bibfnamefont {K.}~\bibnamefont
  {Nakaji}}, \bibinfo {author} {\bibfnamefont {H.}~\bibnamefont {Tezuka}},\
  and\ \bibinfo {author} {\bibfnamefont {N.}~\bibnamefont {Yamamoto}},\
  }\bibfield  {title} {\bibinfo {title} {Deterministic and random features for
  large-scale quantum kernel machine},\ }\href
  {https://arxiv.org/abs/2209.01958} {\bibfield  {journal} {\bibinfo  {journal}
  {arXiv:2209.01958}\ } (\bibinfo {year} {2022})}\BibitemShut {NoStop}%
\bibitem [{\citenamefont {Rahimi}\ and\ \citenamefont
  {Recht}(2007)}]{rahimi2007random}%
  \BibitemOpen
  \bibfield  {author} {\bibinfo {author} {\bibfnamefont {A.}~\bibnamefont
  {Rahimi}}\ and\ \bibinfo {author} {\bibfnamefont {B.}~\bibnamefont {Recht}},\
  }\bibfield  {title} {\bibinfo {title} {Random features for large-scale kernel
  machines},\ }in\ \href
  {https://proceedings.neurips.cc/paper/2007/file/013a006f03dbc5392effeb8f18fda755-Paper.pdf}
  {\emph {\bibinfo {booktitle} {Adv. Neur. Inf. Proc. Sys.}}},\ Vol.~\bibinfo
  {volume} {20}\ (\bibinfo {year} {2007})\BibitemShut {NoStop}%
\bibitem [{\citenamefont {Huang}\ \emph {et~al.}(2022)\citenamefont {Huang},
  \citenamefont {Kueng}, \citenamefont {Torlai}, \citenamefont {Albert},\ and\
  \citenamefont {Preskill}}]{huang2022provably}%
  \BibitemOpen
  \bibfield  {author} {\bibinfo {author} {\bibfnamefont {H.-Y.}\ \bibnamefont
  {Huang}}, \bibinfo {author} {\bibfnamefont {R.}~\bibnamefont {Kueng}},
  \bibinfo {author} {\bibfnamefont {G.}~\bibnamefont {Torlai}}, \bibinfo
  {author} {\bibfnamefont {V.~V.}\ \bibnamefont {Albert}},\ and\ \bibinfo
  {author} {\bibfnamefont {J.}~\bibnamefont {Preskill}},\ }\bibfield  {title}
  {\bibinfo {title} {Provably efficient machine learning for quantum many-body
  problems},\ }\href {https://doi.org/10.1126/science.abk3333} {\bibfield
  {journal} {\bibinfo  {journal} {Science}\ }\textbf {\bibinfo {volume}
  {377}},\ \bibinfo {pages} {eabk3333} (\bibinfo {year} {2022})}\BibitemShut
  {NoStop}%
\bibitem [{\citenamefont {Rudin}(1990)}]{Rudin1990fourier}%
  \BibitemOpen
  \bibfield  {author} {\bibinfo {author} {\bibfnamefont {W.}~\bibnamefont
  {Rudin}},\ }\href {https://doi.org/10.1002/9781118165621} {\emph {\bibinfo
  {title} {Fourier Analysis on Groups}}}\ (\bibinfo  {publisher} {Wiley},\
  \bibinfo {year} {1990})\BibitemShut {NoStop}%
\bibitem [{\citenamefont {Schölkopf}\ and\ \citenamefont
  {Smola}(2002)}]{scholkopf2002learning}%
  \BibitemOpen
  \bibfield  {author} {\bibinfo {author} {\bibfnamefont {B.}~\bibnamefont
  {Schölkopf}}\ and\ \bibinfo {author} {\bibfnamefont {A.~J.}\ \bibnamefont
  {Smola}},\ }\href
  {https://mitpress.mit.edu/9780262536578/learning-with-kernels/} {\emph
  {\bibinfo {title} {Learning with kernels: support vector machines,
  regularization, optimization, and beyond}}},\ Adaptive computation and
  machine learning\ (\bibinfo  {publisher} {{MIT} Press},\ \bibinfo {year}
  {2002})\BibitemShut {NoStop}%
\bibitem [{\citenamefont {Troyer}\ and\ \citenamefont
  {Wiese}(2005)}]{troyer2005computational}%
  \BibitemOpen
  \bibfield  {author} {\bibinfo {author} {\bibfnamefont {M.}~\bibnamefont
  {Troyer}}\ and\ \bibinfo {author} {\bibfnamefont {U.-J.}\ \bibnamefont
  {Wiese}},\ }\bibfield  {title} {\bibinfo {title} {Computational complexity
  and fundamental limitations to fermionic quantum monte carlo simulations},\
  }\href {https://doi.org/10.1103/PhysRevLett.94.170201} {\bibfield  {journal}
  {\bibinfo  {journal} {Phys. Rev. Lett.}\ }\textbf {\bibinfo {volume} {94}},\
  \bibinfo {pages} {170201} (\bibinfo {year} {2005})}\BibitemShut {NoStop}%
\bibitem [{\citenamefont {Kalai}(2013)}]{kalai2013complexity}%
  \BibitemOpen
  \bibfield  {author} {\bibinfo {author} {\bibfnamefont {G.}~\bibnamefont
  {Kalai}},\ }\bibfield  {title} {\bibinfo {title} {The complexity of sampling
  (approximately) the fourier transform of a boolean function},\ }\href
  {{https://cstheory.stackexchange.com/questions/17431/the-complexity-of-sampling-approximately-the-fourier-transform-of-a-boolean-fu}}
  {\bibfield  {journal} {\bibinfo  {journal} {StackExchange, Theoretical
  Computer Science}\ } (\bibinfo {year} {2013})}\BibitemShut {NoStop}%
\bibitem [{\citenamefont {Schwarz}\ and\ \citenamefont {Van~den
  Nest}(2013)}]{schwarz2013simulating}%
  \BibitemOpen
  \bibfield  {author} {\bibinfo {author} {\bibfnamefont {M.}~\bibnamefont
  {Schwarz}}\ and\ \bibinfo {author} {\bibfnamefont {M.}~\bibnamefont {Van~den
  Nest}},\ }\bibfield  {title} {\bibinfo {title} {Simulating quantum circuits
  with sparse output distributions},\ }\href {https://arxiv.org/abs/1310.6749}
  {\bibfield  {journal} {\bibinfo  {journal} {arXiv:1310.6749}\ } (\bibinfo
  {year} {2013})}\BibitemShut {NoStop}%
\bibitem [{\citenamefont {Karlin}(1964)}]{karlin1964total}%
  \BibitemOpen
  \bibfield  {author} {\bibinfo {author} {\bibfnamefont {S.}~\bibnamefont
  {Karlin}},\ }\bibfield  {title} {\bibinfo {title} {Total positivity,
  absorption probabilities and applications},\ }\href
  {https://api.semanticscholar.org/CorpusID:122826125} {\bibfield  {journal}
  {\bibinfo  {journal} {Trans. Am. Math. Soc.}\ }\textbf {\bibinfo {volume}
  {111}},\ \bibinfo {pages} {33} (\bibinfo {year} {1964})}\BibitemShut
  {NoStop}%
\bibitem [{\citenamefont {Hansen}(2010)}]{hansen2010discrete}%
  \BibitemOpen
  \bibfield  {author} {\bibinfo {author} {\bibfnamefont {P.~C.}\ \bibnamefont
  {Hansen}},\ }\href {https://doi.org/10.1137/1.9780898718836} {\emph {\bibinfo
  {title} {Discrete Inverse Problems}}}\ (\bibinfo  {publisher} {Society for
  Industrial and Applied Mathematics},\ \bibinfo {year} {2010})\BibitemShut
  {NoStop}%
\bibitem [{\citenamefont {Ha}(1986)}]{ha1986eigenvalues}%
  \BibitemOpen
  \bibfield  {author} {\bibinfo {author} {\bibfnamefont {C.-W.}\ \bibnamefont
  {Ha}},\ }\bibfield  {title} {\bibinfo {title} {Eigenvalues of differentiable
  positive definite kernels},\ }\href {https://doi.org/10.1137/0517031}
  {\bibfield  {journal} {\bibinfo  {journal} {{SIAM} J. Math. Ana.}\ }\textbf
  {\bibinfo {volume} {17}},\ \bibinfo {pages} {415} (\bibinfo {year}
  {1986})}\BibitemShut {NoStop}%
\bibitem [{\citenamefont {Yaglom}(1987)}]{yaglom1987correlation}%
  \BibitemOpen
  \bibfield  {author} {\bibinfo {author} {\bibfnamefont {A.~M.}\ \bibnamefont
  {Yaglom}},\ }\href {https://link.springer.com/book/10.1007/978-1-4612-4628-2}
  {\emph {\bibinfo {title} {Correlation Theory of Stationary and Related Random
  Functions, Volume I: Basic Results}}},\ Vol.\ \bibinfo {volume} {131}\
  (\bibinfo  {publisher} {Springer},\ \bibinfo {year} {1987})\BibitemShut
  {NoStop}%
\bibitem [{\citenamefont {K{\"u}bler}\ \emph {et~al.}(2021)\citenamefont
  {K{\"u}bler}, \citenamefont {Buchholz},\ and\ \citenamefont
  {Sch{\"o}lkopf}}]{kuebler2021inductive}%
  \BibitemOpen
  \bibfield  {author} {\bibinfo {author} {\bibfnamefont {J.}~\bibnamefont
  {K{\"u}bler}}, \bibinfo {author} {\bibfnamefont {S.}~\bibnamefont
  {Buchholz}},\ and\ \bibinfo {author} {\bibfnamefont {B.}~\bibnamefont
  {Sch{\"o}lkopf}},\ }\bibfield  {title} {\bibinfo {title} {The inductive bias
  of quantum kernels},\ }\href
  {https://proceedings.neurips.cc/paper_files/paper/2021/hash/69adc1e107f7f7d035d7baf04342e1ca-Abstract.html}
  {\bibfield  {journal} {\bibinfo  {journal} {Adv. Neur. Inf. Proc. Sys.}\
  }\textbf {\bibinfo {volume} {34}},\ \bibinfo {pages} {12661} (\bibinfo {year}
  {2021})}\BibitemShut {NoStop}%
\bibitem [{\citenamefont {Peters}\ and\ \citenamefont
  {Schuld}(2022)}]{peters2022generalization}%
  \BibitemOpen
  \bibfield  {author} {\bibinfo {author} {\bibfnamefont {E.}~\bibnamefont
  {Peters}}\ and\ \bibinfo {author} {\bibfnamefont {M.}~\bibnamefont
  {Schuld}},\ }\bibfield  {title} {\bibinfo {title} {Generalization despite
  overfitting in quantum machine learning models},\ }\href
  {https://arxiv.org/abs/2209.05523} {\bibfield  {journal} {\bibinfo  {journal}
  {arXiv:2209.05523}\ } (\bibinfo {year} {2022})}\BibitemShut {NoStop}%
\bibitem [{\citenamefont {Gil-Fuster}(2022)}]{gf2022demo}%
  \BibitemOpen
  \bibfield  {author} {\bibinfo {author} {\bibfnamefont {E.}~\bibnamefont
  {Gil-Fuster}},\ }\bibfield  {title} {\bibinfo {title} {How to approximate a
  classical kernel with a quantum computer},\ }\href
  {{https://pennylane.ai/qml/demos/tutorial_classical_kernels.html}} {\bibfield
   {journal} {\bibinfo  {journal} {PennyLane demo}\ } (\bibinfo {year}
  {2022})}\BibitemShut {NoStop}%
\end{thebibliography}%

\clearpage
\onecolumngrid
\appendix

\begin{center}
\large{Supplementary Material for \\ \enquote{Quantum kernel methods beyond quantum feature maps}
}
\end{center}

\numberwithin{theorem}{section}


\section{Towards diagonalizing integral kernel operators}\label{a:mercer}

    In this section, we give analytic results on how to turn kernel functions into explicit inner products of feature maps.
    As we see in the following, the techniques used in the literature are akin to diagonalizing infinite-dimensional matrices, hence the title of this appendix.

\subsection{Diagonalizing PSD kernel functions}\label{sa:mercer}

    Here we state Mercer's theorem and two of its corollaries, which we use to prove the universality of \emph{embedding quantum kernels} (EQKs) in Theorem~\ref{thm:approx_universality}.
    These results are taken from the textbook by Sch\"olkopf and Smola~\cite{scholkopf2002learning}.

    \begin{theorem}[Mercer's Theorem~\cite{scholkopf2002learning}]\label{thm:mercer}
        Let $k\in L_\infty(\calX^2)$ be a symmetric real-valued function such that the integral operator
        \begin{align}
            (T_kf)(x) &= \int_{\calX}k(x,x')f(x')\mathrm{d}\mu(x')
        \end{align}
        is positive semidefinite.
        That means, for any $f\in L_2(\calX)$ we have
        \begin{align}
            \int_{\calX^2} k(x,x')f(x')f(x)\mathrm{d}\mu(x)\mathrm{d}\mu(x') &\geq 0.
        \end{align}
        Let $(\psi_j)_j$ be the orthonormal eigenbasis of $T_k$ with eigenvalues $\lambda_j>0$, sorted non-increasingly.
        Then
        \begin{enumerate}
            \item The sequence $(\lambda_j)_j$ has finite $1$-norm, $(\lambda_j)_j\in\ell_1$.
            \item For almost every $(x,x')\in\calX^2$ it holds
            \begin{align}\label{eq:mercer_series}
                k(x,x') &= \sum_{j=1}^{N_{\calH}}\lambda_j\psi_j(x)\psi_j(x'),
            \end{align}
            where $N_{\calH}\in\bbN\cup\{\infty\}$.
            In the case $N_{\calH}=\infty$ the series converges absolutely and uniformly for almost every $(x,x')$.
        \end{enumerate}
    \end{theorem}
    
    \begin{corollary}[Mercer kernel map~\cite{scholkopf2002learning}]\label{cor:mercer_map}
        From statement 2. in Theorem~\ref{thm:mercer} it follows that $k(x,x')$ corresponds to a dot product in $\ell_2^{N_{\calH}}$.
    \end{corollary}
    \begin{proof}
        Take $\Phi$ as a feature map
        \begin{align}
            \Phi\colon\calX &\longto\ell_2^{N_{\calH}} , \\
            x &\longmapsto\left(\sqrt{\lambda_j}\psi_j(x)\right)_{j=1,\dots,N_{\calH}}.
        \end{align}
        One then recovers Eq.~\eqref{eq:mercer_series}, which holds for almost every $x,x'\in\calX$.
    \end{proof}
    
    \begin{corollary}[Finite-dimensional Mercer kernel approximation~\cite{scholkopf2002learning}]\label{cor:finite_mercer}
        Let $k(x,x')$ be a kernel and let $N_\calH\in\bbN\cup\{\infty\}$ be the dimension of its realization as a Mercer kernel.
        For any $\varepsilon>0$, there exists $m\in\bbN$ such that even if $N_\calH=\infty$, $k$ can be approximated 
        within accuracy $\varepsilon$ as a dot product in $\bbR^m$ for almost every $x,x'\in\calX$,
        \begin{align}
            \lvert k(x,x') - \langle\Phi_m(x),\Phi_m(x')\rangle\rvert & <\varepsilon,
        \end{align}
        where we have used the definition
        \begin{align}
            \Phi_m\colon \calX &\longrightarrow \bbR^m ,\\
            x &\longmapsto \left(\sqrt{\lambda_1}\psi_1(x),\dots,\sqrt{\lambda_m}\psi_m(x)\right).
        \end{align}
    \end{corollary}
    \begin{proof}
        It follows immediately from the absolute and uniform convergence of the series in Eq.~\eqref{eq:mercer_series}.
    \end{proof}


\subsection{Towards generalizing random Fourier features to non-shift-invariant kernels}\label{a:RFF}

    In this section, we introduce again the \emph{random Fourier features} (RFF) construction 
    from Ref.~\cite{rahimi2007random}, and then point out how it could be generalized further than shift-invariant kernels. The core insight on
    which RFF are built is that kernels of stationary processes can be expressed as the Fourier transform of a probability density function, invoking Bochner's theorem. This Fourier transform can then be efficiently approximated by Monte Carlo methods. 
    Given a shift-invariant kernel $k(x,x')=k(x-x')$, for $x,x'\in\calX\subseteq\bbR^d$, Bochner's theorem (Theorem~\ref{thm:Bochner}) states that there exists a distribution $p(\omega)$ such that
    \begin{align}
        k(x-x') &= \int_{\bbR^d} p(\omega) e^{i\langle\omega, x-x'\rangle}\mathrm{d}\omega.
    \end{align}
    Taking inspiration from measure theory, we can rewrite this integral as an expectation value
    \begin{align}
        k(x-x') &= \bbE_{\omega\sim p} \left[ e^{i\langle\omega,x\rangle}e^{-i\langle\omega,x'\rangle}\right].
    \end{align}
    Said otherwise, there is a probabilistic map $x\mapsto\phi_\omega(x)$ such that
    \begin{align}
        k(x-x') &= \bbE_{\omega\sim p}\langle\phi_\omega(x),\phi_\omega(x')\rangle.
    \end{align}
    Next, we can approximate the expectation value by its finite-sample version (which is the same as doing Monte Carlo estimation for the integral above), to get
    \begin{align}
        \bbE_{\omega\sim p} \langle\phi_\omega(x),\phi_\omega(x')\rangle &\approx \frac{1}{D}\sum_{j=1}^D e^{i\langle\omega_j,x\rangle}e^{-i\langle\omega_j,x'\rangle},
    \end{align}
    where $\omega_1,\ldots,\omega_D\sim p$ are individual samples from $p(\omega)$.
    If we next define $z(x)$ to be the 
    vector of $\phi_{\omega_j}(x)$ for $j\in\{1,\ldots,D\}$ normalized by $1/\sqrt{D}$, we obtain
    \begin{align}
        \langle z(x),z(x')\rangle &= \frac{1}{D}\sum_{j=1}^D e^{i\langle\omega_j,x\rangle}e^{-i\langle\omega_j,x'\rangle} \approx k(x-x').
    \end{align}
    Now we can understand why the RFF approximation works (see Theorem~\ref{thm:RFF_existence}).
    It is because of the parallelism between MC estimation of the Fourier integral and the finite-sample approximation of the expectation value.
    All of the above of course depends deeply on Bochner's theorem.

    If we are interested in non-shift-invariant kernels instead, then Bochner's theorem does not apply, and the RFF construction does not work.
    Instead, we could hope to use Yaglom's theorem~\cite{yaglom1987correlation}.
    \begin{theorem}[Yaglom~\cite{yaglom1987correlation}]\label{thm:yaglom}
        A kernel $k(x,x')$ is PSD in $\bbR^d$ if and only if it has the form
        \begin{align}
            k(x,x') &= \int_{\bbR^d\times\bbR^d} 
            e^{i\left(\langle\omega,x\rangle-\langle\omega',x'\rangle\right)}\mu\!\left(\mathrm{d}\omega,\mathrm{d}\omega'\right),
        \end{align}
        where $\mu\!\left(\mathrm{d}\omega,\mathrm{d}\omega'\right)$ is the Lebesgue-Stieltjes measure associated to some PSD function $f(\omega,\omega')$ with bounded variation.
    \end{theorem}

    The integral can again be rewritten as
    \begin{align}
        k(x,x') &= \int_{\bbR^d\times\bbR^d} f(\omega,\omega') e^{i\left(\langle\omega,x\rangle-\langle\omega',x'\rangle\right)} \mathrm{d}\omega\mathrm{d}\omega'.
    \end{align}    
    The special case where $f(\omega,\omega')=p(\omega)p(\omega')$ has symmetric product structure is the shift-invariant case.
    But in general $f$ does not need to be a probability distribution, because the requirement is that $f$ be PSD, not positive everywhere.
    If $f$ were positive everywhere $f(\omega,\omega')>0$ and normalized $\int_{\bbR^d\times\bbR^d}f(\omega,\omega')\mathrm{d}\omega\mathrm{d}\omega'=1$, then the Fourier integral would again correspond to an expectation value
    \begin{align}
        \int_{\bbR^d\times\bbR^d} f(\omega,\omega') e^{i\left(\langle\omega,x\rangle-\langle\omega',x'\rangle\right)} \mathrm{d}\omega\mathrm{d}\omega' &= \bbE_{(\omega,\omega')\sim f} \left[ e^{i\langle\omega,x\rangle}e^{-i\langle\omega',x'\rangle}\right].
    \end{align}
    This grants another unbiased estimator for the kernel as an Euclidean inner product
    \begin{align}
        k(x,x') &= \bbE_{(\omega,\omega')\sim f}\langle\phi_\omega(x),\phi_{\omega'}(x')\rangle,
    \end{align}
    only this time with different sampled frequencies for the left and right sides.
    Nevertheless, this would once again allow us to exploit the link between Monte Carlo approximation and finite sample approximation.
    Sampling $(\omega_1,\omega_1'),\ldots,(\omega_D,\omega_D')\sim f$, and defining $z$ and $z'$ to be the left and right probabilistic feature vectors, we recover the familiar formula
    \begin{align}
        \langle z(x),z'(x')\rangle &= \frac{1}{\sqrt{D}} \begin{pmatrix} e^{i\langle\omega_1,x\rangle} & \cdots & e^{i\langle\omega_D,x\rangle} \end{pmatrix} \cdot\frac{1}{\sqrt{D}}\begin{pmatrix} e^{-i\langle\omega'_1,x'\rangle} \\ \vdots \\ e^{-i\langle\omega'_D,x'_D\rangle}\end{pmatrix}\\
         &= \frac{1}{D} \sum_{j=1}^D e^{i\langle\omega_j,x\rangle}e^{-i\langle\omega'_j,x'\rangle}\\ 
         &\approx k(x,x').
    \end{align}
    A similar construction could be furbished if some notion of \emph{sampling from $f$} were available even if $f$ were not a probability distribution.
    Notice this approximation requires two different feature maps, $z$ and $z'$.
    So even if $k$ is guaranteed to be PSD by assumption, we would be using different left and right feature maps, similarly to what we described for non-PSD earlier in this appendix.
    Nevertheless, it would be interesting to find restricted cases where such an approximation is possible, as they could potentially hint towards further de-quantization protocols for QML models.


\subsection{Towards diagonalizing smooth and indefinite kernel functions}\label{a:diagonalizing}

    Earlier in this appendix we showcased Mercer's construction for kernels as explicit inner products.
    Here we add considerations from the literature~\cite{karlin1964total, hansen2010discrete, ha1986eigenvalues} that aim at quantifying the required resources of Mercer's finite-dimensional kernel approximation, as well as generalize it to indefinite (non-PSD) kernels.
    Drawing parallels to matrices, we refer to \enquote{the number of dimensions used to approximate a kernel function as an explicit inner product of feature vectors} as the \emph{rank} of said kernel approximation.
    The thesis of Mercer's theorem is that there always exist \emph{finite}-rank approximations for any kernel function.
    Yet, the relevant question is whether there always exist \emph{low}-rank approximations for kernel functions.
    The following claims indicate there always exist low-rank approximations for any kernel function on one-dimensional input data, assuming smoothness.
    These could be taken as starting point for further studying the existence of quantum kernel functions beyond embedding quantum kernels for non-shift-invariant kernels.

    Let $k:\calX\times\calX\to\bbR$ be a PSD kernel function, with $\calX\subseteq\bbR$ bounded.
    If $k$ is square-integrable 
    \begin{align}
        \int_{\calX\times\calX}\mathrm dx\,\mathrm dx' \lvert k(x,x')\rvert^2 &< \infty,
    \end{align}
    then there exist sequences $(s_l)_l\subseteq\bbR_{\geq0}$ of non-negative numbers and $(U_l(x))_l\subseteq \bbR^{\calX}$ of functions such that:
    \begin{align}
        k(x,x') &= \sum_{l=0}^\infty s_l U_l(x)U_l(x'),
    \end{align}
    where $(s_l)_l$ is monotonically decaying and has finite norm (which means its only aggregation point is $0$), and $(U_l(x))_l$ form an orthonormal basis of $\bbR^{\calX}$.

    The following is known about how quickly the eigenvalues $(s_l)_l$ must decay, which confirms the existence of low-rank approximations:
    \begin{itemize}
        \item If $k\in\calC^t[\calX\times\calX]$, then for sufficiently large $l$ we have the asymptotic scaling $s_l\in\calO\left(l^{-t-1}\right)$.
        \item If $k\in\calC^\infty[\calX\times\calX]$, then $s_l \in\calO\left(1/\exp(\Theta(l))\right)$.
        \item Finally, if $k\in\calC^\infty[\bbR^2]$ (now defined on $\calX=\bbR$), then $s_l \in o\left(1/\exp(\Theta(l))\right)$, where the \enquote{little-$o$} notation means this is an upper bound that is strictly not tight.
    \end{itemize}
    Truncating after the first $L$ terms of this series gives the optimal $L$-rank approximation to the function \emph{in $2$-norm}.

    \begin{remark}[Low rank approximation.]
        Having quickly decaying eigenvalues is a sufficient condition to ensure existence of a good low-rank approximation.
        For any orthonormal function basis $(U_l)_l$, each kernel function in $L_2$ admits an exact series expansion.
        When truncating the series to finitely-many terms, orthonormality implies that optimal approximation comes from keeping the terms with the largest eigenvalues.
        Since the sequence of eigenvalues has finite norm, there is a correspondence between the decay speed and the relative magnitude of the largest eigenvalues with respect to the smaller ones.
        For instance, if the decay is exponential, we know that keeping logarithmically-many eigenvalues suffices to ensure a constant approximation error.
    \end{remark}

    \begin{remark}[Input dimension.]
        Notice this is only for $1$-dimensional data!
        One would need to see how the statements generalize to higher dimensions.
        The case of functions that admit a separation as product of $1$-dimensional kernel functions (for instance, the Gaussian kernel) presents itself as a particularly nice one to analyze.
    \end{remark}

    The main difference when dealing with non-PSD functions is that the series expansion
    \begin{align}
        k(x,x') &= \sum_{l=0}^\infty s_l U_l(x)V_l(x')
    \end{align}
    requires two orthonormal bases of eigenfunctions $(U_l)_l,(V_l)_l$, one for the left and one for the right arguments.
    One simple example to see how the difference between the $U_l$ and the $V_l$ functions can come about is the following:
    \begin{enumerate}
        \item Let $k$ be some PSD kernel function.
        \item Let $(s_l)_l$ and $(U_l)_l$ be the sequences of positive eigenvalues and pairwise orthogonal, normalized functions.
        \item We have $k(x,x') = \sum_{l=0}^\infty s_l U_l(x)U_l(x')$, still a PSD function.
        \item Pick any $l'\in\bbN$ and consider the function $k'(x,x') = \sum_{l=0}^\infty (-1)^{\delta_{l,l'}} s_l U_l(x)U_l(x')$, which is not PSD anymore, since it has one negative eigenvalue $s_{l'}<0$.
        \item It is clear to see that $k'$ still admits a series expansion with only positive numbers $s_l$, provided we are allowed to use two different sequences of functions.
            Indeed, we only need to set $V_l$ to be
            \begin{align}
                V_l &= \begin{cases} -U_l &\text{if } l=l', \\ U_l & \text{else.} \end{cases}
            \end{align}
        \item It follows that $k'(x,x') = \sum_{l=0}^\infty s_l U_l(x)V_l(x')$, where, in particular, the sequence $(s_l)_l$ is the exact same one as we had for the original function $k$, and the only difference between the sequences $U_l$ and $V_l$ is a single minus sign in one element.
    \end{enumerate}

    We present a second and final example, which is closely related to random Fourier features.
    One of the basic facts used in the RFF-based approximation is the explicit rewriting
    \begin{align}
        \cos(x-x') &= \begin{pmatrix} \cos(x) & \sin(x)\end{pmatrix}\begin{pmatrix} \cos(x') \\ \sin(x')\end{pmatrix}.
    \end{align}
    In the language we have just introduced, $k(x,x')=\cos(x-x')$, $s_1=s_2=1$, $U_1(x)=\cos(x)$, and $U_2(x)=\sin(x)$.
    This confirms that the simple cosine function is a shift-invariant PSD kernel.
    If we use the sine instead of the cosine function, we lose positive semidefiniteness, but that does not prevent us from envisioning a similar rewriting
    \begin{align}
        \sin(x-x') &= \begin{pmatrix} \cos(x) & \sin(x)\end{pmatrix} \begin{pmatrix} -\sin(x') \\ \cos(x')\end{pmatrix}.
    \end{align}
    In this case, the identifications would be $k(x,x')=\sin(x-x')$, $s_1=s_2=1$, $U_1(x)=\cos(x)$, and $U_2(x)=\sin(x)$, all same as before, except for the necessity to introduce the right functions $V_1(x)=-U_2(x)$, and $V_2(x)=U_1(x)$.
    As we see, the vulneration of positive semi-definiteness can come in a way which does not significantly increase the difficulty of low-rank approximation.


\section{Mapping real-vectors to quantum states}\label{a:rtorho}

    Here we give the details of the quantum feature maps introduced in Sections~\ref{s:universality} and~\ref{s:shift-invariant}.
    The main feature of this quantum feature map is that it preserves inner products up to a multiplicative and an additive factor.
    That is, the Hilbert-Schmidt inner product of the quantum states relates to the Euclidean inner product of the vectors encoded in the states.
    The number of qubits required is logarithmic in the dimension of the vector to be encoded, and the resulting quantum states are mixed.
    First, we present the feature map and prove that it produces valid quantum states, that is properly normalized, trace one, PSD matrices.
    
    \begin{customlemma}{\ref{l:mappingqs}}[Correctness and runtime of Algorithm~\ref{alg:QEPIP}]
        Let $r\in\ell_1^d\subseteq\bbR^d$ be a unit vector with respect to the $1$-norm, $\lVert r\rVert_1=1$.
        Take $n=\lceil\log_4(d+1)\rceil$ and pad $r$ with $0$s until it has length $4^n-1$.
        Let $(P_i)_{i=1}^{4^n-1}$ be the set of all Pauli matrices on $n$ qubits without the identity.
        Then Algorithm~\ref{alg:QEPIP} prepares the following state as a classical mixture
        \begin{align}
            \rho_{(\cdot)}\colon\ell_1^d\to & \operatorname{Herm}(2^n) ,\\
            r \mapsto &\rho_r = \frac{\bbI + \sum_{i=1}^{4^n-1} r_i P_i}{2^n}.
        \end{align}
        The total runtime complexity $t$ of Algorithm~\ref{alg:QEPIP} fulfills $t\in\calO(\poly(d))$.
    \end{customlemma}
    \begin{proof}
    We have
    \begin{align}
        \frac{\bbI + \sum_{i=1}^{4^n-1} r_i P_i}{2^n} &= \frac{1}{{2^n}} \left({\sum_{i=1}^{4^n-1} \lvert r_i\rvert \bbI + \sum_{i=1}^{4^n-1} r_i P_i}\right) \\
        &= \frac{1}{{2^n}} \left({\sum_{i=1}^{4^n-1} \lvert r_i\rvert \bbI + \sum_{i=1}^{4^n-1} \lvert r_i\rvert \sign(r_i) P_i}\right) \\
        &= \frac{1}{{2^n}} \sum_{i=1}^{4^n-1} \lvert r_i\rvert (\bbI + \sign(r_i) P_i) \geq 0, 
    \end{align}
    making use of the fact that $\sum_i |r_i| = \lVert r\rVert_1=1$ and that $\bbI\pm P_i\geq 0$ for every Pauli operator $P_i$.
    Notice that preparing $\bbI \pm P_i$ can be done efficiently by rotating the $\lvert0\rangle$ basis state vector of each qubit to the respective Pauli basis, and then flip the necessary qubits individually.
    Since the latter is a convex combination of quantum states, this preparation is efficient by means of mixing whenever this expression has polynomially many terms, as each term is prepared separately and then mixed.
    \end{proof}

    Next, we show the relation between the Hilbert-Schmidt inner products of quantum states $\Tr\{\rho_r\rho_r'\}$ and the Euclidean inner product of the vectors they encode $\langle r,r'\rangle$.
    
    \begin{customlemma}
    {\ref{l:innerprodqs}}[Euclidean inner products]
        Let $r,r'\in\bbR^d$ be unit vectors with respect to the $1$-norm $\lVert r\Vert_1=1$, $\lVert r'\rVert=1$.
        Then, for $\rho_r,\rho_{r'}$ as constructed with Algorithm~\ref{alg:QEPIP}, the following identity holds
         \begin{align}
            \langle r,r'\rangle &= 2^n \Tr\{\rho_{r}\rho_{r'}\}-1.
        \end{align}
    \end{customlemma}
    \begin{proof}
        In order to prove this statement, we need to resort to two main ideas.
        \begin{itemize}
            \item The trace is linear, and the trace of the tensor product is the product of traces.
            \item All Pauli words are traceless except for the identity, and each Pauli word is its only own inverse. So the product of two different Pauli words is also traceless.
        \end{itemize}
        With these, we only need to expand the inner product as a sum and then manipulate each term according to the previous ideas,
        to get
        \begin{align}
            \Tr\{\rho_{r}\rho_{r'}\} &= \Tr\left\{\frac{1}{4^n} \left(\bbI + \sum_{j=1}^{4^n-1} r_jP_j\right)\left(\bbI + \sum_{k=1}^{4^n-1} r'_kP_k\right) 
            \right\} \\
            &= \frac{1}{4^n} \Tr\left\{\bbI + \sum_{j=1}^{4^n-1} r_jr'_j P_j^2 + \sum_{j\neq k}^{4^n-1} r_jr'_kP_jP_k\right\} \\
            &= \frac{1}{4^n} \left(\Tr\left\{\bbI\right\} + \Tr\left\{\sum_{j=1}^{4^n-1} r_jr'_j P_j^2\right\} +
             \Tr\left\{\sum_{j\neq k}^{4^n-1} r_jr'_kP_jP_k\right\}\right) \\
            &= \frac{1}{4^n} \left(\Tr\left\{\bbI\right\} + \Tr\left\{\sum_{j=1}^{4^n-1} r_jr'_j \bbI\right\}\right) \\
            &= \frac{1 + \sum_{j=1}^{4^n-1} r_jr'_j}{4^n} \Tr\{\bbI\} \\
            &= \frac{1 + \langle r,r'\rangle}{2^n}.
        \end{align}
        This completes the proof.
    \end{proof}
    
    \begin{customlemma}{\ref{l:2normQRFF}}[Inner product normalization]
        Let $r,r'\in\bbR^d$ be $2$-norm unit vectors $\lVert r\rVert_2 = \lVert r'\rVert_2=1$.
        Then, the  identity
         \begin{align}
            \langle r,r'\rangle &= \lVert r\rVert_1\lVert r'\rVert_1\left(2^n \Tr\left\{\rho_{\tr}\rho_{\tr'}\right\} -1
            \right)
        \end{align}
        holds, where $\tr^{(\prime)}=r^{(\prime)}/\lVert r^{(\prime)}\rVert_1$ corresponds to 
        renormalizing with respect to the $1$-norm.  
        Here $\rho_{\tr}$ refers to encoding $\tr$ onto a quantum state using the construction in Lemma~\ref{l:innerprodqs}.
        Note the renormalization is of no concern, since the fact $r,r'$ are $2$-norm unit vector implies that their $1$-norm is lower-bounded by $1$.
    \end{customlemma}
    \begin{proof}
        We prove this directly by just invoking Lemma~\ref{l:innerprodqs} using the re-normalized vectors $\tr,\tr'$, to get
        \begin{align}
            2^n \Tr\{\rho_{\tr}\rho_{\tr'}\}-1 &= \langle \tr,\tr'\rangle \\
            &= \left\langle\frac{r}{\lVert r\rVert_1},\frac{r'}{\lVert r'\rVert_1}\right\rangle \\
            &= \frac{\langle r, r'\rangle}{\lVert r\rVert_1\lVert r'\rVert_1}.
        \end{align}
        This completes the proof.
    \end{proof}    


\section{Precision required for estimating second derivatives with finite difference methods}\label{a:smoothness}
    \begin{corollary}[Finite precision derivative accuracy]\label{cor:smoothness}
        Let $\calX_d\subseteq[-R,R]^d$ be a compact domain, and let $\calX_{d,P}\subseteq\calX_d$ 
        be a subset which we can represent with up to $P$ bits of precision.
        For $\varepsilon>0$, let $k_d$ be a continuous kernel function such that it can be quantum efficiently $\varepsilon$-approximated on $\calX_{d,P}$.
        Assume the fourth derivatives of $k_d$ have bounded magnitude 
        \begin{equation}
        \left\lvert \partial^{(4)}_i k_d(\Delta)\right\rvert\leq L(d).
        \end{equation}
        Then, for any $\varepsilon>0$, the number of bits of precision required to achieve $\lvert\partial^2_i k_d - \delta^2_i k_d\rvert\leq\varepsilon$ is
        \begin{align}
            P = \left\lceil\log_4\left(\frac{L(d)}{12\varepsilon}\right)\right\rceil.
        \end{align}
    \end{corollary}
    \begin{proof}
        For the finite difference second derivative, we take the central version
        \begin{align}
            \delta^2_i k_d \coloneqq \frac{k_d(\Delta+h \he_i) + k_d(\Delta-h\he_i) - 2 k_d(\Delta)}{h^2},
        \end{align}
        where $\he_i$ is the $i^\text{th}$ basis vector.
        We drop the subscript referring to the $i^\text{th}$ second derivative for ease of notation, the following holds for any $i\in[d]$.
        Start by taking the Lagrange formulation of Taylor's theorem, for $h>0$ there exists $\xi_+\in[\Delta,\Delta+h]$ such that
        \begin{align}
            k_d(\Delta+h) &= k_d(\Delta) + \partial k_d(\Delta) h + \frac{\partial^2 k_d(\Delta)}{2} h^2 + \frac{\partial^{(3)} k_d(\Delta)}{3!} h^3 + \frac{\partial^{(4)} k_d(\xi_+)}{4!} h^4.
        \end{align}
        Next, consider the same expansion for $k_d(\Delta-h)$, which will result in a different $\xi_-\in[\Delta-h,\Delta]$, and allows us to use the following trick
        \begin{align}
            k_d(\Delta+h) &= k_d(\Delta) + \partial k_d(\Delta) h + \frac{\partial^2 k_d(\Delta)}{2} h^2 + \frac{\partial^{(3)} k_d(\Delta)}{3!} h^3 + \frac{\partial^{(4)} k_d(\xi_+)}{4!} h^4 \\
            k_d(\Delta-h) &= k_d(\Delta) - \partial k_d(\Delta) h + \frac{\partial^2 k_d(\Delta)}{2} h^2 - \frac{\partial^{(3)} k_d(\Delta)}{3!} h^3 + \frac{\partial^{(4)} k_d(\xi_-)}{4!} h^4 \\ 
            k_d(\Delta+h) + k_d(\Delta-h) &= 2 k_d(\Delta) + \partial^2 k_d(\Delta) h^2 + \frac{\partial^{(4)} k_d(\xi_+)+\partial^{(4)} k_d(\xi_-)}{4!} h^4 \\
            \frac{\partial^{(4)} k_d(\xi_+)+\partial^{(4)} k_d(\xi_-)}{4!} h^2 &= \frac{k_d(\Delta+h) + k_d(\Delta-h) - 2 k_d(\Delta)}{h^2} - \partial^2 k_d(\Delta) \\
            \frac{\partial^{(4)} k_d(\xi_+)+\partial^{(4)} k_d(\xi_-)}{4!} h^2 &= \delta^2 k_d(\Delta) - \partial^2 k_d(\Delta).   
        \end{align}
        Following our assumptions, we have
        \begin{align}
             \lvert\partial^{(4)} k_d(\xi_+)+\partial^{(4)} k_d(\xi_-)\rvert &\leq 2L(d),
        \end{align}
        then the approximation error of the finite difference partial derivative is
        \begin{align}
            \lvert \partial^2 k_d(\Delta)-\delta^2 k_d(\Delta)\rvert\leq \frac{2 L(d)}{4!} h^2.
        \end{align}
        Now, since we can represent inputs with up to $P$ bits of precision, we can afford to set the finite difference to be machine precision $h=2^{-P}$, which results in
        \begin{align}
            \lvert \partial^2 k_d(\Delta)-\delta^2 k_d(\Delta)\rvert\leq \frac{2 L(d)}{4!\,4^P}.
        \end{align}
        Finally, we set
        \begin{align}
            \frac{2 L(d)}{4!\,4^P} &\leq \varepsilon,
        \end{align}
        solve for $P$ as
        \begin{align}
            4^P &\geq \frac{2 L(d)}{4!\varepsilon} = \frac{L(d)}{12\varepsilon}, \\
            P &\geq \log_4\frac{L(d)}{12\varepsilon},
        \end{align}
        and round $P$ up to the next integer, thus completing the proof.        
    \end{proof}

    Notice in practice we are interested for second derivatives at the origin $\Delta=0$, so we need the bound on the magnitude to hold only in a small environment around $\Delta=0$.


\section{Proofs from Section~\ref{s:composition}}\label{a:variance}

     \begin{customproposition}{\ref{prop:composition}} [Performance guarantee of Algorithm~\ref{alg:RFF_PP}]
        Let $f\colon\calX\to[-B,B]^{g_1(d)}$ be a pre-processing function, and let $k_f$ be the Gaussian kernel composed with $f$, as introduced in Eq.~\eqref{eq:composition}.
        Let the parameter of the Gaussian kernel be  $\sigma=g_2(d)$.
        If $g_1(d)\in\calO(\poly(d))$ and $g_2(d)\in\Omega(\poly(d)^{-1})$, then, $k_f(x,x')$ can be $\varepsilon$-approximated efficiently in the number of feature dimensions.
        In particular, the required dimension of a (probabilistic) feature map $D$ is at most
        \begin{align}
            D &\in\calO\left(\frac{\poly(d)}{\varepsilon^2}\log\left(\frac{dB}{\varepsilon}\right)\right) \in \tcalO\left(\frac{\poly(d)}{\varepsilon^2}\right).
        \end{align}
    \end{customproposition}
    \begin{proof}
        In order to reach the claim, we only need to use the random Fourier feature approach of Algorithm~\ref{alg:RFF}, but taking 
        $f(\calX)$ as input domain for the kernel.
        Indeed, we give the Gaussian kernel with parameter $\sigma=g_2(d)$ as input to the Algorithm, and we obtain the probabilistic 
        feature map
        \begin{align}
            z\colon[-B,B]^{g_1(d)} &\to\bbR^{D}, \\
            f(x) &\mapsto z(f(x)) = \sqrt{\frac{2}{D}}\begin{pmatrix}\cos\langle\omega_1,f(x)\rangle \\ \sin\langle\omega_1,f(x)\rangle \\ \vdots \\ 
            \cos\langle\omega_{D/2},f(x)\rangle \\ \sin\langle\omega_{D/2}, f(x)\rangle\end{pmatrix}.
        \end{align}
        Now, using the identity $\sigma_p^2=d/g_2(d)$ from Lemma~\ref{l:gaussian_variance}, it follows that
        \begin{align}
            \bbP \left[\sup_{x,x'\in\calX} \lvert\langle z(f(x)), z(f(x'))\rangle - k_f(x,x')\rvert \geq\varepsilon\right] &\leq2^8\left(\sqrt{\frac{d}{g_2(d)}}\frac{2B\sqrt{g_1(d)}}{\varepsilon}\right)^2\exp\left(-\frac{D\varepsilon^2}{4(g_1(d)+2)}\right).
        \end{align}
        In turn, this guarantees $\varepsilon$ approximation with any constant success probability using $D$ features, with
        \begin{align}
            D \in\calO\left(\frac{g_1(d)}{\varepsilon^2}\log\sqrt{\frac{d}{g_2(d)}}\frac{2B\sqrt{g_1(d)}}{\varepsilon}\right).
        \end{align}
        Now, setting $g_1(d)\in\calO(\poly(d))$ and $g_2(d)\in\Omega(\poly(d)^{-1})$ we obtain $\sqrt{d g_1(d) /g_2(d)}\in\calO(\poly(d))$, resulting in
        \begin{align}
            D\in\calO\left(\frac{\poly(d)}{\varepsilon^2}\log\frac{dB}{\varepsilon}\right)\in \tcalO\left(\frac{\poly(d)}{\varepsilon^2}\right),
        \end{align}
        thus reaching the claim.
    \end{proof}
    Recall one last time the definition 
    \begin{align}
        k(\Delta) &= e^{-\frac{\lVert\Delta\rVert^2}{2\sigma^2}}
    \end{align}
    of the Gaussian kernel, with parameter $\sigma>0$.
    This is a PSD symmetric function, so a 
    valid shift-invariant kernel, and also it fulfills $k(0)=1$.
    Accordingly, Bochner's theorem ensures that the FT of this function
    \begin{align}
        \operatorname{FT}[k](\omega) &\coloneqq p(\omega) \\
        &= \left(\frac{\sigma}{\sqrt{2\pi}}\right)^D e^{-\frac{\sigma^2\lVert\omega\rVert^2}{2}}
    \end{align}
 is a probability distribution.
    One can check that $p(\omega)>0$ everywhere, and that $\int_{\bbR^D} p(\omega)\mathrm{d}\omega =1$.
    Also, one can compute the variance of the distribution $p$, which we now do.
    
    \begin{lemma}[Gaussian kernel variance]\label{l:gaussian_variance}
        Let $p$ be the inverse $\mathrm{FT}$ of the Gaussian kernel with parameter $\sigma$.
        Let $\sigma_p^2=\bbE_p[\lVert\omega\rVert^2]$ be the variance of the distribution $p$.
        It holds that $\sigma_p^2=d/\sigma$.
    \end{lemma}
    \begin{proof}
        We prove the proposition directly, plugging in the formulas.
        We need the two identities:
        \begin{align}
            \int_{-\infty}^\infty e^{-\frac{a^2 x^2}{2}} \mathrm{d}x &= \frac{\sqrt{2\pi}}{a} \quad\text{(normalization, $0$th moment)} ,\\
            \int_{-\infty}^\infty x^2 e^{-\frac{a^2 x^2}{2}} \mathrm{d}x &= \frac{\sqrt{2\pi}}{a^2} \quad\text{($2$nd moment)}.
        \end{align}
        At this point, we merely  need
        to expand $\bbE_p\left[\lVert\omega\rVert^2\right]$ as a sum, 
        and then substitute the formulas for the moments, to arrive at
        \begin{align}
            \sigma_p^2 &= \bbE_p\left[\lVert\omega\rVert^2\right] \\
            &= \sum_{i=1}^d \bbE_p\left[\omega_i^2\right] \\
            &= \sum_{i=1}^d \int_{\bbR^d} \omega_i^2 p(\omega)\mathrm d\omega \\
            &= \sum_{i=1}^d \int_{\bbR^d} \omega_i^2 \left(\frac{\sigma}{\sqrt{2\pi}}\right)^d e^{-\frac{\sigma^2\left(\omega_1^2 + \cdots + \omega_d^2\right)}{2}} \mathrm d\omega \\
            &= \sum_{i=1}^d \left(\frac{\sigma}{\sqrt{2\pi}}\right)^d \int_{-\infty}^\infty \omega_i^2 e^{-\frac{\sigma^2\omega_i^2}{2}} \mathrm d\omega_i \prod_{j\neq i}^d \int_{-\infty}^\infty e^{-\frac{\sigma^2\omega_j^2}{2}}\mathrm d\omega_j \\ 
            &= \sum_{i=1}^d \left(\frac{\sigma}{\sqrt{2\pi}}\right)^d \frac{\sqrt{2\pi}}{\sigma^2} \prod_{j\neq i}^d \frac{\sqrt{2\pi}}{\sigma} \quad\text{(from normalization and $2$nd moment)} \\
            &= \sum_{i=1}^d \frac{1}{\sigma} \\
            &= \frac{d}{\sigma}. 
        \end{align}
        The proof is complete.
    \end{proof}

    \begin{customproposition}{\ref{prop:f_hardness}}[No efficient classical approximation]
        There exists a function $f\colon\calX\to[0,1]^d$ which can be $\varepsilon$-approximated quantum efficiently in $d$ for which the composition kernel $k_f$ cannot be $\varepsilon$-approximated classically efficiently in $d$, with
        \begin{align}
            k_f(x,x') &\coloneqq \exp\left(-\lVert f(x)-f(x')\rVert^2\right).
        \end{align}
        We have taken $\sigma^2=1/2$ for simplicity.
    \end{customproposition}
    \begin{proof}
        We prove the statement constructively by giving an example of such a function.
        
        Let $f=(f_1, f_2, \ldots, f_d)$ be a vector of PQC-based functions $f_i\colon\calX\to[0,1]$ for $i\in\{1,2, \ldots, d\}$ with the following assumptions:
        \begin{enumerate}
            \item Each function $f_i$ requires a circuit with $d$ qubits and $d$ non-commuting layers of gates to be $\varepsilon$-approximated, such that evaluating $f_i(x)$ is quantum efficient but classically inefficient for every $i\in\{1,\ldots,d\}$.
            \item There exists a subset $\calX_1\subseteq\calX$ such that, for every $x\in\calX_1$, $f(x)=(f_1(x), 0, \ldots, 0)$.
            \item There exists a subset $\calX_2\subseteq\calX$ such that, for every $x\in\calX_2$, $f(x)=(0, 0, \ldots, 0)$.
        \end{enumerate}
        
        Such a function would fulfill the hypothesis of the proposition, and now we show that it also fulfills the thesis.
        That is, we prove that evaluating the composition of this function with the 
        Gaussian kernel also requires a quantum computer.
        We reach the conclusion by \emph{reductio ab absurdum}: We 
        assume $k$ is not hard and reach a contradiction (namely, that $f_1$ is not hard).
        In particular, if we could evaluate $k$ on a classical computer efficiently, we show that we would also be able to evaluate $f_1$ on a classical computer efficiently.
        
        We take $x_1\in\calX_1$ and $x_2\in\calX_2$, then we have
        \begin{align}
            k(x_1,x_2) &\coloneqq \exp\left(-\lVert f(x_1)-f(x_2)\rVert^2\right) = \exp\left(-\lVert(f_1(x_1), 0, \ldots, 0) - (0, \ldots, 0) \rVert^2\right) \\
            &= \exp\left(-\lvert f_1(x_1)\rvert^2\right) = \exp\left(-f_1(x_1)^2\right).
        \end{align}
        In the last line, we
        have used the fact that the range of all $f_i$ functions comprises only non-negative numbers.
        Now, because of injectivity of the exponential and the square over positive functions, it follows that the hardness of evaluating $\exp(-f_1(x_1)^2)$ is the same as the hardness of evaluating $f_1(x_1)$.
        So, if $k$ is easy to evaluate, then $f_1$ is also easy to evaluate, which contradicts the hypothesis.
        With this, we have that efficiently evaluating $k$ requires a quantum computer, thus completing the proof.
    \end{proof}

    \begin{customproposition}{\ref{prop:projected}} [Projected quantum kernel as an efficient EQK]
        The projected quantum kernel $k^\text{PQ}$ fulfills the assumptions of Proposition~\ref{prop:composition}, so it can be efficiently approximated as an EQK with the number of dimensions required $D$ fulfilling
    \begin{align}
        D\in \tcalO\left(\frac{\poly(d)}{\varepsilon^2}\right).
    \end{align}
    \end{customproposition}
    \begin{proof}
        The Gaussian RBF kernel is the same, up to rewiring the constant parameter.
        And then the function $f$ is the map from the input data $x$ to each entry of each reduced density matrix of the embedded quantum state $\rho(x)$.
        Density matrix elements are bounded, they have norm at most $1$, so we can injectively re-scale them to fit the $[0,1]$ range without altering the hardness of approximation.
        
        Formally, we let $f$ be characterized by three indices $k,m,n$ and we see the correspondence with the entries of the reduced 
        quantum states
        \begin{align}
            \sum_k\lVert\rho_k(x) - \rho_k(x')\rVert^2_F &= \sum_k \sum_{m,n=0}^{d_k-1} \left(\rho_{k,m,n}(x) - \rho_{k,m,n}(x')\right)^2 \\
            &= \sum_{k,m,n} \left(f_{k,m,n}(x) - f_{k,m,n}(x')\right)^2 \\
            &= \lVert f(x) - f(x')\rVert_2^2,
        \end{align}
        taking $f_{k,m,n}(x)$ to be the $(m,n)$ entry of the $k^\text{th}$ reduced density matrix of $\rho(x)$.
        In particular, we have that the dimension of the target space is $g_1(d) = \sum_{k=1}^Kd_k$, where $K$ is the number of subsystems considered, and $d_k$ is the local dimension of each of the subsystems.
        In the case where all subsystems have the same local dimension $d_k=d'$, we have $g_1(d) = Kd'$.
        Usually one takes $K=n$ to be the number of qubits, and $d_k=2$ corresponding to a single qubit.
        This, together with the assumption $n\in\calO(\poly(d))$ renders $g_1(d)\in\calO(\poly(d))$, as required by the proposition.
    \end{proof}


\section{Towards quantum kernels from parametrized quantum circuits}\label{a:variational}

    Embedding quantum kernels offer a straightforward recipe for constructing quantum kernels.
    Given a quantum feature map $x\mapsto\rho(x)$, one needs only take the Hilbert-Schmidt inner product $\kappa_\rho(x,x') = \Tr\{\rho(x)\rho(x')\}$.
    This approach holds irrespective of whether $\rho(x)$ is prepared with a \emph{parametrized quantum circuit} 
    or with some other, more complex state preparation protocol.
    Yet, often, $\rho(x)$ will correspond to a circuit with a few parametrized single-qubit gates interspersed with some fixed two-qubit entangling gate, as in the case of the IQP encoding used in Ref.~\cite{havlicek2019supervised}.

    To put things in perspective, one would construct a learning model from an \emph{encoding first} PQC as a three-step process:
    \begin{enumerate}
        \item Encode the classical data $x\mapsto\rho(x) = U(x)\lvert0\rangle\!\langle0\rvert U\dagg(x)$.
        \item Apply a trainable information-processing unitary $\rho(x)\mapsto \trho(x;\vartheta)= V(\vartheta)\rho(x)V\dagg(\vartheta)$.
        \item Measure the expectation value of a fixed observable $\trho(x;\vartheta)\mapsto f_\vartheta(x) = \langle \calM\rangle_{\trho(x;\vartheta)} = \langle0\rvert U\dagg(x)V\dagg(\vartheta)\calM V(\vartheta)U(x)\lvert0\rangle$.
    \end{enumerate}
    The corresponding EQK would just use the first step.
    
    Conversely, in order to construct a learning model from a \emph{data re-uploading} PQC, one would interleave encoding gates $U(x)$ with trainable gates $V(\vartheta)$ without a physical separation between the two:
    \begin{align}
        f(x) &= \langle 0\rvert V\dagg(\vartheta_0)U\dagg(x)V\dagg(\vartheta_1)U\dagg(x)\cdots U\dagg(x)V\dagg(\vartheta_L)\calM V(\vartheta_L)U(x)\cdots V(\vartheta_1)U(x)V(\vartheta_0)\lvert0\rangle.
    \end{align}
    By now we know that data re-uploading circuits can be rewritten as encoding-first circuits~\cite{jerbi2023beyond}, but still it might be intuitive to think of them as a separate entity.

    Re-uploading circuits can give rise to EQKs at least in two different ways.
    As has been proposed in Ref.~\cite{jerbi2023beyond}: translating the re-uploading circuit as an encoding-first circuit and then building the canonical EQK.
    As proposed in Ref.~\cite{hubregtsen2021training}, by considering the parametrized quantum feature map
    \begin{align}
        \lvert\phi(x,\vartheta)\rangle &= V(\vartheta_L)U(x)\cdots V(\vartheta_1)U(x)V(\vartheta_0)\lvert0\rangle,\\
        x\mapsto \rho(x,\vartheta) &= \lvert\phi(x,\vartheta)\rangle\!\langle\phi(x,\vartheta)\rvert ,\\
        \kappa_\vartheta(x,x') &= \Tr\{\rho(x,\vartheta)\rho(x',\vartheta)\}.
    \end{align}
    The result is a parametrized EQK $\kappa_\vartheta$.
    While it is not uncommon for kernels to have a tunable parameter (e.g., the Gaussian kernel has $\sigma$), it is rare to allow for more than a few.
    In Ref.~\cite{hubregtsen2021training} an algorithm for finding optimal values for the parameters $\vartheta$ has been proposed, but this approach is vulnerable to the problems of non-convex optimization we find in gradient-based QML approaches.

    There should be at least one other approach to build kernels from PQCs.
    In particular, nothing prevents us from encoding two different input variables into a re-uploading PQC
    \begin{align}
        f(x,x') &= \Tr\{\rho(x,x')\calM\}.
    \end{align}
    We make the $\vartheta$ dependence implicit for clarity, but in principle $\rho(x,x')$ could also be trainable.
    Notice traditional EQKs are a special case of this form, where for instance
    \begin{align}
        \rho(x,x') &= \rho'(x)\otimes\rho'(x'),
    \end{align}
    and $\calM$ would be the observable used in the $\SWAP$ test, resulting in $f(x,x')=\Tr\{\rho'(x)\rho'(x')\}$.
    Yet, there should exist other PQCs producing a parametrized state $\rho(x,x',\vartheta)$ that does not admit the separation into one part only dependent on $x$, and one part only dependent on $x'$.
    Indeed, for us to build a kernel like this it should suffice to take a general parametrization of a re-uploading circuit, and then find the parameters that ensure exchange symmetry and PSD, as in Definition~\ref{def:kernel}.

    As mentioned above, given a function described as a black-box, it is hard to establish whether the function fulfills the PSD property.
    Accordingly, except for the few recipes mentioned in this appendix, given a function described as a re-uploading PQC, it is hard to establish whether the function fulfills the PSD property.
    In the restricted case of shift-invariant functions, we can use the Fourier description of quantum circuits from Ref.~\cite{schuld2021fourier}.

    Recall that from Ref.~\cite{schuld2021fourier}, we know that PQCs give rise to real-valued trigonometric polynomials\footnote{
    Technically they give rise to truncated Fourier series, of which trigonometric polynomials are a special case.
    The two defining properties of trigonometric polynomials are that $(1)$ the sum has a finite number of terms, and $(2)$ all frequencies are integer multiples of a fixed base frequency.
    Thus, for a truncated Fourier series to be a trigonometric polynomial, there cannot be any pairwise incommensurate frequencies.
    Further, if input data is encoded with Pauli gates only, we are guaranteed all the frequencies are integer-valued.
    Since Pauli-gate encoding is the most common one, we abuse nomenclature slightly and call these functions trigonometric polynomials, even though in all generality that could be not the case.
    } $f(x)$ as
    \begin{align}
        f(x) &= \sum_{\omega\in\Omega} a_\omega\cos\langle\omega,x\rangle + b_\omega\sin\langle\omega,x\rangle.
    \end{align}
    The sum is over the so-called \emph{frequency spectrum} $\Omega\subset\bbN^d$, where $d$ is also the dimension of the input space $x\in\bbR^d$.
    The \emph{Fourier coefficients} $(a_\omega)_{\omega\in\Omega}$ and $(b_\omega)_{\omega\in\Omega}$, with $a_\omega,b_\omega\in\bbR$, can be separated into sine and cosine terms because $f$ is real-valued.
    We take the function to be $2\pi$-periodic and have base frequency $\omega_0=1$, without loss of generality.
    
    In combination with Bochner's theorem, we can establish whether a shift-invariant PQC gives rise to a PSD function by looking at the Fourier coefficients of the output function.
    Remember for shift invariant functions we introduce the variable $\Delta\coloneqq x-x'$, so that we have $f(x,x')=f(\Delta)$, abusing notation.
    Also, notice that for shift-invariant functions being symmetric under exchange is equivalent with being even $f(x-x')=f(x,x')=f(x',x)=f(x'-x)$.
    
    \begin{theorem}[Shift-invariant kernels from trigonometric polynomials]\label{thm:trigpolyPSD}
        Let $f(\Delta)$ be a real-valued trigonometric polynomial.
        Then, $f$ is even $f(\Delta)=f(-\Delta)$ if and only if all sine
         coefficients are zero, $b_\omega=0$ for all $\omega\in\Omega$.
        If $f$ is even, then $f$ is PSD if and only if all cosine coefficients 
        are non-negative, $a_\omega\geq0$ for all $\omega\in\Omega$.
    \end{theorem}
    \begin{proof}
        The first statement is a definitory trait of even functions.
        Every square-integrable function (also trigonometric polynomials) 
        admits a decomposition into even (cosine) part and odd (sine) part.
        For the function to be even, its odd part needs to vanish.

        The second statement follows from Bochner's theorem: a continuous, even, shift-invariant function is PSD if and only if its inverse FT is a non-negative measure
        \begin{align}
            p(\omega) &= \frac{1}{2\pi}\int_{[0,2\pi)^d} f(\Delta) e^{-i\langle\omega,\Delta\rangle}\mathrm{d}\Delta.
        \end{align}
        Now we expand the integral into the sum of $f$, using the facts that the sine terms are $0$, and that the cosine functions form an orthonormal basis with respect to the $L_2$ norm $\int_{[0,2\pi)^d} \cos\langle\omega,\Delta\rangle \cos\langle\omega',\Delta\rangle\mathrm{d}\Delta=\delta_{\omega,\omega'}$, $\int_{[0,2\pi)^d} \cos\langle\omega,\Delta\rangle \sin\langle\omega',\Delta\rangle\mathrm{d}\Delta=0$,
        to get
        \begin{align}
            \int_{[0,2\pi)^d} f(\Delta) e^{-i\langle\omega,\Delta\rangle}\mathrm{d}\Delta &= \int_{0,2\pi)^d} \sum_{\omega'\in\Omega} a_{\omega'}\cos\langle\omega',\Delta\rangle e^{-i\langle\omega,\Delta\rangle} \mathrm{d}\Delta \\
            &= \sum_{\omega'\in\Omega}a_{\omega'} \left(\int_{[0.2\pi)^d} \cos\langle\omega',\Delta\rangle\cos\langle\omega,\Delta\rangle \mathrm{d}\Delta -i \int_{[0,2\pi)^d}\cos\langle\omega',\Delta\rangle\sin\langle\omega,\Delta\rangle\mathrm{d}\Delta\right) \\
            &= \sum_{\omega'\in\Omega}a_{\omega'} \delta_{\omega,\omega'}.
        \end{align}
        It follows that $p(\omega)$ is a discrete function with peaks on the frequencies $\omega\in\Omega$ with height proportional to $a_\omega$.
        So, for $p(\omega)$ to be non-negative, each of the coefficients $a_\omega$ must also be non-negative, thus completing the proof.
    \end{proof}

    Using Theorem~\ref{thm:trigpolyPSD}, one could take a PQC $f$ with $\Delta$ as input and try to look for parameter specifications that made the Fourier coefficients fulfill the restrictions: sine coefficients are zero, cosine coefficients are non-negative.
    Yet, from the lens of efficient EQK-approximation, this approach is very likely doomed to fail.
    Given a PQC that fulfilled the hypotheses, with frequency spectrum $\Omega_f$, in order for us to be able to evaluate it efficiently, the circuit could have at most polynomially many encoding gates $U_1(\Delta),\ldots,U_L(\Delta)$, $L\in\calO(\poly(d))$, which we will assume to fulfill $U_l\dagg(\Delta)=U_l(-\Delta)$.
    Then, we could consider the  quantum feature map $x\mapsto\lvert\phi(x)\rangle$ defined as
    \begin{align}
        \lvert\phi(x)\rangle &:=  \left(\bigotimes_{l=1}^L U_l(x)\right) V_\phi \lvert0\rangle, 
    \end{align}
    where $V_\phi$ is a unitary acting on as many gates as all encoding gates combined.
    With this, we can construct an EQK with the same Fourier spectrum $\Omega_f$ as the function $f$, we have started from
    \begin{align}
        \kappa_\phi(x,x') &= \Tr\{\lvert\phi(x)\rangle\!\langle\phi(x)\rvert\phi(x')\rangle\!\langle\phi(x')\rvert\} \\
        &= \lvert\langle\phi(x)\rvert\phi(x')\rangle\rvert^2 \\
        &= \left\lvert\langle0\rvert V_\phi\dagg \left(\bigotimes_{l=1}^L U\dagg_l(x)\right) \left(\bigotimes_{l=1}^L U_l(x')\right) V_\phi \lvert0\rangle\right\rvert^2 \\
        &= \left\lvert\langle0\rvert V_\phi\dagg \left(\bigotimes_{l=1}^L U_l(x'-x)\right) V_\phi \lvert0\rangle\right\rvert^2 .
    \end{align}
    Now, this construction $\kappa_\phi$ has the same Fourier spectrum as $f$, since it uses the same encoding gates.
    Also, while $\kappa_\phi$ requires potentially more qubits $n$ than $f$, the scaling is at most polynomial in the number of input dimensions $n\in\calO(\poly(L))\in\calO(\poly(d))$.
    Finally, the construction proposed retains flexibility in the choice of unitary gates $V_\phi$.
    Different unitaries result in different Fourier coefficients, but by construction it is ensured that $\kappa_\phi$ is always an even, PSD function.
    Ref.~\cite{gf2022demo} contains some visuals on the construction.
    In all, $K:=\{\kappa_\phi \,|\, V_\phi\in\operatorname{SU}(2^n)\}$ constitutes a family of EQKs that covers part of the space of kernel functions achievable as trigonometric polynomials with spectrum $\Omega_f$.
    That means, for $f$ to be an efficient non-EQK, it not only needs to fulfill the restrictions on the Fourier coefficients, but it also needs to not fall inside of $K$.
    We do not claim this is impossible, but it would be quite surprising if that were not the case.

\end{document}